\documentclass[11pt]{article}

\usepackage[usenames,dvipsnames]{xcolor}
\definecolor{Gred}{RGB}{219, 50, 54}
\definecolor{Ggreen}{RGB}{60, 186, 84}
\definecolor{Gblue}{RGB}{72, 133, 237}
\definecolor{Gyellow}{RGB}{247, 178, 16}
\definecolor{ToCgreen}{RGB}{0, 128, 0}
\definecolor{myGold}{RGB}{231,141,20}
\definecolor{myBlue}{rgb}{0.19,0.41,.65}
\definecolor{myPurple}{RGB}{175,0,124}

\usepackage{hyperref}
\hypersetup{
			colorlinks=true,
	citecolor=ToCgreen,
	linkcolor=Sepia,
	filecolor=Gred,
	urlcolor=Gred
	}

\title{Learning Polynomials of Few Relevant Dimensions}

\author{Sitan Chen\thanks{This work was supported in part by a Paul and Daisy Soros Fellowship, NSF CAREER Award CCF-1453261, and NSF Large CCF-1565235.} \\
\texttt{sitanc@mit.edu}\\
MIT 
\and Raghu Meka \\
\texttt{raghum@cs.ucla.edu} \\
UCLA
}
\usepackage{preamble}
\usepackage[artemisia]{textgreek}
\newcommand{\h}[2]{h^{#1,#2}}

\newcommand{\pcaerror}{\varepsilon}
\newcommand{\bargrad}[2]{\overline{\nabla} {#1}(#2)}
\newcommand{\gradvec}[2]{\nabla^{\mathsf{vec}} {#1}(#2)}
\newcommand{\gradcoef}[2]{\nabla^{\mathsf{coef}} {#1}(#2)}
\newcommand{\bargradvec}[2]{\overline{\nabla}^{\mathsf{vec}} {#1}(#2)}

\newcommand{\St}{\text{St}}
\newcommand{\G}[2]{\text{G}(#1,#2)}

\newcommand{\procr}[1]{d_P(#1)}
\newcommand{\chord}[1]{d_C(#1)}
\newcommand{\polydiff}{\text{\textdelta}}
\newcommand{\condnumber}{\nu_{\mathsf{cond}}}
\newcommand{\D}[2]{\mathrm{D}_{#1}\, #2}
\newcommand{\dpm}{\{\pm 1\}}
\newcommand{\mnote}[1]{{\color{red}[Raghu: #1]}}
\newcommand{\iprod}[2]{\langle #1,#2\rangle}
\newcommand{\Iprod}[2]{\left\langle #1,#2\right\rangle}
\newcommand{\trigE}{\calE}
\newcommand{\residual}{\mathfrak{R}}
\newcommand{\DeltaV}{\Delta_{\mathsf{vec}}}
\newcommand{\tildeDeltaV}{\tilde{\Delta}_{\mathsf{vec}}}
\newcommand{\Deltac}{\Delta_{\mathsf{coef}}}

\newcommand{\domvec}{X}
\newcommand{\domcoef}{Y}
\newcommand{\Evecone}{E_1}
\newcommand{\Evectwo}{E_2}
\newcommand{\Ecoef}{E}
\newcommand{\alphasmooth}{\upsilon^{\mathsf{sm}}}
\newcommand{\alphacurve}{\upsilon^{\mathsf{cu}}}
\newcommand{\etac}{\eta_{\mathsf{coef}}}
\newcommand{\etav}{\eta_{\mathsf{vec}}}

\newcommand{\overhat}{\widehat}
\newcommand{\hatnab}{\widehat{\nabla}}
\newcommand{\hath}{\widehat{h}}

\newtheorem{question}{Question}

\usepackage{constants}
\newconstantfamily{c}{symbol=c}
\newconstantfamily{C}{symbol=C}

\newcommand{\ignore}[1]{}

\usepackage[algo2e,ruled,linesnumbered,vlined]{algorithm2e}

\usepackage{etoolbox}
\appto\appendix{\addtocontents{toc}{\protect\setcounter{tocdepth}{1}}}

\appto\listoffigures{\addtocontents{lof}{\protect\setcounter{tocdepth}{1}}}
\appto\listoftables{\addtocontents{lot}{\protect\setcounter{tocdepth}{1}}}

\begin{document}

\maketitle
%!TEX root = ./poly_main.tex

\begin{abstract}
Polynomial regression is a basic primitive in learning and statistics. In its most basic form the goal is to fit a degree $d$ polynomial  to a response variable $y$ in terms of an $n$-dimensional input vector $x$. This is extremely well-studied with many applications and has sample and runtime complexity $\Theta(n^d)$. Can one achieve better runtime if the intrinsic dimension of the data is much smaller than the ambient dimension $n$? 

Concretely, we are given samples $(x,y)$ where $y$ is a degree at most $d$ polynomial in an unknown $r$-dimensional projection (the relevant dimensions) of $x$. This can be seen both as a generalization of phase retrieval and as a special case of learning multi-index models where the link function is an unknown low-degree polynomial. Note that without distributional assumptions, this is at least as hard as junta learning.

In this work we consider the important case where the covariates are Gaussian. We give an algorithm that learns the polynomial within accuracy $\epsilon$ with sample complexity that is roughly $N = O_{r,d}(n \log^2(1/\epsilon) (\log n)^d)$ and runtime $O_{r,d}(N n^2)$. Prior to our work, no such results were known even for the case of $r=1$.

We introduce a new \emph{filtered PCA} approach to get a warm start for the true subspace and use \emph{geodesic SGD} to boost to arbitrary accuracy; our techniques may be of independent interest, especially for problems dealing with subspace recovery or analyzing SGD on manifolds.
\end{abstract}

\newpage

\tableofcontents

\newpage

%!TEX root = ./poly_main.tex

\section{Introduction}

Consider the classical \emph{polynomial regression} problem in learning and statistics. In its most basic form, we receive samples of the form $(x,y)$ with $x \in \R^n$ coming from some distribution and $y$ is $P(x)$ for a degree at most $d$ polynomial in $x$. Our goal is to \emph{learn} the polynomial $P$. Here learning could either mean learning the coefficients of $P$ or even finding some other function that gets small prediction error (as in find $Q$ with $E[(Q(x) - P(x))^2] \ll Var(y)$). 

Polynomial regression of course is one of the most basic primitives in statistics and machine learning especially in the more general \emph{non-realizable} case. For example, it is crucial in many kernalization applications, and it gives the best known PAC learning algorithms for various central complexity classes such as constant-depth circuits \cite{linial1993constant}, intersection of halfspaces \cite{klivans2004learning}, DNFs \cite{klivans2004learning2}, convex sets \cite{klivans2008learning,vempala2010learning}, the last of which even exploits intrinsic dimension as we do but for a different problem.

The basic bound for polynomial regression is that one can achieve good error with sample complexity and run-time that are $O(n^d)$. This dependence is also necessary (the space of degree $d$ polynomials is of dimension $\approx n^d$) even when $y = P(x)$. But often, such high complexity either in run-time or sample requirements is not feasible for many applications. 

This begs the question: can we formulate natural and useful scenarios where one can beat $n^d$ complexity? One such example is the work of \cite{andoni2014learning} who study \emph{sparse polynomials} and achieve complexity that is $f(d) poly(n,s)$ where $s$ is sparsity (in a suitable basis). 

Motivated by the rich body of work on \emph{phase retrieval} (see, e.g., \cite{candes2013phaselift,candes2015phase,conca2015algebraic,netrapalli2013phase} and references therein), work on \emph{multi-index models} in learning (see Section~\ref{sec:multiindexmodels_refs} below) and the above broad question, we study the question of learning polynomials that depend on few relevant dimensions. We call such polynomials \emph{low-rank polynomials}:

\begin{definition}
    A degree $d$ polynomial $P:\R^n \to \R$ is of rank $r$ if there exists a degree $d$ polynomial $p:\R^r \to \R$ and vectors $u_1^*,\ldots,u_r^* \in \R^n$ such that $P(x) = p(\iprod{u_1^*}{x}, \iprod{u_2^*}{x}, \ldots, \iprod{u_r^*}{x})$.  We will refer to $p$ as the \emph{link polynomial} and $U^* \triangleq \Span(u_1^*,\ldots,u_r^*)$ as the \emph{hidden subspace}.
\end{definition}

In other words, even though the ambient dimension of the polynomial $P$ is $n$, its \emph{intrinsic dimension} is only $r$. If we knew the subspace spanned by $u_1^*,\ldots,u_r^*$, then we could learn $P$ with sample-complexity that does not depend on $n$ at all and run-time that is linear in $n$ (and not $n^d$). Here, there are many natural notions of learning $P$ one could consider. Arguably the two most important goals are 1) to recover the hidden subspace $U^*$ spanned by $u_1^*,\ldots,u_r^*$, and 2) to find a polynomial $q$ that is close to $P$? 

Concretely, we are given samples $(x,y)$ where $y = P(x)$ and $P$ is a low-rank polynomial. For most natural distributions $y$, one can show it is information-theoretically possible to learn $P$ with sample-complexity that is only $O_{d,r}(n)$. That is, the dependence on the ambient dimension is only linear. Can we achieve this goal \emph{efficiently}? Henceforth, by efficient we mean that the sample-complexity and run-time are at most some fixed polynomial in $n$ that is of the form $O(f(r,d) n^c)$ for universal constant $c$. 

As desirable as the above goal is, it might be too good to be true for general distributions. For example, if $x$ is uniform on the hypercube $\{1,-1\}^n$, then the above question can encode the problem of learning \emph{$k$-juntas}. There, we are given samples $(x,f(x))$ where $x \in_u \dpm^n$ and $f$ is a function of at most $k$ variables, and the goal is to recover the indices of the relevant variables. Despite much attention, the best algorithms run in time $n^{\Omega(k)}$, and achieving $f(k) poly(n)$ sample complexity is an outstanding challenge conjectured to be computationally hard \cite{mossel2003learning}. The connection to rank is that any $k$-junta is a polynomial of rank and degree at most $k$.

Nevertheless, it makes sense to ask the question for other natural distributions. The most basic question in this vein (as we will further motivate later) is the case when $x$ is Gaussian: 

\begin{question}
Given samples $(x,y=P(x))$ where $x \sim \N(0,\Id_n)$, and $P$ is an unknown degree-$d$, rank-$r$ polynomial, can one approximately recover the subspace defining $P$ efficiently? Can we efficiently approximate $P$? Further, what is the dependence on the error $\epsilon$?
\end{question}

\noindent Note that while we ask the question for isotropic Gaussian covariates, our guarantees immediately carry over to general Gaussians, because the space of low-rank polynomials is affine invariant. Before stating our results, we first briefly discuss different ways of looking at the above question. 

\paragraph{Learning Multi-Index Models}
While we motivated the above problem from the context of polynomial regression, an equally valid way to introduce it is from the perspective of learning \emph{multi-index models} in Gaussian space. 

Here, we are given samples from a distribution $(x,y)$ where $x \sim \N(0,\Id_n)$ and
$$y = g(\iprod{u_1^*}{x}, \iprod{u_2^*}{x}, \ldots,\iprod{u_r^*}{x}),$$
where $g:\R^r \rightarrow \R$ is some unknown \emph{link function} and $u_1^*, u_2^*,\ldots,u_r^*$ are unknown orthonormal vectors, and the goal is to learn the subspace $U^*$ spanned by $u_1^*,\ldots,u_r^*$. 

The main question we study is the case where the unknown link function $g$ is a low-degree polynomial.  Most relevant to the present work is the recent work of \cite{pmlr-v75-dudeja18a} which we discuss next. There is a tremendous amount of work on learning multi-index models, and we refer to \cite{pmlr-v75-dudeja18a} for a detailed overview of previous work. \cite{pmlr-v75-dudeja18a} address the case where $g$ is \emph{smooth} in a Lipschitz sense quantified by a parameter $R$. They show:
\begin{enumerate}
    \item For \emph{single-index models} (i.e. when $r=1$): an algorithm that takes $\tilde{O}(n^{O(R^2)}) + n/\epsilon^2)$ samples and computes a direction $u$ that is $\epsilon$-close to the hidden direction.
    \item For \emph{multi-index models}: an algorithm that takes  $\tilde{O}(n^{O(r R^2)}) + n/\epsilon^2)$ samples and computes a direction $u$ that has at least $1-\epsilon$ of its $\ell_2$-mass in the span of the unknown $u_1^*,u_2^*,\ldots, u_r^*$. 
\end{enumerate}

Firstly, note that while most works on learning multi-index models assume some sort of Lipschitz-smoothness of the link function, polynomials are a natural class of link functions that do not satisfy such smoothness. More importanty, unlike existing works on multi-index models, our main goal is to achieve near-linear sample complexity, run-time scaling with $n^c$ for $c$ independent of $r,d$, and polylogarithmic dependence on the error $\epsilon$.

\paragraph{Generalizing Phase Retrieval} Further impetus for the above problem comes from the vast literature on phase retrieval. Here, one is given samples of the form $(x, \iprod{w}{x}^2)$ where $x$ is typically Gaussian for most provable guarantees \cite{candes2013phaselift,candes2015phase,conca2015algebraic,netrapalli2013phase}, and the goal is to learn $w$. Besides being natural by itself, the problem is extremely important in practice: as is explained in the references above, in certain physical devices one only observes the \emph{amplitudes} of linear measurements (corresponding to $\iprod{w}{x}^2$) and not the phase. In this setting, the signal and the inputs are taken to be complex but the question is often studied over the reals as well.

Note that the low-rank polynomial in question here is rank $1$ and degree $2$; moreover the link polynomial $p(z) = z^2$ is even known \emph{a priori}. In this sense, the problem we consider in this work is a substantial generalization, the study of which could potentially lead to new insights for phase retrieval, especially over more general covariate distributions.

\paragraph{Connections to Tensor Decompositions}
Our work also broadly fits in the category of \emph{tensor decompositions}. A $k$-ary tensor in $n$-dimensions is a multi-dimensional array $T \in \R^{[n]^k}$. More relevant to the present work, one can also view a tensor $T$ as a multi-linear map from $T:(\R^n)^k \rightarrow \R$ as $T(x^1,x^2,\ldots,x^k) = \sum_{1 \leq i_1 \leq i_2 \leq \cdots \leq i_k \leq n} T[i_1,i_2,\ldots,i_k] x^1_{i_1} x^2_{i_2} \ldots x^k_{i_k}$. For tensors, the term ``rank'' has a different meaning: a rank $1$ tensor is a tensor of the form $v^1 \otimes\cdots\otimes  v^k$, and in general, the rank of a tensor $T$ is the least number of rank one tensors whose sum is $T$.

The basic problem in tensor decomposition is to find a low-rank decomposition of a given tensor. Tensor decomposition algorithms have received a lot of attention recently \cite{anandkumar2014analyzing,anandkumar2015learning,ge2015decomposing,hopkins2015tensor,hopkins2016fast,schramm2017fast,ma2016polynomial} with various works studying many different aspects. The connection to our polynomial learning problem comes from the fact that a degree $d$ polynomial can be viewed as a $d$-ary tensor. Moreover, if a polynomial has rank $r$, then the corresponding $d$-ary tensor has rank roughly $O(rd)$. 

However, our goals and setting are quite different from those studied in the literature. For one, we are not given access to the tensor directly but only implicitly in the form of evaluations of the symmetric multi-linear form of the tensor on random inputs. Secondly, the central goal for us is to exploit the implicit representation to run in time that is much less than the time to even store the corresponding $d$-ary tensor. As far as we can tell, existing methods for tensor decompositions do not have these properties, at least provably. It is an intriguing question to find further scenarios where one could find tensor decompositions with much better run-time, for instance for constant-rank tensors, when the tensor has a \emph{succinct implicit representation}.

\subsection{Main Result}

Our main result is that we can indeed efficiently learn low-rank polynomials in Gaussian space. To the best of our knowledge, no such results were known even for the rank-$1$ case. Before stating our result formally, we have to introduce a definition to deal with \emph{degeneracy} in the notion of low-rank.

To understand the issue, consider the example where the link polynomial $p(z_1,z_2) = z_1 + z_2$. Then, if we look at $P(x) = p(\iprod{w_1^*}{x}, \iprod{w_2^*}{x})$, even though the polynomial is represented as a rank two polynomial, it is really only of rank one and we cannot hope to recover the span of $w_1^*, w_2^*$ but only the span of $w_1^* + w_2^*$. The following is necessary to overcome such non-identifiability issues:

\begin{definition}(Informal; see Definition~\ref{defn:degenformal})\label{defn:degenintro}
    A polynomial $P$ is $\alpha$-non-degenerate rank $r$ if $P$ is of rank $r$ and for any $(r-1)$-dimensional subspace $H$, the conditional variance of $P(x)$ given the projection of $x$ onto $H$ is at least $\alpha\cdot Var(p)$.
\end{definition}

Intuitively, there should not be a $(r-1)$-dimensional space that captures all of the variance of $P$. We give an equivalent analytic definition in Section~\ref{sec:prelims}. Note that any rank-$1$ polynomial satisfies the condition with $\alpha=1$. 

\begin{theorem}\label{thm:main}
There exists a universal constant $c_0$ and for all $r,d,\alpha$, there exists $C_0(r,d,\alpha)$ such that the following holds. For all $\delta > 0$ and $\epsilon \in (0,1)$, there is an efficient algorithm that takes $N = C_0(r,d,\alpha) (\ln(n/\delta))^{c_0 d} \cdot n \log^2(1/\epsilon)$ samples $(x,P(x))$, where $x \sim \N(0,\Id_n)$ and $P$ is an unknown $\alpha$-non-degenerate rank $r$, degree-$d$ polynomial defined by hidden subspace $U^*$, and outputs
\begin{enumerate}
    \item Orthonormal $u_1,\ldots,u_r\in\S^{n-1}$ such that $d_P(\Span(u_1,\ldots,u_r), U^*) \leq \epsilon$
    \item Degree $d$, $r$-variate polynomial $g$ such that $\E[(y - g(\iprod{u_1}{x},\ldots,\iprod{u_r}{x})^2] \leq \epsilon\cdot Var(y)$. 
\end{enumerate}
The run-time of the algorithm is at most $O(r^{c_0 d} N n^2)$. 
\end{theorem}

This will follow from Theorem~\ref{thm:introwarmstart} and Theorem~\ref{thm:boost} later in the paper. Here, $d_P(U,U^*)$ denotes the \emph{Procrustes distance} which is one of the standard measures for quantifying distances between subspaces. See Definition~\ref{defn:procrustes} for the exact definition. 

Note that the run-time of the algorithm is essentially $O_{r,d}(n^3 (\log n)^{O(d)})$ --- a fixed polynomial in $n$ as desired\footnote{One can save a further factor of $n$ as the $n^3$ comes from computing the top $r$ eigenvectors of a matrix which can be done better, see e.g. \cite{allen2016lazysvd}. We do not belabor this issue here.}. The sample complexity is also essentially linear in the ambient dimension $n$ and poly-logarithmic in $1/\epsilon$. No such result was known even for the rank $1$ case.

\begin{remark}
A word about the constant $C_0(r,d,\alpha)$ in the theorem. Our proof involves a compactness argument and as a result does not give an explicit upper bound on this quantity. Bounding this comes down to an extremal problem for low-degree polynomials in $r$ variables. For instance for $r=1$, $C_0(1,d,1)$ is essentially the inverse of 
$$\sup_\tau \inf_{h}(\E[1(|p(g)| > \tau) (g^2-1)]),$$
where $g \sim \N(0,1)$ and the infimum is over degree $d$ polynomials of variance $1$. We believe that this quantity is at least $2^{-C d^2}$ (as achieved by a suitably scaled degree $d$ Chebyshev polynomial). In general, our arguments can potentially yield a bound of  $C(r,d,\alpha) \approx 2^{O(rd)^2)}/\alpha^{\Theta(1)}$.

Also, we study the noiseless case where $Y = P(X)$.  It is possible to modify the first part of our argument (Theorem~\ref{thm:introwarmstart}) to get a version tolerant to some noise in $Y$, but we do not focus on this here. In any case, one of our main technical emphases is on getting run-time and sample complexity scaling with $\poly(\log(1/\epsilon))$, which would not be possible in the presence of noise.
\end{remark}

\subsection{Related Work}

\paragraph{Filtering Data by Thresholding}

Our algorithm for obtaining a warm start (see Theorem~\ref{thm:introwarmstart}) relies on filtering the data via some form of thresholding. This general paradigm has been used in other, unrelated contexts like robustness, see \cite{shen2019iterative,shen2018learning,diakonikolas2019robust,li2018principled,diakonikolas2019sever,diakonikolas2017being} and the references therein, though typically the points which are \emph{smaller} than some threshold are removed, whereas our algorithm, \textsc{TrimmedPCA}, is an intriguing case where the opposite kind of filter is applied.

\paragraph{Riemannian Optimization}

It is beyond the scope of this paper to reliably survey the vast literature on Riemannian optimization methods, and we refer the reader to the standard references on the subject \cite{udriste1994convex,absil2009optimization} which mostly provide asymptotic convergence guarantees, as well as the thesis of Boumal \cite{boumal2014optimization} and the references therein. Some notable lines of work include optimization with respect to orthogonality constraints \cite{edelman1998geometry}, applications to low-rank matrix and tensor completion \cite{mishra2013low,vandereycken2013low,ishteva2011best,kressner2014low}, dictionary learning \cite{sun2016complete}, independent component analysis \cite{shen2009fast}, canonical correlation analysis \cite{liu2015maximization}, matrix equation solving \cite{vandereycken2010riemannian}, complexity theory and operator scaling \cite{allen2018operator}, subspace tracking \cite{balzano2010online,zhang2016global}, and building a theory of geodesically convex optimization \cite{zhang2016first,hosseini2015matrix,zhang2016riemannian}.

We remark that the update rule we use in our boosting algorithm is very similar to that of \cite{balzano2010online,zhang2016global}, as their and our work are based on geodesics on the Grassmannian manifold. That said, they solve a very different problem from ours, and the analysis is quite different.

\paragraph{Single/Multi-Index Models and Other Link Functions}
\label{sec:multiindexmodels_refs}

As mentioned above, the problem of learning low-rank polynomial is a special case of that of learning a multi-index model, for which there is also a large literature which we cannot hope to cover here. In addition to \cite{pmlr-v75-dudeja18a} other works include those based on a connection to Stein's lemma \cite{plan2017high,neykov2016agnostic,brillinger2012generalized,li1992principal,plan2016generalized,yang2017stein}, sliced inverse regression \cite{babichev2018slice} as introduced in \cite{li1991sliced}, and gradient-based estimators \cite{hristache2001direct,hristache2001structure,dalalyan2008new}. Other works consider specific link functions or families of link functions: \begin{itemize}
	\item $z\mapsto \text{sgn}(z)$, i.e. one-bit compressed sensing \cite{plan2013one,ai2012one,gopi2013one}.
	\item $z\mapsto |z|^2$, i.e. phase retrieval \cite{candes2013phaselift,candes2015phase,conca2015algebraic,netrapalli2013phase}.
	\item $z\mapsto F(z)$ where $F:\R^r\to\R$ is computable by a constant-layer neural network \cite{ge2017learning,bakshi2018learning,janzamin2015beating,ge2018learning,goel2016reliably,goel2017learning}.
	\item $z\mapsto \bone{\epsilon_i\cdot \text{sgn}(z_i) \ \forall \ i\in[r]}$ for signs $\epsilon\in\{\pm 1\}^r$, i.e. intersections of halfspaces \cite{vempala2010random,klivans2009baum,klivans2004learning,khot2008hardness,vempala2010learning,diakonikolas2018learning}.
	\item $z\mapsto F(z)$ for some function $F: \R^r\to\{0,1\}$, i.e. subspace juntas \cite{vempala2011structure,de2019your}.
\end{itemize}

That said, none of the above seem to imply the guarantees for learning low-rank polynomials that we want, namely a run-time that is a fixed polynomial in $n$ and poly-logarithmic in $1/\epsilon$.% dsample complexity with nearly-linear dependence on the ambient dimension and polylogarithmic in $1/\epsilon$ and run-

%!TEX root = ./poly_main.tex

\section{Outline of Algorithm and Analysis}

A natural first step is to try to adapt the various techniques from the phase retrieval literature or existing works on multi-index models to the problem. But this seems challenging even for rank $1$. For example, the phase retrieval problem corresponds to the polynomial $p(z) = z^2$, which is rather special (see below), and if we don't even know the polynomial, then there are further difficulties. The works on multi-index models such as \cite{pmlr-v75-dudeja18a} also seem to be difficult to apply off the shelf. For one, they require smoothness of the link function. While it may be possible to circumvent the strict smoothness condition, it seems hard to find useful notions where the smoothness would not grow with the degree, leading to inefficient algorithms.

We present a different line of attack, inspired by ideas of \cite{pmlr-v75-dudeja18a}, \cite{candes2015phase}, \cite{balzano2010online}. Let $P(x) = p(\iprod{u_1^*}{x}, \iprod{u_2^*}{x}, \ldots, \iprod{u_r^*}{x})$ be the unknown $\alpha$-non-degenerate rank $r$ polynomial. For the remainder of the paper, let $\calD$ denote the distribution $(x,y)$ where $x \sim \N(0,\Id_n)$ and $y = P(x)$. Let $U^* = \Span(u_1^*,\ldots,u_r^*)$ be the hidden subspace. Without loss of generality assume $\Var(y) = 1$.\footnote{We can do so as our algorithms only need a good lower and upper bound on the variance $y$ which can be obtained easily.}

Our approach has two modular steps:
\begin{enumerate}
    \item \textbf{Warm start}: Obtain a ``good'' approximation to the true subspace $U^*$ by a modified PCA. 
    \item \textbf{Boost accuracy}: Use the subspace computed above as a starting point to boost the accuracy by \emph{Riemannian stochastic gradient descent}. 
\end{enumerate}

We next explain the steps at a high-level. The methods to carry out each of the steps could potentially be useful elsewhere especially for problems dealing with subspace recovery. %We are not aware of other rigorous analysis that use Riemannian stochastic gradient descent to get linear convergence for a non-convex problem as we do.   

\subsection{Getting a Warm Start}
The first step is to find a good subspace $V$ of dimension $r$ that $\epsilon$-close to $U^*$ (i.e., $d_P(V,U^*) \leq \epsilon$) in $O_{r,d}(n/\epsilon^2)$ samples. Note that identifying the subspace $U^*$ is the best we can do as the individual directions are not uniquely identifiable.  

\paragraph{Rank-One Case:} To motivate the algorithm, let us first focus on the rank $1$ or single-index case. Here $P(x) = p(\iprod{u^*}{x})$ where $u^*\in\S^{n-1}$ and our goal is to find some $u\in\S^{n-1}$ close to $u^*$. 

To do so, we propose a modified PCA by estimating a matrix of the form $M^\phi \equiv E[\phi(y) x x^T] - E[\phi(y)] E[x x^T]$ where $\phi:\R \to \R$ is a suitable ``filtering'' function. The intuition behind looking at $M^\phi$ is that the matrix has kernel of dimension $n-1$ corresponding to directions orthogonal to $u^*$. Thus, the non-zero eigenvalue of $M^\phi$, if any,  could help us approximate or even identify $u^*$. 

But what should the function $\phi$ be? For example, for phase retrieval where $P(x) = \iprod{u^*}{x}^2$, taking $\phi(z) = z^2$ suffices. The key issue is that this choice of $\phi$ does not work for general link polynomials. For example, if the link polynomial $p$ is $p(z) = z^2 - 3$, the matrix $M^\phi$ for this particular choice of $\phi$ is identically zero.  

We propose overcoming this by instead applying a simple \emph{thresholding} filter for $\phi$. Specifically, for a parameter $\tau > 0$ to be chosen later, let
$$M^\tau \triangleq \E[1(|y| > \tau) (x x^T - I)].$$

We show that for all $d$ there exists $\tau \equiv \tau(d)$ that only depends on $d$ such that $M^\tau$ is a non-zero matrix. Note that this by itself is not enough for our purposes: if the least non-zero eigenvalue of $M^\tau$ were extremely small, then this would affect our sample complexity in estimating $M^\tau$. We show there exists $\tau$ such that $M^\tau$ has an eigenvalue with magnitude at least $\lambda_{d} > 0$ for some constant depending on $d$ only. As argued before, the corresponding eigenvector is $u^*$. The intuition behind the proof is that conditioning on $|y| > \tau$ makes $x$ more likely to be large in the relevant direction.

The above structural statement is enough to get a warm start for $u^*$ by looking at the empirical approximation of $M^\tau$: for $N$ samples, let 
$$\overhat{M}^\tau \triangleq\frac{1}{N}\sum_{i=1}^N 1(|y_i| > \tau) (x_i x_i^T - I).$$

We can now use standard matrix concentration inequalities to argue that for $N = O_{d}(n/\epsilon^2)$ samples, the top eigenvector $\hat{u}$ of $\overhat{M}^\tau$ satisfies $\|u^* - \hat{u}\| \leq \epsilon$. 

\paragraph{Higher-Rank Case:} Extending the above to higher ranks seems much more challenging. 

A natural attempt would be to look at a matrix $M^\tau$ as above for a suitable $\tau$. It is once again easy to argue that $M^\tau$ has $n-r$ vectors in its kernel corresponding to the vectors orthogonal to $U^*$. We would now like to say that for some suitable $\tau \equiv \tau(r,d)$, the top $r$ eigenvalues of $M^\tau$ are at least $\lambda_{r,d}$. If so, we can proceed as before to get an approximation to $U^*$ (the non-zero eigenvectors are in $U^*$). While we can currently show that there is at least one such eigenvalue, we do not know if the matrix $M^\tau$ has rank at least $r$ and it seems considerably more challenging to prove. The difficulty is that unlike the rank $1$ case, while conditioning on $|y| > \tau$ should intuitively bias $x$ to have large norm in the relevant directions, it is not clear if it does so in every relevant direction. 

Instead, we follow an iterative strategy where we identify one direction at a time in $U^*$. This is similar in spirit to the standard technique of computing the eigenvalues of a matrix by first computing the top eigenvector, projecting it out, and then iterating. 

Concretely, suppose we have identified orthonormal vectors $V = \{v_1,v_2,\ldots,v_\ell\}$ for $\ell < r$ that individually have most of their mass in $U^*$. Let $\Pi_{V^\perp}$ be the projection operator onto the space orthogonal to $v_1,\ldots,v_\ell$. Then, to compute the next direction we look at the top eigenvector of
$$M^{\ell,\tau} \triangleq \Pi_{V^\perp} \E[1(|y| > \tau) 1(|\iprod{v_i}{x}| \leq 1, \;\forall i \leq \ell) (x x^T - I)] \Pi_{V^\perp} .$$

As before, we argue that the top eigenvector of the above matrix will have most of its mass in $U^*$ and this gives us our next vector $v_{\ell+1}$. While the sequence of matrices we look at are a bit more complicated, standard random matrix concentration inequalities still allow us to identify the new directions with sample complexity $O_{r,d}(n)$. 

In summary, we get the following:
\begin{theorem}\label{thm:introwarmstart}
For all $r,d, \alpha$, there exists $C(r,d,\alpha)$ such that the following holds. For all $\delta > 0$ and $\epsilon \in (0,1)$, there is an efficient algorithm that takes $N = C(r,d,\alpha)n \log(1/\delta)/\epsilon^2$ samples $(x,P(x))$ for $x\sim\N(0,\Id_n)$ and unknown $P$ which is $\alpha$-non-degenerate of rank $r$, and outputs a subspace $U$ such that with probability at least $1-\delta$, $d_P(U,U^*) < \epsilon$. The algorithm runs in time $O(r (N n^2 + n^3))$. 
\end{theorem}

% \ignore{
% \begin{theorem}\label{thm:introwarmstart}
% For all $r,d, \alpha \in (0,1)$, there exists $\tau = \tau(r,d)$ and $c = c(r,d,\alpha)$ such that for any $\alpha$-non-degenerate rank $r$ polynomial $P$ of variance $1$, the matrix 

% $$M^\tau = E[1(|P(X)| > \tau) X X^T] - E[1(|P(X)| > \tau)] E[ X X^T],$$
% has $r$ singular values of value at least $c$. The remaining $n-r$ singular values are zero. 
% \end{theorem}

% \mnote{Can we actually show such a theorem?}}
%!TEX root = ./poly_main.tex

\subsection{Boosting via Geodesic-Based Riemannian Gradient Descent}

The results from the previous section give us a way to find a subspace $U$ that is $\epsilon$-close to the true subspace $U^*$ with sample complexity $O_{r,d}(n/\epsilon^2)$.

However, the dependence on $\epsilon$ above is problematic and quite far from what is achievable, e.g., for the special case of phase retrieval. There, results starting with work of \cite{candes2015phase} show that one can get \emph{exact} recovery of the unknown direction with sample complexity $\tilde{O}(n)$; in this case, while the sample complexity is $\tilde{O}(n)$, the \emph{run-time} to get within error $\epsilon$ scales with $\log(1/\epsilon))$. In a similar vein, the result of \cite{netrapalli2013phase} shows that one can find a vector $w$ that is $\epsilon$-close to the unknown vector with sample-complexity $\tilde{O}(n\log(1/\epsilon))$ and a similar run-time. We address this issue next and give an algorithm that achieves error $\epsilon$ with sample-complexity $\tilde{O}_{r,d}(n\log^2(1/\epsilon))$ and run-time $\tilde{O}_{r,d}(n^2\log^2(1/\epsilon))$. In the proceeding discussion, we will use some basic terminology from differential geometry in motivating our algorithm, though we emphasize that the algorithm itself is stated solely in terms of matrices, and its proof only involves, e.g., linear algebra and concentration of measure.

First, it is important to understand what fundamentally changes when going from phase retrieval to the more general problem of learning an unknown, low-rank polynomial. At a high level, there are two closely related challenges: \begin{enumerate}
	\item \textbf{Unknown $r$-variate polynomial}: Unlike in phase retrieval where we know that the link polynomial is $h(z) = z^2$ \emph{a priori}, in our setting we are not given the coefficients of the true polynomial. The natural workaround is to simply run gradient descent jointly on the space of coefficients and the space of $n\times r$ matrices $V$. As we will see in Section~\ref{sec:rank1example} next, this poses novel difficulties even in the rank-1 case.
	\item \textbf{Identifiability only up to rotation}: A more fundamental issue is the number of inherent symmetries in the problem, which explodes as $r$ increases. Indeed, there is an infinitely large orbit of parameters $\Theta^* = (\vec{c}^*,V^*)$ which give rise to the same underlying low-rank polynomial $P$, parametrized by the group of all rotations of the underlying subspace. Whereas for $r = 1$ it is easy to quotient out most of the symmetries by simply running projected gradient descent on the unit sphere, as we will see in Section~\ref{sec:whichspace}, to define the right quotient geometry we will need to run gradient descent on a manifold for which the corresponding optimization landscape is far less straightforward. In addition, as we will see in Section~\ref{sec:tracking}, these symmetries also pose problems for defining and analyzing a suitable progress measure.
\end{enumerate}

In light of 2), it will be good to give a name to the set of parameters $\Theta^* = (\vec{c}^*,V^*)$ which correspond to the underlying low-rank polynomial.

\begin{definition}
	For a collection of coefficients $\vec{c}^*$ of a degree-$d$ $r$-variate polynomial, and a column-orthonormal matrix $V^*\in\R^{n\times r}$, we say that the parameters $\Theta^* = (\vec{c}^*,V^*)$ are a \emph{realization} of $\calD$ if the polynomial $p_*(z)\triangleq \sum_I c^*_I\phi_I(z)$ satisfies $P(x) = p_*({V^*}^{\top}x)$ for all $x\in\R^n$, where $\{\phi_I\}$ are the (normalized) tensor-product Hermite polynomials of degree at most $d$ over $r$ variables (see Section~\ref{subsec:hermite}).
\end{definition}

\subsubsection{Not Knowing the Polynomial: A Toy Calculation}
\label{sec:rank1example}

The issue of not knowing $p$ manifests even in the $r = 1$ case. Below, we examine at a high level where the calculations for analyzing gradient descent for phase retrieval break down for us.

Let us try to imitate the approach of \cite{candes2015phase}. Let $\Theta^* = (\vec{c}^*,v^*)$ be one of the two possible realizations of $\calD$ for which $v^*\in\S^{n-1}$, and suppose we already have a warm start of $\Theta = (\vec{c},v)$, where the coefficients $\vec{c}$ and $\vec{c}^*$ define the univariate degree-$d$ polynomials $p(z) \triangleq \sum^d_{i=1}c_i \phi_i(z)$ and $p_*(z) \triangleq \sum^d_{i=1}c^*_i\phi_i(z)$ respectively. Given samples $(x^1,y^1),...,(x^N,y^N)\sim\calD$, a natural approach would be to analyze vanilla gradient descent over $\R^{d+1}\times \R^n$ for the empirical risk \begin{equation}L(\Theta) \triangleq \frac{1}{N}\sum^N_{i=1}(F_{x^i}(\Theta) - y_i)^2 \ \ \ \ \ \ \text{for} \ \ \ \ \ \ F_x(\Theta) \triangleq p(V^{\top}x).\end{equation} To show that this converges linearly from a warm start, the first thing to show would be that the negative gradient at $\Theta$ is correlated with the direction in which we would like to move, a property that sometimes goes under the name \emph{local curvature}. Noting that $\frac{1}{2}\grad{L}{\Theta} = \frac{1}{N}\sum^N_{i=1}(F_{x^i}(\Theta) - F_{x^i}(\Theta^*))\cdot \grad{F}{x^i}(\Theta)$, using the fact that we initialize at a warm start in order to linearly approximate $F_x(\Theta) - F_x(\Theta^*)$ by $\grad{F_x}{\Theta^*}\cdot \langle \Theta - \Theta^*\rangle$, and explicitly computing the gradient of $F_x$ (see Proposition~\ref{prop:derivs}), one can check that \begin{align}\left\langle \frac{1}{2}\grad{L}{\Theta}, \Theta - \Theta^*\right\rangle &\approx \frac{1}{N}\sum^N_{i=1}\iprod{\grad{F_{x^i}}{\Theta^*}}{\Theta - \Theta^*}^2 \\
&= \frac{1}{N}\sum^N_{i=1}\left[\langle v - v^*, x^i\rangle \cdot p'_*(\langle v^*,x^i\rangle) + (p - p_*)(\langle v^*,x^i\rangle)\right]^2.\end{align} The expectation of this quantity is \begin{equation}
	\mu\triangleq \E_g\left[\left(\langle v - v^*, g\rangle \cdot p'_*(\langle v^*,g\rangle) + (p - p_*)(\langle v^*,g\rangle)\right)^2\right] 
\end{equation} 
Write $v - v^* = \alpha\cdot v^* + \beta\cdot v^{\perp}$ for
$v^{\perp}\in\S^{n-1}$ orthogonal to $v^*$, where $\alpha = \langle v,v^*\rangle - 1\approx -\norm{v - v^*}^2_2$.
By some elementary calculations which we omit here, one can show that \begin{equation}
	\mu = \beta^2 \cdot \E[p'_*(x)^2] + \sum^d_{\ell = 0}\left((\alpha\ell+1) \cdot c_{\ell} + a\sqrt{(\ell + 1)(\ell+2)}\cdot c_{\ell+2} - c^*_{\ell}\right)^2.\label{eq:basic_curvature}
\end{equation}
 In the case of phase retrieval, $p(z) = p_*(z) = z^2 = \sqrt{2} + \sqrt{2}\cdot\phi_2(z)$, so $\vec{c} = \vec{c}^* = (\sqrt{2},0,\sqrt{2})$ and we simply get that \begin{equation}
	\mu = 12\alpha^2 + 4\beta^2 \ge 4\norm{v - v^*}^2_2.
\end{equation} In other words, the correlation between the negative gradient and the residual direction $v^* - v$ in which we would like to go is positive and scales with the squared norm of the residual. This simple calculation lies at the heart of the proof that vanilla gradient descent converges linearly to $v^*$ from a warm start for phrase retrieval.

More generally, if $\vec{c}^* = \vec{c}$, then the quantity in \eqref{eq:basic_curvature} will enjoy this positive scaling with $\norm{v^* - v}_2^2$, and one can also show linear convergence of vanilla gradient descent. But it is apparent that when $\vec{c}^* \neq \vec{c}$, $\mu$ can be arbitrarily close to zero, e.g. by taking $\beta$ to be much smaller than $\alpha$. So when $\vec{c}^*\neq \vec{c}$, we may get stuck at spurious infinitesimal-curvature points of the optimization landscape and fail to make sufficient progress in a single step.

The basic underlying issue is simply that vanilla gradient steps can move us in unhelpful directions, e.g. we might end up moving mostly in the direction of $v$ when we should be moving in directions orthogonal to $v$. And whereas this evidently does not pose an issue when $\vec{c} = \vec{c}^*$, which corresponds to the case where we know the underlying polynomial and only need to run gradient descent to learn the hidden direction, in the case where $\vec{c} \neq \vec{c}^*$ and we must run gradient descent jointly on $v$ and $\vec{c}$, the usual analysis of vanilla gradient descent fails.

\subsubsection{Non-Identifiability: Which Space to Run SGD In?}
\label{sec:whichspace}

The workaround for the issue posed in Section~\ref{sec:rank1example} is clear at least in the rank-1 case: to avoid moving in the wasteful directions which are orthogonal to the current iterate $v$, simply compute the vanilla gradient and project to the orthogonal complement of $v$. We would also like to ensure that our iterates themselves are unit vectors like $v^*$, so the following two-step update rule would suffice: 1) walk against the projected gradient and then 2) project back to $\S^{n-1}$. In fact, one can show that this algorithm actually achieves linear convergence for learning arbitrary unknown rank-1 polynomials.

It turns out there is a principled way to extend this approach to higher rank. Indeed, the abovementioned projected gradient scheme is nothing more than (retraction-based) gradient descent on the Riemannian manifold $\S^{n-1}$: the orthogonal complement of $v$ is precisely the tangent space of $\S^{n-1}$ at $v$, and the projection back to $\S^{n-1}$ is a special instance of a \emph{retraction}, roughly speaking a continuous mapping from the tangent spaces of a manifold back onto the manifold itself. We do not attempt to define these notions formally, referring the reader to, e.g. \cite{absil2009optimization}.

The rank-$r$ analogue of $\S^{n-1}$ is the Grassmannian $\G{n}{r}$ of $r$-dimensional subspaces of $\R^n$. However, while various retraction operations, e.g. via QR decomposition, can be constructed, retraction-based Riemannian optimization is somewhat more difficult to analyze in our setting. Instead, we appeal to an alternative formulation of Riemannian gradient descent via geodesics.

Roughly, geodesics are acceleration-free curves on a manifold determined solely by their initial position on the manifold, initial velocity, and length. Gradient descent on a Riemannian manifold $\calM$ via geodesics is then very simple to formulate: at an iterate $p\in\calM$, 1) compute the gradient $\nabla$ after projecting to the tangent space at $p$, 2) walk along the geodesic that starts at $p$ and has initial velocity $\nabla$ and length $\eta$, where $\eta$ is the learning rate.

We now see what this would yield in our setting. Let $\Theta = (\vec{c},V)$ be an iterate. For now, we will keep $\vec{c}$ fixed and describe how to update $V$, regarded as a column-orthonormal $n\times r$ matrix of basis vectors for the subspace $V$, by following the appropriate geodesic on $\G{n}{r}$. Given a single sample $(x,y)$, define the single-sample empirical risk $L^{\vec{c}}_x(V) = (F_x(\Theta) - y)^2$. Let $\grad{L^{\vec{c}}_x}{V} \in\R^{n\times r}$ be the vanilla gradient, where $L^{\vec{c}}_x(V) \triangleq L_x(\Theta)$. It turns out its projection to the tangent space at $V$ is simply $\nabla\triangleq \Pi^{\perp}_V\cdot\grad{L^{\vec{c}}_x}{V} \in\R^{n\times r}$, where $\Pi^{\perp}_V$ denotes projection to the orthogonal complement of $V$ (note that this is a natural generalization of the tangent spaces for $\S^{n-1}$).

The geodesic $\Gamma$ with initial point $V$ and velocity $\nabla$, and length $\eta$ has a simple closed form in terms of the SVD of $\nabla$, which is made even simpler by the fact that in our setting, $\nabla$ turns out to be rank-1. We defer the details of the exact update, which can be computed in time $O(n)$, to Section~\ref{sec:boost_details}.

% Lastly, we remark that both retraction-based and geodesic-based Riemannian optimization have been studied extensively since the pioneering work of \cite{edelman1998geometry} \sitan{references}, and in particular the update rule in the GROUSE algorithm of \cite{balzano2010online,zhang2016global} is very similar to ours, though tailored to a completely different application.

\subsubsection{Tracking Progress in both {\mdseries \textbf{c}} and {\mdseries \emph{V}}}
\label{sec:tracking}

In the previous section we sketched our approach for updating our estimate $V$ for the subspace given an estimate $\vec{c}$ for the coefficients of the polynomial, but did not explain how to update $\vec{c}$. As $\vec{c}$ just lives in Euclidean space, we can simply update $\vec{c}$ to some $\vec{c}'$ via vanilla gradient descent on $L^V$, where $L^V(\vec{c}) \triangleq L(\Theta)$, and this is the approach we take.

To analyze such an approach, one would want to show that each step $(\vec{c},V)\mapsto (\vec{c}',V')$ contracts some suitably defined progress measure. Indeed, the natural progress measure one could try analyzing is \begin{equation}\inf_{(\vec{c}^*,V^*) \ \text{realizing} \ \calD}\norm{\vec{c} - \vec{c}^*}^2_2 + \norm{V - V^*}^2_F.\label{eq:hard_progress_measure}
\end{equation} The key difficulty here is that the minimizing realization $(\vec{c}^*,V^*)$ could change with each new iterate, and tracking how this changes is tricky as there is no clean non-variational proxy for \eqref{eq:hard_progress_measure}.

Our workaround is to have our boosting algorithm alternate between two phases. For an iterate $V\in\R^{n\times r}$, we run the following algorithm, \textsc{GeoSGD}, which alternates between two phases: 1) recomputing a good $\vec{c}$, and 2) updating $V$ using that $\vec{c}$. An informal specification of this algorithm is given in Algorithm~\ref{alg:geosgd_informal} below.

\begin{algorithm2e}\caption{\textsc{GeoSGD} (informal)}\label{alg:geosgd_informal}
 	\DontPrintSemicolon
 	\KwIn{Sample access to $\calD$, warm start $V^{(0)}\in\R^{n\times d}$, target error $\epsilon$, failure probability $\delta$}
 	\KwOut{Estimate $(\vec{c}^{(T)},V^{(T)})$ which is $\epsilon$-close to a realization of $\calD$}
 	\For{$0\le t<T$}{
 		Run \textsc{RealignPolynomial} using $V^{(t)}$. That is, draw samples and run vanilla gradient descent with respect to empirical risk $L^{V^{(t)}}$ over those samples to produce $\vec{c}^{(t)}$ which approximates the ``best'' choice of $\vec{c}$ given fixed $V^{(t)}$.\;
 		Run \textsc{SubspaceDescent} initialized to $V^{(t)}$ and using $\vec{c}^{(t)}$. That is, draw samples and, starting from $V^{(t)}$, run a small step of geodesic gradient descent with respect to empirical risk $L^{\vec{c}}_x$ for each of those samples $x$. Call the result $V^{(t+1)}$\;
 	}
 	Output $V^{(T)}.\;$
\end{algorithm2e}

We will defer an exact specification of \textsc{GeoSGD} and the subroutines \textsc{RealignPolynomial} and \textsc{SubspaceDescent} until Section~\ref{sec:boost_details}.

To analyze this scheme, rather than track progress in \eqref{eq:hard_progress_measure} we can simply track progress in $d_P(V,V^*) = \inf_{V^*}\norm{V - V^*}^2_F$, where $V^*$ ranges over $n\times r$ matrices whose columns form an orthonormal basis for the true subspace. This progress measure is, up to constants, simply the Procrustes distance between our current subspace $V$ and the true subspace $V^*$, and can be approximated by the \emph{chordal distance} which has a simple closed-form expression amenable to analysis.

Roughly, we will show the following: 

\begin{theorem}[Informal, see Theorem~\ref{thm:realign_guarantee}]\label{thm:basic_realign_guarantee}
	If $V$ is sufficiently close to the true subspace in Procustes distance, then running \textsc{RealignPolynomial} using $V$ will yield $\vec{c}$ such that for the realization $(\vec{c}^*,V^*)$ of $\calD$ where $d_P(V,V^*) = \norm{V - V^*}_F$, $\norm{\vec{c} - \vec{c}^*}_2 \approx \procr{V - V^*}$.
\end{theorem}
\begin{theorem}[Informal, see Theorem~\ref{thm:subspacedescent_guarantee}]\label{thm:basic_subspace_guarantee}
	If $V$ is sufficiently close to the true subspace in Procustes distance, If $V$ and $\vec{c}$ are such that $\norm{\vec{c} - \vec{c}^*}_2 \approx \procr{V - V^*}$for the realization $(\vec{c}^*,V^*)$ of $\calD$ where where $d_P(V,V^*) = \norm{V - V^*}_F$, then running \textsc{SubspaceDescent} initialized to $V$ and using $\vec{c}$ will yield $V'$ so that the progress measure $d_P(V,V^*)$ contracts by a factor of $1 - \tilde{O}_{r,d}(1/n)$.
\end{theorem}

Having defined the ``right'' gradient descent subroutines, the proofs of Theorems~\ref{thm:basic_realign_guarantee} and \ref{thm:basic_subspace_guarantee} will be based on showing the same kind of estimates alluded to in Section~\ref{sec:rank1example}. That is, for instance we must show that the steps in both subroutines have good correlation with the direction in which we want to go. Showing this holds with high probability will then entail exhibiting the appropriate second moment bounds. In the case of Theorem~\ref{thm:basic_realign_guarantee}, we can then invoke standard hypercontractivity-based tail bounds to show concentration. In the case of Theorem~\ref{thm:basic_subspace_guarantee}, concentration will be more delicate as each small step of \textsc{SubspaceDescent} will be a geodesic gradient step with respect to a \emph{single-sample} empirical risk $L^{\vec{c}}_x$. For the analysis to be doable, it is crucial that these risks be single-sample so that the geodesic steps are \emph{rank-one} updates. But then, to show concentration over a sequence of small geodesic steps, we must invoke non-standard martingale concentration inequalities, see Section~\ref{subsec:martingale_concentration}. Intuitively, if we take the sizes of these small steps to scale with $O(1/T)$, the corresponding martingale does not move away from its starting point by too much, and the sum of the martingale differences ends up behaving more or less like a sum of iid random variables (see the beginning of Section~\ref{sec:local_curve_all_iters}). We refer the reader to Sections~\ref{sec:realignpoly} and \ref{sec:subspacedescent} for the complete proofs of Theorems~\ref{thm:basic_realign_guarantee} and \ref{thm:basic_subspace_guarantee} respectively.

\paragraph{Roadmap}

In Section~\ref{sec:prelims} we introduce notation and miscellaneous technical facts that we will use in our proofs. In Section~\ref{sec:warmstart}, we give our algorithm \textsc{TrimmedPCA} for obtaining a warm start. In Section~\ref{sec:boost_details}, we give the formal specification for our boosting algorithm \textsc{GeoSGD}, and in Sections~\ref{sec:realignpoly} and \ref{sec:subspacedescent} we prove guarantees for its key subroutines. We complete the proof of correctness for \textsc{GeoSGD} in Section~\ref{sec:finish}. In Appendix~\ref{app:martingale} we give the martingale concentration inequalities we will need, and in Appendix~\ref{app:realign} and Appendix~\ref{app:subspace} we complete proofs deferred from the body of the paper.

%!TEX root = poly_main.tex

\section{Notations and Preliminaries}
\label{sec:prelims}

Throughout, $n$ will denote the ambient dimension, $r$ the rank of the polynomial, and $d$ the degree.

Given a vector space $U$, let $\Pi_U$ denote the orthogonal projection operator onto $U$, and let $U^\perp$ denote the orthogonal complement of $U$. We will often abuse notation and also use $U$ to refer to a set of column vectors $\{u_1,\ldots,u_\ell\}$, in which case $\Span(U)$ denotes the span of these vectors, $\Pi_U$ denotes $\Pi_{\Span(U)}$, and $\Pi_{U^\perp}$ denotes $\Pi_{\Span(U)^\perp}$. Given vector spaces $U\subset V$, $V\backslash U = V \cap U^\perp$ denotes the orthogonal complement of $U$ in $V$. Given $v\in\S^{d-1}$, we will use $\Pi_v \triangleq vv^{\top}$ and $\Pi^{\perp}_v\triangleq \Id - vv^{top}$ to denote projection to the span of $v$ and its orthogonal complement, respectively. More generally, given $V\in\R^{n\times r}$ whose columns are orthonormal, we will use $\Pi_V \triangleq VV^{\top}$ and $\Pi^{\perp}_V\triangleq \Id - VV^{\top}$ to denote projection to the span of the columns of $V$ and its orthogonal complement, respectively.

Given matrix $M\in\R^{m\times n}$, let $\norm{M}_F$ denote its Frobenius norm, and $\norm{M}_2$ its operator norm.

Let $\St^n_r$ denote the \emph{Stiefel manifold} of $n\times r$ matrices with orthonormal columns, and let $\G{n}{r}$ denote the \emph{Grassmannian} of $r$-dimension subspaces of $n$. $\G{n}{r}$ can be regarded as the quotient of $\St^n_r$ under the natural action of $O(r)$, that is, given any subspace $U\in \G{n}{r}$ and any $V\in\St^n_r$ whose columns form a basis for $U$, we can associate $U$ to the equivalence class $[V]\triangleq \{V\cdot O: O\in O(r)\}$.
   
For $r > 0$, let $\B^n_r \subset\R^n$ denote the Euclidean ball of radius $r$ centered at the origin. When the context is clear, we will suppress the superscript $n$.

For polynomial $p: \R^r\to\R$, define $\Var[p] = \E[(p - \E[p])^2]$. Given indices $\vec{j}\triangleq (j_1,...,j_{\ell})\in[r]^{\ell}$, and $z\in\R^r$ we will use the shorthand \begin{equation}\label{eq:deriv_notation}
	\D{\vec{j}}{p(z)} \triangleq \frac{\partial}{\partial z_{j_1}\cdots \partial z_{j_{\ell}}}p(z).
\end{equation} 
Similarly, for $F: \R^{n\times r}\to\R$, indices $\vec{i}\in[n]^{\ell}$ and $\vec{j}\in[r]^{\ell}$, and $V\in\R^{n\times r}$, we will use the shorthand \begin{equation}\label{eq:matrix_deriv_notation}
	\D{\vec{i},\vec{j}}{F(V)} \triangleq \frac{\partial}{\partial V_{i_1,j_1}\cdots \partial V_{i_{\ell},j_{\ell}}}F(V).
\end{equation}

\subsection{Non-degeneracy} Recall the notion of $\alpha$-non-degenrate rank $r$ polynomials introduced in Definition~\ref{defn:degenintro}. While that notion is intuitive, it is less amenable to analysis. It turns out that the notion is essentially equivalent (up to scaling $\alpha$ by $d$) to the following and we will use this going forward. 

\begin{definition}
    \label{defn:degenformal}
    A polynomial $h:\R^r \rightarrow \R$ is $\alpha$ non-degenerate if $ M = \E_{g\sim\N(0,\Id_r)}\left[\grad{h}{g}\grad{h}{g}^{\top}\right]$ satisifes $M \succeq \alpha \cdot \|M\|_2 \I_r$.
    
    We say a rank $r$ polynomial $P:\R^n \to \R$ is $\alpha$ non-degenerate if $P$ is non-degenerate in the $r$-dimensional space corresponding to the relevant directions. That is, there exist orthonormal vectors $u_1,\ldots,u_r$ such that $P(x) = h(\iprod{u_1}{x},\ldots,\iprod{u_r}{x})$ and $h$ is $\alpha$ non-degenerate.
\end{definition}

While it is not clear immediately from the definition, the notion above does not depend on the specific basis chosen. Henceforth, fix constant $\condnumber>0$. we will let $\calP^{\condnumber}_{n,r,d}$ denote the set of all $\condnumber$ non-degenerate rank $r$ polynomials $P$ of degree at most $d$ in $n$ variables that satisfy the normalization conditions  $\E_{X\sim\N(0,\Id_n)}[P(X)] = 0$ and $\E_{g\sim\N(0,\Id_r)}\left[\grad{h}{g}\grad{h}{g}^{\top}\right] \preceq \Id_n$. We write $\calP^{\condnumber}_{r,d}$ for $\calP^{\condnumber}_{r,r,d}$. 

Finally, we will use the following elementary property of non-degeneracy. 
\begin{fact}\label{fact:cor_id}
If $P \in \calP^{\condnumber}_{n,r,d}$, then $\condnumber/d \le \Var[P(X)] \le r$.
\end{fact}

\begin{proof}
It suffices to consider $n=r$. 
	For the upper bound, we have $\Var[P] \le \E_g\left[\norm{\grad{p_*}{g}}^2_2\right] \le r$ by taking traces in the definition of non-degeneracy and invoking Lemma~\ref{lem:normgrad_lb} below.

	For the lower bound, we have $\Var[P] \ge \E_g\left[\norm{\grad{p_*}{g}}^2_2\right]/rd \ge \condnumber/d$ by taking traces and invoking Lemma~\ref{lem:normgrad} below.
\end{proof}

\subsection{Concentration Inequalities}

In this section we record some concentration inequalities. Let $\zeta_1,...,\zeta_T$ be independent atom variables which each take values in Euclidean space. 

\subsubsection{Standard Concentration}
\label{subsec:standard_concentration}

We will need the following matrix concentration inequality in our analysis of \textsc{TrimmedPCA}.

% \begin{lemma}[\cite{vershynin2010introduction}, see Remark 5.40]
% 	If $A$ is an $N\times n$ matrix whose rows are iid sub-Gaussian random vectors in $\R^n$ with $O(1)$ sub-Gaussian norm and with second moment matrix $\Sigma$, then for every $t\ge 0$, we have that \begin{equation}
% 		\Pr\left[\Norm{\frac{1}{N}A^{\top}A - \Sigma}_2 \le \{\Max{\beta}{\beta^2}\}\right] \ge 1 - 2e^{-\Omega(t^2)} \ \ \ \text{for} \ \ \ \beta = O(\sqrt{n/N}) + t/\sqrt{N}.
% 	\end{equation}
% 	\label{lem:matrixconcentration}
% \end{lemma}

\begin{lemma}[\cite{vershynin2010introduction}]
\label{lem:matrixconcentration}
	Let $\phi:\R\to[0,1]$ be any function. Let $M = \E_{x\sim\N(0,\Id_n)}[\phi(x)\cdot (xx^{\top} - \Id)]$. If $x_1,...,x_N\sim\N(0,\Id_n)$ for $N = \Omega(\{\Max{n}{\log(1/\delta)}\}/\epsilon^2)$, then \begin{equation}
		\Pr\left[\Norm{\vec{M} - \frac{1}{N}\sum^N_{i=1}\phi(x_i)\cdot (x_ix_i^{\top} - \Id)}_2 \ge \epsilon\right] \le\delta.\end{equation}
\end{lemma}

\begin{proof}
	This follows from standard sub-Gaussian concentration; see e.g. Remark 5.40 in \cite{vershynin2010introduction}.
\end{proof}

In our analysis of \textsc{GeoSGD}, we will also need the following standard consequence of Fact~\ref{fact:hypercontractivity}.

\begin{lemma}
\label{lem:polynomial_concentration}
	Let $Z_1,...,Z_T$ be iid scalar random variables which are each given by polynomials of degree $d$ in $\zeta_1,...,\zeta_T$ respectively. If $\Var[Z] \le \sigma^2$ for each $i\in[T]$, then then for any $t > 0$, \begin{equation}
		\Pr\left[\left|\frac{1}{T}\sum^{T}_{i=1}(Z_i - \E[Z_i]) \right| \ge \frac{1}{\sqrt{T}}\cdot O(\log(1/\delta))^{d/2} \cdot\sigma\right] \le \delta.
	\end{equation}
\end{lemma}

Additionally, we will need the following concentration inequality for sums of random variables which only satisfy one-sided bounds. This is a specialization of the martingale concentration result of \cite{bentkus2003inequality} to the iid case, though we also need that result in its full generality for Lemma~\ref{lem:martingale2} below.

\begin{lemma}[Special case of \cite{bentkus2003inequality}]\label{lem:one_sided}
	Let $Z_1,...,Z_T$ be iid, mean-zero random variables. Let $c,s>0$ be deterministic constants for which $Z_i\le c$ with probability one and $\Var[Z_i]\le s^2$ for all $i\in[T]$. Let $\sigma = \Max{c}{s}$. Then for any $\delta > 0$, \begin{equation}
		\Pr\left[\frac{1}{T}\sum^T_{i=1}Z_i\ge \frac{1}{\sqrt{T}}\cdot \sqrt{2}\log(1/\delta)\cdot\sigma\right] \le \delta.
	\end{equation}
\end{lemma}

\subsubsection{Martingale Concentration}
\label{subsec:martingale_concentration}

We now generalize the two scalar concentration inequalities of Section~\ref{subsec:standard_concentration} to the martingale setting. In this section, let $Y(\zeta_1,...,\zeta_T)$ be a real-valued random variable depending on the atom variables $\zeta_1,...,\zeta_T$ which each take values in Euclidean space. Define the martingale differences $Z_i(\zeta) \triangleq \E[Y|\zeta_1,...,\zeta_i] - \E[Y|\zeta_1,...,\zeta_{i-1}]$. When the context is clear, we will suppress the parenthetical $\zeta$.
For brevity, we will use the acronym MDS throughout to refer to martingale difference sequences.

The first lemma is the martingale analogue of Lemma~\ref{lem:polynomial_concentration}, with the slight twist that the moment bounds only hold with high probability. The bounds are slightly weaker than those of Lemma~\ref{lem:polynomial_concentration} but will suffice for our applications.

\begin{lemma}\label{lem:martingale1_polynomial}
	There is a constant $\Cl[c]{weibull}>0$ for which the following holds. Let $\sigma>0$, and suppose the atom variables $\zeta_1,...,\zeta_T$ each take values in $\R^n$, and suppose the martingale differences $\{Z_i\}$ are such that for any realization of $\zeta_1,...,\zeta_{i-1}$, $Z_i(\zeta)$ is a polynomial of degree at most $d$ in $\zeta_i$, and moreover $\Pr\left[\E[Z_i^2|\zeta_1,...,\zeta_{i-1}] \le \sigma^2\right] \ge 1 - \beta$ for each $i\in[T]$. Then for any $t > 0$, \begin{equation}
		\Pr\left[\max_{\ell\in[T]}\left|\sum^{\ell}_{i=1}Z_i\right| \ge (2\log(1/\delta)\cdot d)^{\Cr{weibull} d}\cdot \sqrt{T}\cdot\sigma\right] \le \delta + T\cdot\beta.
	\end{equation}
\end{lemma}

The second lemma is the martingale analogue of Lemma~\ref{lem:one_sided}, again with the twist that the bounds on the differences only hold with high probability.

\begin{lemma}\label{lem:martingale2}
	Let $\{c_i\}_{i\in[T]}$ and $\{s_i\}_{i\in[T]}$ be collections of positive constants, and let $\calE_i$ be the event that $Z_i \le c_i$ and $\E[Z^2_i|\zeta_1,...,\zeta_{i-1}]\le s^2_i$. Let $\sigma_i = \Max{c_i}{s_i}$, and define $\sigma^2 = \sum_i \sigma^2_i$. Then if $\Pr[\calE_i|\zeta_1,...,\zeta_{i-1}] \ge 1 - \beta$ for each $i\in[T]$, then for any $\delta > 0$, \begin{equation}
		\Pr\left[\sum^T_{i=1}Z_i\ge \sqrt{2}\log(1/\delta)\cdot \sigma\right] \le \delta + T\cdot \beta.
	\end{equation}
\end{lemma}

\subsection{Hermite Polynomials and Gradients}
\label{subsec:hermite}

For every $\ell\in\Z_{\ge 0}$, define the oscillator $\phi_{\ell}(z) = \frac{1}{(\ell!)^{1/2}}\He_{\ell}(z)$, where $\He_{\ell}$ is the degree-$\ell$ (probabilist's) Hermite polynomial. $\{\phi_{\ell}(z)\}$ forms an orthonormal basis for $L_2(\R)$ with respect to the Gaussian inner product.

The following elementary identities will be useful.

\begin{fact}[Derivatives]\label{fact:derivs}
	For any $0\le i\le k$, $\phi^{[i]}_k = \sqrt{\frac{k!}{(k - i)!}}\cdot \phi_{k-i}$.
\end{fact}

\begin{fact}[Recurrence Relation]\label{fact:recursion}
	For any $k\ge 0$, $x\cdot \phi_k(x) = \sqrt{k+1}\phi_{k+1}(x) + \sqrt{k}\cdot \phi_{k-1}(x)$.
\end{fact}

\begin{fact}[Linearization Coefficients]\label{cor:linearization}
For any $a,b,c\in \Z_{\ge 0}$ such that $a + b \ge c$, $a+c\ge b$, $b+c\ge a$, and $a + b + c$ is even.
	\begin{equation}
		\E_{g\sim \N(0,1)}\left[\phi_a(g)\phi_b(g)\phi_c(g)\right] = \frac{\sqrt{a!\cdot b!\cdot c!}}{\left(\frac{a+b-c}{2}\right)!\cdot \left(\frac{a-b+c}{2}\right)!\cdot \left(\frac{-a+b+c}{2}\right)!}
	\end{equation} For all other $a,b,c$, this quantity is zero.
\end{fact}

\begin{corollary}
	For any $0\le a\le b$,
	\begin{equation}
		\E_{g\sim\N(0,1)}[g\cdot \phi_a(g)\phi_b(g)] = \bone{b = a + 1}\cdot \sqrt{a+1}
	\end{equation}

	\begin{equation}
		\E_{g\sim\N(0,1)}[\phi_2(g)\phi_a(g)\phi_b(g)] = \bone{b = a + 2}\cdot \sqrt{\frac{(a+1)(a+2)}{2}} + \bone{b = a}\cdot a\sqrt{2}
	\end{equation}
\end{corollary}

\begin{fact}
	For any $v,v'\in\S^{n-1}$ and $\ell,\ell'\in\Z_{\ge 0}$,
	\begin{equation}\E_{g\sim\N(0,\Id_n)}\left[\phi_{\ell}(\langle v,g\rangle\cdot \phi_{\ell'}(\langle v',g\rangle)\right] = \bone{\ell = \ell'}\cdot \langle v,v'\rangle^{\ell}.\end{equation}
\end{fact}

We also record some basic facts about gradients and moments of polynomials in Gaussians.

\begin{fact}[Hypercontractivity]\label{fact:hypercontractivity}
	For a polynomial $f: \R^r\to\R$ of degree $d$, and integer $q \ge 1$, \begin{equation}\E[f(g)^q]^{1/q} \le (q-1)^{d/2}\E[f(g)^2]^{1/2},\end{equation} where the expectation is over $g\sim\N(0,1)$.
\end{fact}

\begin{corollary}\label{cor:normgaussian}
	For any integer $q\ge 1$, $\E_{g\sim\N(0,\Id_m)}[\norm{g}^{2q}_2]^{1/q} \le (q - 1)\cdot (m + 1)$.
\end{corollary}

\begin{proof}
	By Fact~\ref{fact:hypercontractivity} applied to $f(g) \triangleq \norm{g}^2_2$ and $d = 2$, we have that $\E[\norm{g}^{2q}_2]^{1/q} \le \E[\norm{g}^4_2]^{1/2}$. But it is straightforward to compute $\E[\norm{g}^4_2] = m^2 + 2m$, from which the claim follows.
\end{proof}

\begin{corollary}\label{cor:derivpower_bound}
	% L_q norm of derivatives of p are small relative to Var(p)
	For any polynomial $p\in\R_d[x_1,...,x_r]$, $\vec{j} = (j_1,...,j_{\ell})\in[r]^{\ell}$, and integer $q\ge 1$, 
	\begin{equation}
		\E_g[(\D{\vec{j}}{p(g)})^q]^{1/q} \le (q-1)^{d/2}\cdot d^{\ell/2}\cdot \Var[p]^{1/2}
	\end{equation}
\end{corollary}

\begin{proof}
	By Fact~\ref{fact:hypercontractivity}, \begin{equation}
		\E_g[(\D{\vec{j}}{p(g)})^q]^{1/q} \le (q-1)^{d/2}\E_g[(\D{\vec{j}}{p(g)})^2]^{1/2}
	\end{equation} Write $\D{\vec{j}}{p}$ as $\frac{\partial^{\ell}}{\partial x_1^{a_1}\cdots \partial x_r^{a_r}}p$, where $a_i$ is the number of entries of $\vec{j}$ equal to $i$, and write $p$ in the tensored Hermite basis $p= \sum_I c_I\phi_I$. By Fact~\ref{fact:derivs}, \begin{equation}
		\D{\vec{j}}{p(x)} = \frac{\partial^{\ell}}{\partial x_1^{a_1}\cdots \partial x_r^{a_r}}p(x) = \sum_I c_I\left(\prod_{i\in[r]}\phi^{[a_i]}_{I_i}(x_i)\right) = \sum_I c_I \left(\prod_{i\in[r]}\sqrt{\frac{I_i!}{(I_i - a_i)!}}\phi_{I_i - a_i}(x_i)\right),
	\end{equation} so by orthogonality and the fact that $a_1 + \cdots + a_r = \ell$, we see that \begin{equation}
		\E_g[(\D{\vec{j}}{p(g)})^2] = \sum_I c^2_I \cdot \prod_{i\in[r]}\frac{I_i!}{(I_i - a_i)!} \le \sum_{I\neq \emptyset}c^2_I \cdot \prod_{i\in[r]}d^{a_i}= d^{\ell}\cdot\Var[p],
	\end{equation} from which the claim follows.
\end{proof}

We can use Corollary~\ref{cor:derivpower_bound} to bound the moments of $\norm{\grad{p}{g}}^2_2$.

\begin{lemma}\label{lem:normgrad}
	For any polynomial $p\in \R_d[x_1,...,x_r]$ and any integer $q\ge 1$,
	\begin{equation}
		\E[\norm{\grad{p}{g}}^{2q}_2]^{1/q} \le rd\cdot (2q - 1)^{d}\cdot \Var[p]
	\end{equation}
\end{lemma}

\begin{proof}
	We have \begin{equation}
		\E[\norm{\grad{p}{g}}^{2q}_2] \le r^{q-1}\cdot \E[\norm{\grad{p}{g}}^{2q}_{2q}] = r^{q-1}\cdot \sum^r_{i=1}\E\left[\left(\frac{\partial}{\partial x_i}p(g)\right)^{2q}\right] \le r^q\cdot (2q - 1)^{dq}\cdot d^{q}\cdot \Var[p]^q,
	\end{equation} where the first inequality follows by Holder's, and the last step follows by Corollary~\ref{cor:derivpower_bound}.
\end{proof}

It will be useful to give a corresponding lower bound for $\E[\norm{\grad{p}{g}}^2_2]$:

\begin{lemma}\label{lem:normgrad_lb}
	For any polynomial $p\in \R_d[x_1,...,x_r]$, $\E_g[\norm{\grad{p}{g}}^2_2] \ge \Var[p]$.
\end{lemma}

\begin{proof}
	Again, write $p$ in the tensored Hermite basis $p= \sum_I c_I\phi_I$. We know that \begin{equation}
		\sum_i \E\left[\left(\frac{\partial}{\partial x_i}p(g)\right)^2\right] = \sum_Ic_I^2 \cdot \sum_i I_i \ge \sum_{I\neq \emptyset}c_I^2 = \Var[p],
	\end{equation} from which the claim follows.
\end{proof}

The following more careful estimate gives something better than what Cauchy-Schwarz, Corollary~\ref{cor:normgaussian}, and Lemma~\ref{lem:normgrad} imply.

\begin{lemma}\label{lem:x2grad2}
	For any $p\in \R_d[x_1,...,x_r]$, $\E_g\left[\norm{g}^2\cdot \norm{\grad{p}{g}}^2_2\right]^{1/2} \le O(rd)\cdot \Var[p]^{1/2}$.
\end{lemma}

\begin{proof}
	Take any $i,j\in[r]$. Let $q^{i,j}_I$ denote the polynomial $\prod_{\ell\in[|I|]: \ell\neq i,j}\phi_{I_{\ell}}(x_{\ell})$. If $i = j$, then \begin{align}
		\E\left[g^2_i \cdot \left(\frac{\partial}{\partial x_j}p(g)\right)^2\right] &= \E\left[\left(\sum_I c_I \cdot q^{i,i}_I(g) \cdot \sqrt{I_i}\cdot g_i \cdot \phi_{I_i - 1}(x_i)\right)^2\right] \\
		&= \E\left[\left(\sum_I c_I \cdot q^{i,i}_I(g) \cdot \sqrt{I_i}\cdot \left(\sqrt{I_i}\cdot \phi_{I_i}(g_i) + \sqrt{I_i - 1}\cdot \phi_{I_i - 2}(g_i)\right)\right)^2\right] \\
		&\le 2\left(\sum_I c^2_I \cdot I_i^2 + \sum_I c^2_I \cdot I_i(I_i - 1)\right) \le 4d^2\Var[p],
	\end{align} where the second step follows by Corollary~\ref{cor:linearization}, and the third step follows by the elementary inequality $(a+b)^2 \le 2a^2+2b^2$. Likewise, if $i\neq j$, then we have that \begin{align}
		\E\left[g^2_i \cdot \left(\frac{\partial}{\partial x_j}p(g)\right)^2\right] &= \E\left[\left(\sum_I c_I \cdot q^{i,j}_I(g)\cdot g_i\cdot \phi_{I_i}(g_i)\cdot \sqrt{I_j}\cdot \phi_{I_j - 1}(g_j)\right)^2\right] \\
		&= \E\left[\left(\sum_I c_I \cdot q^{i,j}_I(g)\cdot \left(\sqrt{I_i+1}\cdot \phi_{I_i+1}(g_i) + \sqrt{I_i}\cdot \phi_{I_i - 1}(g_i)\right)\cdot \sqrt{I_j}\cdot\phi_{I_j - 1}(g_j)\right)^2\right] \\
		&\le 2\left(\sum_I c^2_I \cdot (I_i + 1)I_j + \sum_I c^2_I \cdot I_iI_j\right) \le 4d(d+1)\Var[p] \le 5d^2\Var[p].
	\end{align} The lemma follows upon summing over $i,j\in[r]$.
\end{proof}

The following basic inequality will also be useful.

\begin{lemma}\label{lem:sumphisquared}
	Let $\calS$ denote the collection of all multisets $I$ of size at most $d$ consisting of elements of $[r]$. Then $\E\left[\left(\sum_I \phi_I(g)^2\right)^2\right] \le O(r)^{2d}$.
\end{lemma}

\begin{proof}
	We have that \begin{equation}
		\E\left[\left(\sum_I \phi_I(g)^2\right)^2\right] \le |\calS|\cdot \E\left[\sum_I \phi_I(g)^4\right] = |\calS|\cdot 9^d\sum_I \E\left[\phi_I(g)^2\right] = |\calS|^2\cdot 9^d = O(r)^{2d},
	\end{equation} where the first step follows by Cauchy-Schwarz, the second by Fact~\ref{fact:hypercontractivity}, the third by orthonormality of $\{\phi_I\}$, and the last by the fact that $|\calS| = O(r)^d$.
\end{proof}

\subsection{Tail Bounds}

We will need the following elementary estimates for Gaussian tails and correlated Gaussians. Define $\erf(\beta) \triangleq \Pr_{h\sim\N(0,1)}[|h|\le \beta]$ and $\erfc(\beta) \triangleq 1 - \erf(\beta)$ (note we eschew the usual normalization). It is an elementary fact that under this normalization, for all $z > 0$ we have that $\erfc(z) \le e^{-z^2/2}$.

\begin{fact}[e.g. Proposition 2.1.2 in \cite{vershynin2018high}]\label{fact:gaussian_tails}
	\begin{equation}\label{eq:gaussian_tails}
		\left(\frac{1}{t} - \frac{1}{t^3}\right)\cdot\frac{1}{\sqrt{2\pi}}e^{-t^2/2} \le \erfc(t) \le \frac{1}{t}\cdot\frac{1}{\sqrt{2\pi}}e^{-t^2/2}.
	\end{equation}
\end{fact}

\begin{fact}\label{fact:corr_gaussians}
	Let $\rho \in (1/2,1)$. For vectors $v,v'\in\S^{n-1}$ for which $\langle v,v'\rangle = \rho$, we have that \begin{equation}\label{eq:corr_gaussians}
		\Pr_{x\sim\N(0,\Id_d)}\left[|\langle v,x\rangle| > 1 \ \wedge \ |\langle v',x\rangle| \le 1 \right] \le O(\sqrt{1 - \rho^2})
	\end{equation}
\end{fact}

\begin{proof}
	We will bound the probability that $\langle v,x\rangle > 1$ and $\langle v',x\rangle \le 1$, from which the desired probability bound in the claim follows up to a constant factor.
	
	First, we may write $v = \rho \cdot v' + \sqrt{1 - \rho^2}\cdot v^{\perp}$ for $v^{\perp}\in\S^{n-1}$ orthogonal to $v'$. Then $\langle v',x\rangle$ and $\langle v^{\perp},x\rangle$ are independent standard Gaussians $g'$ and $g^{\perp}$. Also define $g\triangleq \langle v,x\rangle$ so that $g,g'$ are $\rho$-correlated Gaussians. Provided $g>1$ and $g'\le 1$, the conditional density of $g'$ relative to $g$ is given by $\int^{\frac{1 - \rho g}{\sqrt{1 - \rho^2}}}_{-\infty}\N(0,1,x)$, where $\N(0,1,x)$ denotes the density of the standard Gaussian at $x$. When $g > 1/\rho$, this integral is simply $\frac{1}{2}\cdot \erfc\left(\frac{1 - \rho g}{\sqrt{1 - \rho^2}}\right) \le \frac{1}{2}\exp\left(-\frac{(1 - \rho g)^2}{2(1 - \rho^2)}\right)$.

	We will also crudely upper bound the probability that $1<g \le 1/\rho$ and $g'\le 1$ by the probability that $1 < g\le 1/\rho$, which can be upper bounded by $\frac{1}{4}\left(\frac{1}{\rho} - 1\right) = O(1 - \rho)$.

	We conclude that the quantity on the left-hand side of \eqref{eq:corr_gaussians} is at most \begin{align}
		O(1 - \rho) + \frac{1}{2}\int^{\infty}_{1/\rho}\exp\left(-\frac{(1 - \rho g)^2}{2(1 - \rho^2)}\right)dg &= O(1 - \rho) + \frac{1}{2}\int^{\infty}_{0}\exp\left(-\frac{g^2}{2\cdot (\rho^{-2} - 1)}\right)dg \\
		&\le O(1 - \rho) + \frac{1}{4}\cdot \sqrt{2\pi}\cdot \sqrt{\rho^{-2} - 1} \\ 
		&= O(1 - \rho) + O(\sqrt{1 - \rho^2}) = O(\sqrt{1 - \rho^2}),
	\end{align} where the first step is by shifting the integrand, the second by standard Gaussian integration.
\end{proof}

A similar argument to the above shows the following:
\begin{fact}\label{fact:threegs}
	Let $\rho \in (1/2,1)$. For vectors $v,v'\in\S^{n-1}$ for which $\langle v,v'\rangle = \rho$, and an arbitrary unit vector $v$,  we have that \begin{equation}
		\E_{x\sim\N(0,\Id_d)}\left[\iprod{v}{x}^2 \,|\,\langle v,x\rangle| > 1 \ \wedge \ |\langle v',x\rangle| \le 1 \right] = O(1).
	\end{equation}
\end{fact}

\subsection{Subspaces and Subspace Distances}
\label{subsec:subspace_distances}

\begin{definition}\label{defn:procrustes}
	Given $V,V'\in\St^n_r$, the \emph{Procrustes distance} $\procr{V,V'}$ is given by \begin{equation}\label{eq:procrustes_def}
		\procr{V,V'} \triangleq \min_{O\in O(r)}\norm{V - V'O}_F.
	\end{equation} Let $0\le \theta_1\le \cdots \le \theta_r\le \pi/2$ be the principal angles between $V$ and $V'$. Then we also have that \begin{equation}
		\procr{V,V'} = 2\left(\sum^r_{i=1}\sin^2(\theta_i/2)\right)^{1/2}.
	\end{equation}
\end{definition} 

\begin{definition}\label{defn:chordal}
	Given $V,V'\in\St^n_r$, the \emph{chordal distance} $\chord{V,V'}$ is given by \begin{equation}
		\chord{V,V'}\triangleq (d - \norm{V^{\top}V'}^2_F)^{1/2}
	\end{equation} Let $0\le \theta_1\le \cdots \le \theta_r\le \pi/2$ be the principal angles between $V$ and $V'$. Then we also have that \begin{equation}
		\chord{V,V'} = \left(\sum^r_{i=1}\sin^2\theta_i\right)^{1/2}.
	\end{equation}
\end{definition}

\begin{fact}[Triangle inequality for Procrustes]\label{fact:triangle}
	Given any $V_1,V_2,V_3\in\St^n_r$, \begin{equation}\procr{V_1,V_2} + \procr{V_2,V_3} \ge \procr{V_1,V_3}.\end{equation}
\end{fact}

\begin{lemma}\label{lem:procrustes_chordal}
	$\procr{V,V'}^2/2 \le \chord{V,V'}^2 \le \procr{V,V'}^2$.
\end{lemma}

\begin{proof}
	This follows immediately from the elementary inequality $2\sin^2(\theta/2)\le \sin^2(\theta) \le 4\sin^2(\theta/2)$ for $\theta\in[0,\pi/2]$.
\end{proof}

Next, we give a more refined estimate for $\procr{V,V'}^2 - \chord{V,V'}^2$.

\begin{lemma}\label{lem:diffprocrchord}
	$\procr{V,V'}^2 - \chord{V,V'}^2 \le \procr{V,V'}^4$.
\end{lemma}

\begin{proof}
	From the elementary inequality $4\sin^2(\theta/2) - \sin^2(\theta)\le\sin^4(\theta)$ for $\theta\in[0,\pi/2]$, we see that \begin{equation}\procr{V,V'}^2 - \chord{V,V'}^2 = \left(\sum^r_{i=1}\sin^4\theta_i\right)^2 \le \left(\sum^r_{i=1}\sin^2\theta_i\right)^2 = \chord{V,V'}^4 \le \procr{V,V'}^4\end{equation} as claimed.
\end{proof}

The following consequence of Lemma~\ref{lem:diffprocrchord} will be useful in our analysis of \textsc{GeoSGD}.

\begin{lemma}\label{lem:sigmamax_bound}
	For $V,V^*\in\St^n_r$, we have that $\norm{\Id - V^{\top}V^*}_2 \le \norm{V - V^*_F}$. If $V,V^*$ additionally satisfy that $\norm{V - V^*}_F = \procr{V,V^*}$, then we have that $\norm{\Id - V^{\top}V^*}_2 \le \procr{V,V^*}^2$.
\end{lemma}

\begin{proof}
	It suffices to upper bound $\norm{\Id - V^{\top}V^*}_F$. Note that
	\begin{align}
		\norm{\Id - V^{\top}V^*}^2_F &= d - 2\Tr(V^{\top}V^*) + \norm{V^{\top}V^*}^2_F \\
		&= \norm{V - V^*}^2_F - \chord{V,V^*}^2 \le \norm{V,V^*}^2_F,
	\end{align} from which the first part of the lemma follows.

	For the second bound, note that 
	\begin{equation}
		\norm{V - V^*}^2_F - \chord{V,V^*}^2 = \procr{V,V^*}^2 - \chord{V,V^*}^2 \le \procr{V,V^*}^4,
	\end{equation} where the final step follows by Lemma~\ref{lem:diffprocrchord}.
\end{proof}

The following says that if a set of $r$ orthogonal unit vectors all have large component in $U^*$, then their span is close to the true subspace in the sense of either of the distances above.

\begin{lemma}\label{lem:proj_and_chord}
	Let $\Pi$ denote orthogonal projection to a subspace $U_1\in\G{n}{\ell}$. Let $v_1,...,v_{\ell}\in\S^{n-1}$ be orthogonal and satisfy $\norm{\Pi v_i}_2 \ge 1 - \pcaerror$ for all $i\in[r]$. Let $U_2\triangleq \Span(\{v_i\})$. Then $\chord{U_1,U_2} \le \pcaerror\cdot \ell$ and $\procr{U_1,U_2} \le \sqrt{2} \pcaerror \cdot \ell$. 
\end{lemma}

\begin{proof}
	Let $V_1\in\St^n_{\ell}$ be any frame with columns forming a basis for $U_1$, and let $V_2\in\St^n_{\ell}$ be the frame with columns given by $\{v_i\}_{i\in[\ell]}$. Observe that \begin{equation}
		\chord{U_1,U_2}^2 = \ell - \norm{V_1^{\top}V_2}^2_F = \ell - \Tr\left(V^{\top}_2\Pi V_2\right) \ge \ell - \sum^{\ell}_{i=1}\norm{\Pi v_i}_2 = \pcaerror\cdot \ell.
	\end{equation} as claimed.
\end{proof}

We will also need the gap-free Wedin theorem of \cite{allen2016lazysvd}:

\begin{lemma}[\cite{allen2016lazysvd}, Lemma B.3]\label{lem:wedin}
	Let $\epsilon,\gamma,\mu > 0$. For psd matrices $\vec{A}, \hat{\vec{A}}\in\R^{d\times d}$ for which $\norm{\vec{A} -\hat{\vec{A}}}_2 \le \epsilon$, if $U$ is the matrix whose columns consist of the eigenvectors of $\vec{A}$ with eigenvalue at least $\mu$, and $\hat{U}$ is the matrix whose columns consist of the eigenvectors of $\hat{\vec{A}}$ with eigenvalue at most $\mu - \gamma$, then $\norm{\vec{U}^{\top}\hat{\vec{U}}}_2 \le \epsilon/\gamma$.
\end{lemma}

\begin{claim}\label{claim:swap_proj}
	For any $M\in\R^{n\times n}$ and projectors $\Pi_1,\Pi_2\in\R^{n\times n}$ to subspaces $U_1,U_2\in\G{n}{\ell}$, $\norm{\Pi^{\top}_1 M \Pi_1 - \Pi^{\top}_2 M\Pi_2}_2 \le O(\norm{M}_2 \cdot \chord{U_1,U_2})$.
\end{claim}

\begin{proof}
	We bound $\norm{(\Pi_1 - \Pi_2)^{\top}M\Pi_1}_2$ and $\norm{\Pi^{\top}_2M(\Pi_1-\Pi_2)}_2$ and apply triangle inequality.

	By sub-multiplicativity of the operator norm and the fact that projections have spectral norm 1, $\norm{(\Pi_1 - \Pi_2)^{\top}M\Pi_1}_2 \le \norm{\Pi_1 - \Pi_2}_2 \cdot \norm{M}_2$. Finally, note that \begin{equation}
		\norm{\Pi_1 - \Pi_2}_2 \le \norm{\Pi_1 - \Pi_2}_F = \sqrt{2}\cdot \chord{U_1,U_2},
	\end{equation} from which the claim follows.
\end{proof}

\begin{lemma}\label{lem:projtobasis}
	Let $U^*\in\G{n}{r}$, $V\in\St^n_r$, and $\epsilon>0$. Suppose the columns $v_i$ of $V$ satisfy $\norm{\Pi_U v_i}_2\ge 1 - \epsilon$ for every $i\in[\ell]$. Then there exist orthogonal vectors $v^*_1,...,v^*_{\ell}\in U$ for which $\langle v_i, v^*_i\rangle \ge 1 - \epsilon^2\ell^2/2$ for every $i\in[\ell]$.
\end{lemma}

\begin{proof}
	Let $U\triangleq \Span(\{v_i\})$. By Lemma~\ref{lem:proj_and_chord}, $\procr{U,U^*} \le \epsilon\cdot \ell$, so there exists a frame $V^*\in\St^n_r$ for $U^*$ such that $\norm{V - V^*}_F \le \epsilon\cdot\ell$. Note that $\norm{V - V^*}^2_F = 2\ell - 2\Tr(V^{\top}V^*) = 2\sum^{\ell}_{i=1}(1 - \langle v_i,v^*_i\rangle)$. As $v_i,v^*_i$ are unit vectors $1 - \langle v_i,v^*_i\rangle \ge 0$ for every $i\in[\ell]$, so we conclude that $\langle v_i,v^*_i\rangle \ge 1 - \epsilon^2\ell^2/2$ for each $i\in[\ell]$.
\end{proof}

%!TEX root = poly_main.tex

\section{Warm Start via Trimmed PCA}
\label{sec:warmstart}

The main result of this section is the proof of Theorem~\ref{thm:introwarmstart}.  Let $\calD$ denote the distribution $(X,Y)$ where $Y = P(X)$ is a $\alpha$ non-degerate polynomial of rank $r$ and degree at most $d$ as in the hypothesis of the theorem. Let $U^*$ be the true hidden subspace defining $P$. The proof follows the outline described in the introduction closely. To this end, for a \emph{threshold parameter} $\tau > 0$ and a collection of unit vectors $V = \{v_1,\ldots,v_\ell\}$, define the matrix
\begin{equation}\label{eq:Mdef}
	\vec{M}^{\tau}_{V} \triangleq \Pi_{V^\perp}\cdot\left(\E_{(x,y)\sim\calD}\left[\bone{\left\{|y| > \tau\right\} \ \wedge \ \left\{|\langle v_i , x\rangle| \le 1, \ \forall \ i\in[\ell]\right\}}\cdot (xx^{\top} - \Id)\right]\right)\cdot \Pi_{V^\perp}.
\end{equation}

\begin{algorithm2e}\caption{\textsc{TrimmedPCA}($\calD,\epsilon,\delta$)}
	\DontPrintSemicolon
	\KwIn{Sample access to $\calD$, target error $\epsilon$, failure probability $\delta$}
	\KwOut{Frame for a subspace $U$ with $\procr{U,U^*} \leq \epsilon$, with probability at least $1 - \delta$}
	$V_0 \gets \emptyset$.\;
	$\tau \gets \tau(r,d,\alpha)$\tcp*{Lemma~\ref{lem:warmstartmain})}
	\For{$0 \leq \ell \leq r-1$}{
		Draw $N = O_{r,d,\epsilon}(n)$ samples $(x_1,y_1),...,(x_N,y_N)$ \tcp*{Theorem~\ref{thm:introwarmstart})}
	    Compute an empirical approximation $\overhat{M}^\ell$ to $M^\tau_{V_\ell}$ by drawing $N = O_{r,d,\epsilon}(n)$ samples from the distribution $\calD$.\;
	    Let $v^{\ell+1}$ be the eigenvector with the largest eigenvalue of $\overhat{M}^\ell$. \;
	    $V_{\ell+1} \gets V_\ell \cup \{v^{\ell+1}\}$.\;
	}
	Output $V_r$\;
\end{algorithm2e}

We will show that the above algorithm satisfies the guarantees of Theorem \ref{thm:introwarmstart}. The core of its analysis will be the following main inductive lemma. 

\newcommand{\errorw}{\rho}
\begin{lemma}\label{lem:warmstartmain}
There exists $\tau = \tau(r,d,\alpha)$, a constant $C = C(r,d,\alpha)$ such that the following holds. Let $V = \{v_1,\ldots,v_\ell\}$ for $\ell < r$ be orthonormal vectors such that $\|\Pi_{U^*} v_i\| \geq 1-\errorw$, and $M$ a matrix such that $\|M - M^\tau_V \| \leq \errorw$. Then, the largest eigenvector $v$ of $M$ satisfies $\|\Pi_{U^*} v\| \geq 1 - C \ell^2 \errorw$. 
\end{lemma}

Before proving the lemma, we first show how the main theorem follows from the above.

\begin{proof}[Proof of Theorem \ref{thm:introwarmstart}]
Let $C,\tau$ be as in the above lemma. For a $\rho_0$ to be chosen later, let $\rho_{\ell+1} = C \ell^2 \rho_{\ell}$ for $\ell \geq 0$. Let $N = O(n\log(r/\delta)/\rho_0^2)$. 

We will show by induction that $\|\Pi_{U^*} v_\ell\| \geq 1 - \rho_\ell$. Suppose we have the statement for $v_1,\ldots,v_\ell$ computed by the algorithm. Then, in the next iteration, by Lemma \ref{lem:matrixconcentration}, with probability at least $1-\delta/r$, we will have $\|\overhat{M^\ell} - M^\tau_{V_\ell}\| \leq \errorw_\ell$. In this case, the top eigenvector $v_{\ell+1}$ of $\overhat{M^\ell}$ satisfies $\|Pi_{U^*} v_{\ell+1}\| \geq 1 - C \ell \rho_\ell^{1/4} = 1 - C \ell \rho_{\ell+1}$. 

By a union bound over the $r$ events, we get that with probability at least $1-\delta$, we would have computed orthonormal vectors $v_1,\ldots,v_r$ such that $\|\Pi_{U^*} v_i\| \geq 1 - \rho_r$. Now, by Lemma \ref{lem:proj_and_chord}, $d_P(sp(v_1,\ldots,v_r), U^*) \leq O(\rho_r r)$.

As $\rho_r \leq C^r r^{2r} \rho_0$, the lemma follows by setting $\rho_0 = \epsilon/(Cr)^{2r}$. The overall sample complexity will be $ N = O(r \cdot n \log(r/\delta)/\rho_0^2 = C(r,d,\alpha) n \log(r/\delta)/\epsilon^2$ as stated in the theorem. 

Each iteration of the for loop takes time $O(n^2 N)$ to form the matrix $\overhat{M^\ell}$ and further $O(n^3)$ time to compute the top eigenvector. So the total run-time is $O(r(n^2 N + n^3))$. 
\end{proof}

\ignore{
\begin{theorem}\label{thm:pca_main}
	Let $\epsilon>0$. There is an algorithm \textsc{TrimmedPCA} which takes as input $N = O_{r,d,\epsilon}(n)$ iid samples $(x_1,y_1),...,(x_N,y_N)\sim\calD$ and outputs a basis for a subspace $U\in\G{n}{r}$ for which $\procr{U,U^*} \le \epsilon$.
\end{theorem}}

\subsection{Proof of Lemma \ref{lem:warmstartmain}}
We next prove the Lemma \ref{lem:warmstartmain} which allows us to identify one direction at a time. The proof proceeds as follows:
\begin{enumerate}
    \item We first show a lower bound on the largest eigenvalue of the matrix $M^\tau_V$ when the vectors $v_1,\ldots,v_\ell$ lie in the subspace $U^*$. This is the heart of the proof and follows from a compactness argument. This essentially gives a proof of the lemma when $V \subseteq U^*$ (and $M$ approximates $M^\tau_V$). See Lemmas \ref{lem:exact_top}, \ref{lem:robust_top}. 
    \item The second step is to reduce to the above case. Given $V$ as in the lemma, we find orthonormal vectors $V^* = \{v_1^*,\ldots,v_\ell^*\} \in U^*$ such that $\|v_i - v_i^*\| \leq O(\ell \errorw)$. We then do a perturbation analysis (using elementary linear algebra) to argue that perturbing the vectors $V$ slightly will only incur a small error in the matrix $M^\tau_V$. Specifically, we will show that $\|M^\tau_V - M^\tau_{V^*}\| \leq O(\ell \errorw^{1/4})$. See Lemma \ref{lem:spectral_diff}. 
\end{enumerate}
For brevity, in the remainder of this section let $\Pi^*$ denote orthogonal projection to the true subspace $U^*\subset\R^n$.

First, we show that if the vectors in $V^* = \{v^*_1,...,v^*_{\ell}\}$ were vectors in the true subspace, then the top eigenvector of $\vec{M}^{\tau}_{V^*}$ will be a new vector in the subspace orthogonal to the preceding ones.

\begin{lemma}\label{lem:exact_top}
	There are absolute constants $\tau = \tau_{r,d,\condnumber} > 0$ and $\lambda = \lambda_{r,d,\condnumber} > 0$ for which the following holds. Suppose $V^* = \{v^*_1,...,v^*_{\ell}\}\subset \S^{n-1}$ are orthogonal and is in $U^*$. Then \begin{enumerate}
	 	\item The kernel of $\vec{M}^{\tau}_{V^*}$ contains $\Span(v^*_1,...,v^*_{\ell})$ as well as the orthogonal complement of $U^*$.
	 	\item The top eigenvalue of $\vec{M}^{\tau}_{V^*}$ is at least $\lambda$ and corresponds to a vector in $U^*\backslash \Span(V^*)$.
	 \end{enumerate}
\end{lemma}
Note that Lemma~\ref{lem:exact_top} already gives a nontrivial algorithmic guarantee for $\ell = 0$: given exact access to $\vec{M}^{\tau}_{\emptyset}$, we can recover a vector inside the true subspace by taking its top eigenvector.
\begin{proof}
	Let $\{v^*_i\}_{i\in[\ell]}$ to an orthonormal basis $\{v^*_i\}_{i\in[r]}$ of $U^*$, and let $p^*((V^*)^{\top}x)$ be a realization of the true low-rank polynomial, where the frame $V^*\in\St^n_r$ consists of these basis elements.
	
	(Proof of 1) Certainly $\Span(\{v^*_i\}_{i\in[\ell]})$ lies in the kernel of $\vec{M}^{\tau}_{V^*}$ by definition. Moreover for any $v\in\S^{n-1}$ orthogonal to $U^*$, because $\langle v^*_i,x\rangle,...,\langle v^*_r,x\rangle, \langle v,x\rangle$ are independent Gaussians, call them $g_1,...,g_r,g_{\perp} \sim \N(0,1)$, we have that \begin{align}
		{v^*}^{\top}\vec{M}^{\tau}_{V^*}v &= \E\left[\bone{\{|p^*(g_1,...,g_r)| > \tau\} \ \wedge \ \{|g_i|\le 1 \ \forall \ i\in[\ell]\}}\cdot (g_{\perp}^2-1)\right] \\
		&= \E\left[\bone{\{|p^*(g_1,...,g_r)| > \tau\} \ \wedge \ \{|g_i|\le 1 \ \forall \ i\in[\ell]\}}\right]\cdot\E\left[(g_{\perp}^2-1)\right] \\
		&= 0.
	\end{align}
	(Proof of 2) The fact that the top eigenvector lies in $U^*\backslash \Span(\{v^*_i\}_{i\in[\ell]})$ follows immediately from the fact that it must be orthogonal to both $\Span(\{v^*_i\}_{i\in[\ell]}$ and the orthogonal complement of $U^*$.

	To get a bound on the top eigenvalue, define the quantities $Z_i \triangleq {v^*}^{\top}_i \vec{M}^{\tau}_{V^*} v^*_i$ for $\ell < i\le r$. Again using the fact that $\langle v^*_i,x\rangle,...,\langle v^*_r,x\rangle$ are independent Gaussians $g_1,...,g_r$, we have \begin{equation}\label{eq:sumZs}
		\sum^r_{i = \ell + 1}Z_i = \E\left[\bone{\{|p^*(g_1,...,g_r)| > \tau\} \ \wedge \ \{|g_i|\le 1 \ \forall \ i\in[\ell]\}}\cdot \left(\sum_{i>\ell}g_i^2-(r - \ell)\right)\right].
	\end{equation} We would like to lower bound this quantity, at which point by averaging over $i$ we conclude the proof of the lemma.

	Let $K\subset\R^r$ denote the set of all points $x$ for which $|x_i|\le 1$ for all $1\le i\le \ell$ and for which $\sum^r_{i = \ell+1}x^2_i \le 2(r-\ell)$. For any $p\in\calP^{\condnumber}_{r,d}$, define $\norm{p}_{K} \triangleq \sup_{x\in K}|p(x)|$. By compactness of $K$, $\norm{p}_K<\infty$ for all $p$, and furthermore $\norm{p}_K$ is a continuous function of $p$. If we take $\tau = \tau(\condnumber,r,d,\ell) \triangleq \sup_{p\in\calP^{\condnumber}_{r,d}}\norm{p}_K$, then by compactness of $\calP^{\condnumber}_{r,d}$, is some finite quantity depending only on $\condnumber$, $r$, $d$, and $\ell$. For this choice of $\tau$, we conclude that if a point $(g_1,...,g_r)\in\R^r$ satisfies $|p^*(g_1,...,g_r)| > \tau$ and $|g_i|\le 1$ for all $i\in[\ell]$, then it must lie outside $K$. We conclude that \begin{equation}
		\sum^r_{i= \ell + 1}Z_i \ge (r - \ell)\cdot \Pr\left[\left\{|p^*(g_1,...,g_r)| > \tau\right\} \ \wedge \ \left\{g\not\in K\right\}\right].
	\end{equation} In particular, there exists some $i>\ell$ for which $Z_i \ge \Pr\left[\left\{|p^*(g_1,...,g_r)| > \tau\right\} \ \wedge \ \left\{g\not\in K\right\}\right]$. The right-hand side is a continuous function in $p$, call it $A_p$. For any $p$, there must exist some point $x\not\in K$ for which $p^*(x) > \tau$, so again by compactness of $\calP^{\condnumber}_{r,d}$, we see that $Z_i \ge \lambda$ for some strictly positive constant $\lambda$ depending only on $\condnumber,r,d,\ell$.
\end{proof}

Henceforth, for brevity, we will denote the constants $\tau_{r,d,\condnumber}$ and $\lambda_{r,d,\condnumber}$ from Lemma~\ref{lem:exact_top} by $\tau$ and $\lambda$ respectively.

% temporary definition
\newcommand{\spectraldiff}{\delta}

We next show that the above lemma implies Lemma \ref{lem:warmstartmain} for the case when $V \subseteq U^*$. 

%The following consequence of Lemma~\ref{lem:wedin} and Lemma~\ref{lem:exact_top} tells us that if we take the top eigenvector of any matrix which is spectrally close to $\vec{M}^{\tau}_{V^*}$ where $\{v^*_i\}$ consists of orthogonal vectors $v^*_1,...,v^*_{\ell}$ lying in $U^*$, then we will get a new vector which has large component in $U^*$.

\begin{lemma}\label{lem:robust_top}
	Given orthonormal vectors $V^* = \{v^*_1,...,v^*_{\ell}\} \subseteq U^*$, and a matrix $\vec{M}$ for which $\norm{\vec{M} - \vec{M}^{\tau}_{V^*}}_2 \le \errorw$, the top eigenvector $v$ of $\vec{M}$ satisfies \begin{equation}\norm{\Pi^* v}_2 \ge \left(1 - \frac{\errorw}{\lambda - \errorw}\right)^{1/2} \ge 1 - (2/\lambda) \errorw.\end{equation}
\end{lemma}

\begin{proof}
	By Lemma~\ref{lem:spectral_diff}, the top eigenvalue of $\vec{M}^{\tau}_{V^*}$ is at least that of $\vec{M}^{\tau}_{V^*}$ minus $\errorw$. Let $V^*\in\St^n_{r-\ell}$ be the matrix whose columns consist of $v^*_{\ell+1},...,v^*_r$. Invoking the first part of Lemma~\ref{lem:exact_top}, let $B$ be the matrix whose columns consist of a basis $(v^*_1,...,v^*_{\ell},w_{\ell+1},...,w_n)$ for the kernel of $\vec{M}^{\tau}_{V^*}$, so that \begin{equation}\label{eq:basis}V^*{V^*}^{\top} + BB^{\top} = \Id_n.\end{equation} By applying Lemma~\ref{lem:wedin} to $\vec{M}$ and $\vec{M}^{\tau}_{V^*}$ with $\mu = \gamma = \lambda - \errorw$, we get that $\norm{v^{\top}\cdot B}_2 \le \frac{\errorw}{\lambda - \errorw}$. By \eqref{eq:basis}, \begin{equation}
		\norm{v^{\top}\cdot V^*}_2 = \left(1 - \norm{v^{\top}\cdot B}^2_2\right)^{1/2} \ge \left(1 - \frac{\errorw}{\lambda - \errorw}\right)^{1/2}.
	\end{equation} 
	
	Note that $\norm{\Pi^*v}_2 = \norm{v^{\top}\cdot V^*}_2$. The lemma now follows. 
\end{proof}

Finally, we show that for orthonormal vectors $V = \{v_1,...,v_{\ell}\}$ which all have large component in $U^*$, the matrix $\vec{M}^{\tau}_{V}$ is spectrally close to some $\vec{M}^{\tau}_{V^*}$ for $V^* = \{v^*_1,\ldots v^*_\ell\}$ in $U^*$.

\begin{lemma}\label{lem:spectral_diff}
	There is an absolute constant $\Cl[c]{upsilon}>0$ for which the following holds. Given orthonormal vectors  $V = \{v_1,...,v_{\ell}\}$ for which $\norm{\Pi^* v_i}_2\ge 1 - \pcaerror$ for some $0\le \pcaerror < 1$ for all $i\in[\ell]$, there exist orthonormal vectors $V^* = \{v^*_1,\ldots,v^*_\ell\} \subset U^*$ such that $\Norm{\vec{M}^{\tau}_{V} - \vec{M}^{\tau}_{V^*}}_2 \le \Cr{upsilon} \pcaerror \ell^2$.
\end{lemma}
\begin{proof}
Let $V^* = \{v_1^*,\ldots,v_\ell^*\}$ be orthonormal vectors in $U^*$ guaranteed by Lemma \ref{lem:projtobasis} such that $\iprod{v_i}{v_i^*} \geq 1 - \pcaerror^2 \ell^2/2$. 

For each $0\le a \le \ell$, define the hybrid collections of vectors $V^{(a)} \triangleq \{v^*_1,...,v^*_{\ell-1}, v_{\ell},...,v_r\}$, and also define the hybrid matrices \begin{equation}
		\vec{M}^{(a)} \triangleq \left(\Pi^{\perp}_{\{v_i\}}\right)^{\top}\cdot \left(\E_{(x,y)\sim\calD}\left[\bone{\{|y|>\tau\} \ \wedge \ \{|\langle v^{(a)}_i , x\rangle|\le 1 \ \forall \ i\in[\ell]\}}\cdot (xx^{\top} - \Id)\right]\right)\cdot \Pi^{\perp}_{\{v_i\}}.
	\end{equation} 
	
	Note that $V^{(0)} = V$ and $V^{(\ell)} = V^*$, and similarly $\vec{M}^{(0)} = \vec{M}^{\tau}_{V}$.

	We will bound $\norm{\vec{M}^{(a+1)} - \vec{M}^{(a)}}_2$ for every $0\le a < \ell$, and then bound $\norm{\vec{M}^{(\ell)} - \vec{M}^{\tau}_{V^*}}_2$. The lemma will then follow by triangle inequality.

	\begin{claim}\label{claim:swap_a}
		For any $0\le a < \ell$, $\norm{\vec{M}^{(a+1)} - \vec{M}^{(a)}}_2 \le O(\pcaerror \ell)$.
	\end{claim}

	\begin{proof}
		We will bound $v^{\top}(\vec{M}^{(a+1)} - \vec{M}^{(a)})v$ for any $v\in\R^n$; without loss of generality, we may assume $v$ is orthogonal to $v_1,...,v_{\ell}$.
		
		Let $\calE$ denote the event that $|\iprod{v_a}{x}| > 1$ and $\{|\iprod{v_a^*}{x} \leq 1\}$ or vice-versa. Now, note that the indicator events in the definitions of $M^{(a)}$ and $M^{(a+1)}$ only differ when $\calE$ occurs. Therefore, 
		
		\begin{align}
		   \left|v^{\top}(\vec{M}^{(a+1)} - \vec{M}^{(a)})v\right| &\leq E[1(\calE) \cdot (\iprod{v}{x}^2 + 1)]\\
		   &\leq E[1(\calE)] \cdot (1 + E[\iprod{v}{x}^2 | \calE])\\
		   &= O(\Pr[\calE]), 
		\end{align}
		where the last inequality follows by Fact \ref{fact:threegs}. 
		
\ignore{
		For simplicity, denote by $\calE^{(a)}$ the event over $(x,y)\sim\calD$ that $|y|>\tau$ and $|\iprod{v^{(a)}_i}{x}| \le 1$ for all $i\in[\ell]$ so that
		\begin{align}
			\left|v^{\top}(\vec{M}^{(a+1)} - \vec{M}^{(a)})v\right| &= \left|\E\left[\left(\bone{\calE^{(a+1)}} - \bone{\calE^{(a)}}\right)\cdot\left(\langle v,x\rangle^2 - 1\right)\right]\right| \\
			&\le \E\left[\left(\bone{\calE^{(a+1)}} - \bone{\calE^{(a)}}\right)^2\right]^{1/2} \cdot \E\left[(\iprod{v}{x}^2 - 1)^2\right]^{1/2} \\
			&= \sqrt{2}\cdot \Pr[\calE^{(a+1)} \, \triangle \, \calE^{(a)}]^{1/2},
		\end{align} where the second step follows by Cauchy-Schwarz.}
		
Finally note that by Fact~\ref{fact:corr_gaussians}, 
$$Pr[\calE] \leq O(\sqrt{1 - \iprod{v_i}{v_i^*}}) = O(\errorw \ell).$$
The claim now follows. 
\ignore{
		Note that \begin{equation}
			\Pr[\calE^{(a+1)}\backslash \calE^{(a)}] \le \Pr\left[\left\{\left|\iprod{v_a}{x}\right| > 1\right\} \ \wedge \ \left\{\left|\iprod{v^*_a}{x}\right| \le 1\right\}\right] \le O\left(\sqrt{1 - \norm{\Pi^* v_a}^2_2}\right),
		\end{equation} where the last step follows by Fact~\ref{fact:corr_gaussians}, from which the claim follows upon noting that $1 - \norm{\Pi^* v_a}^2_2 \ge O(\pcaerror)$.}
	\end{proof}

	To bound $\norm{\vec{M}^{(\ell)} - \vec{M}^{\tau}_{V^*}}_2$, we will use Claim~\ref{claim:swap_proj}. We note that the matrix \begin{equation}
		\E_{(x,y)\sim\calD}\left[\bone{\left\{|y| > \tau\right\} \ \wedge \ \left\{|\langle v_i , x\rangle| \le 1 \ \forall \ i\in[\ell]\right\}}\cdot (xx^{\top} - \Id)\right]
	\end{equation} has spectral norm at most $\norm{\E[xx^{\top}]}_2 + 1 = 2$. So if $U \triangleq \Span(v_1,...,v_{\ell})$ and $U'\triangleq \Span(v^*_1,...,v^*_{\ell})$, then by Claim~\ref{claim:swap_proj}, \begin{equation}\norm{\vec{M}^{(\ell)} - \vec{M}^{\tau}_{V^*}}_2 \le O(\chord{U,U'}) \le O(\pcaerror\cdot \ell),\end{equation} where the last step follows by Lemma~\ref{lem:proj_and_chord}.

Lemma~\ref{lem:spectral_diff}  follows by applying the above inequality, Claim~\ref{claim:swap_a} for all $0\le a < \ell$, and triangle inequality.
\end{proof}

We now put Lemmas \ref{lem:robust_top}, \ref{lem:spectral_diff} together to prove Lemma \ref{lem:warmstartmain}. 

\begin{proof}
    [Proof of Lemma~\ref{lem:warmstartmain}]
    Choose $\tau$ to be as in Lemma~\ref{lem:exact_top}. We will choose $C = C \ell^2/\lambda$ for $\lambda$ as in the lemma and $C$ a universal constant.  
    
    Let $V^*$ be the set of $\ell$ orthonormal vectors in $U^*$ as in Lemma~\ref{lem:spectral_diff} so that 
    $$\|M^\tau_V - M^\tau_{V^*}\| \leq O(\errorw \ell^2).$$
    
    Thus, we have $\|M - M^\tau_{V^*}\| \leq O(\errorw \ell^2)$. The lemma now follows by applying Lemma~\ref{lem:robust_top}. 
\end{proof}

\section{Boosting via Stochastic Riemannian Optimization}
\label{sec:boost_details}

In this section we describe our algorithm for boosting a warm start to arbitrary accuracy and defer the details of its analysis to Sections~\ref{sec:subspacedescent} and \ref{sec:realignpoly}.

\begin{theorem}[Error Guarantee for \textsc{GeoSGD}]\label{thm:boost}
	There is an absolute constant $\Cl[c]{finalwarmstart}>0$ such that the following holds. Let $U^*$ be the true subspace of $\calD$. Given $V^{(0)}\in\St^n_r$ spanning a subspace $U$ for which $\procr{U,U^*} \le (\Cr{finalwarmstart}\cdot dr^3)^{-d-2}$, if in the specification of \textsc{GeoSGD} we take \begin{equation}\label{eq:geosgd_T_bound}
		T = \frac{n}{\condnumber}\cdot\log(1/\epsilon)\cdot\poly(\ln(1/\condnumber),r,d,\ln(1/\delta),\ln(n))^d,
	\end{equation} then \textsc{GeoSGD}($\calD, V^{(0)}, \epsilon, \delta$) returns $(\vec{c}^{(T)},V^{(T)})$ for which there exists a realization $(\vec{c}^*,V^*)$ of $\calD$ such that $\procr{V^{(T)},V^*} \le \epsilon$ and $\norm{\vec{c}^{(T)} - \vec{c}^*}_2 \le \epsilon$.
\end{theorem}

\begin{theorem}[Complexity of \textsc{GeoSGD}]\label{thm:boost_runtime}
	Let $T_1 \triangleq O(rd^4)^{d+1}\cdot \log(1/\epsilon)$, $B \triangleq O(\log(T_1\cdot T/\delta))^{2d}$, and $T_2 \triangleq (r/\condnumber)^2 \cdot O(d\cdot\log(T/\delta))^{2\Cr{weibull}d}$. Then \textsc{GeoSGD} draws \begin{equation}N \triangleq T\cdot(B\cdot T_1 + T_2) = \tilde{O}\left(\frac{n\log^2(1/\epsilon)}{\condnumber^3}\cdot\poly(\ln(1/\condnumber),r,d,\ln(1/\delta),\ln(n))^d\right)\end{equation} samples and runs in time $n\cdot r^{O(d)}\cdot N$ time.
\end{theorem}

\subsection{Preliminaries}

Let $M = r^{O(d)}$ be the dimension of the linear space of polynomials of polynomials of degree $d$ over $r$ variables. For $\vec{c} = \{c_I\}\in\R^M$, where $I$ ranges over multisets of size at most $d$ consisting of elements of $[r]$, and $V\in\St^n_r$, let parameters $\Theta = (\vec{c}, V)$ correspond to a rank-$r$ polynomial $F_x(\Theta)\triangleq \sum_I c_I\phi_I(V^{\top}x)$ in the variable $x$. Given a sample $(x,y)\sim\calD$, let $L_x(\Theta)\triangleq (F_x(\Theta) - y)^2$ denote the empirical risk of a single sample.

We will often regard $F_x$ and $L_x$ as functions solely in $\vec{c}$ (resp. $V$) for a fixed choice of $V$ (resp. $\vec{c}$): given a fixed $V$ (resp. a fixed $\vec{c}$), define $F^V_x(\vec{c})$ and $L^V_x(\vec{c})$ (resp. $F^{\vec{c}}_x(V)$ and $L^{\vec{c}}_x(V)$) in the obvious way.

Let $\grad{F_x}{\Theta}$ denote the gradient of $F_x$ as a function on Euclidean space, and let $\gradvec{F_x}{\Theta} \triangleq \grad{F^c_x}{V}$ and $\gradcoef{F_x}{\Theta}\triangleq \grad{F^V_x}{\vec{c}}$ denote its components corresponding to $V$ and $\vec{c}$ respectively. We can compute their gradients, indeed all of their higher derivative tensors, explicitly:

\begin{proposition}\label{prop:derivs}
	For any $x\in\R^n$, $a,b\in\Z_{\ge 0}$, and $\Theta = (c, V)$,
	\begin{equation}
		\frac{\partial^{a+b}}{\partial c_{I^{(1)}}\cdots \partial c_{I^{(a)}}\partial V_{i_1,j_1}\cdots \partial V_{i_b,j_b}}F_x(\Theta) = \begin{cases}
			\left(\prod^b_{\nu = 1}x_{i_{\nu}}\right)\cdot p^{[b]}(V^{\top}x) & \text{if} \ a = 0 \\
			\left(\prod^b_{\nu = 1}x_{i_{\nu}}\right)\cdot \phi^{[b]}_{I}(V^{\top}x) & \text{if} \ a = 1 \\
			0 & \text{otherwise}
		\end{cases}
	\end{equation}
\end{proposition}

From Proposition~\ref{prop:derivs} we conclude that \begin{equation}\label{eq:grad_explicit}
	\gradvec{F_x}{\Theta} = x\cdot (\grad{p}{V^{\top}x})^{\top} \ \ \ \text{and} \ \ \ \gradcoef{F_x}{\Theta} = \{\phi_I(V^{\top}x)\}_I .
\end{equation}
It will be important to consider $\bargradvec{F_x}{\Theta} \triangleq \Pi^{\perp}_V \gradvec{F_x}{\Theta}$ the projection of $\gradvec{F_x}{\Theta}$, to the tangent space of $\G{n}{r}$ at the point $[V]$.

Lastly, we record here an elementary estimate which will be used repeatedly in the proceeding sections and defer its proof to Appendix~\ref{app:main_taylorterms}.

\begin{lemma}\label{lem:main_taylorterms}
	For any integer $m \ge 1$ and $\vec{\ell} = (\ell_1,...,\ell_m)\in[d+1]^m$, 
	\begin{equation}
		\left|\E\left[\prod^m_{\nu = 1}\left\langle \nabla^{[\ell_{\nu}]}F_x(\Theta),(\Theta^* - \Theta)^{\otimes \ell_{\nu}}\right\rangle\right]\right| \le 2^m\cdot \left(2mdr^2\right)^{m(d+1)/2} \cdot \norm{V^* - V}^{\sum_{\nu}\ell_{\nu}}_F \cdot \left(1 + \frac{\norm{\vec{c}-\vec{c}^*}_2}{\norm{V^* - V}_F}\right)^m
	\end{equation}
\end{lemma}

\subsection{Gradient Updates: Vanilla and Geodesic}

\textsc{GeoSGD} alternates between one of two phases: updating $\vec{c}$ or updating $V$. Our updates for $\vec{c}$ are straightforward: at iterate $\Theta = (\vec{c},V)$ and given a batch of samples $(x_0,y_0),...,(x_{B-1},y_{B-1})\sim\calD$, we fix $V$ and take a vanilla gradient descent step using $\frac{1}{B}\sum^{B-1}_{i=0} L^V_{x_i}(\vec{c})$. For learning rate $\etac$, this leads to the update \begin{equation}\label{eq:cprime}
	c'_I = c_I - 2\etac\cdot \frac{1}{B}\sum^{B-1}_{i=0}(F_{x_t}(\Theta) - F_{x_i}(\Theta^*))\cdot \phi_I(V^{\top}x_i) \triangleq c_I - \frac{1}{B}\sum^{B-1}_{i=0}\left(\Deltac^{\Theta,x_i}\right)_I \ \forall \ I.
\end{equation}
The updates for $V$ will be less standard. At iterate $\Theta = (\vec{c},V)$, and given a sample $(x,y)\sim\calD$, consider the geodesic $\Gamma$ on $\G{n}{r}$ with initial point $[V]\in\G{n}{r}$ and initial velocity $\dot{\Gamma}(0)\triangleq \Pi^{\perp}_V \grad{L^{\vec{c}}_x}{V}$, where $L^{\vec{c}}_x(V) \triangleq L_x(\Theta)$.\footnote{We emphasize that technically this is not well-defined as this velocity depends on the choice of representative $V$; indeed, $F^{\vec{c}}_x(V)$ cannot be regarded as a function on $\G{n}{r}$, as $\vec{c}$ is fixed so that different rotations of $V$ will actually yield different values. But as our goal is simply to produce an update rule, we can freely ignore this point and see where this line of reasoning leads.}

Define the vectors $\h{\Theta}{x}\in\R^n,\nabla^{\Theta,x}\in\R^r$ by \begin{equation}\label{eq:defhnabla}
	\h{\Theta}{x} \triangleq 2(F_x(\Theta) - F_x(\Theta^*))\cdot \Pi^{\perp}_V\cdot x \ \ \ \text{and} \ \ \ \nabla^{\Theta,x}\triangleq \grad{p}{V^{\top}x}
\end{equation} so that $\dot{\Gamma}(0) = \h{\Theta}{x}\cdot(\nabla^{\Theta,x})^{\top}$. Geodesics on $\G{n}{r}$ are determined by the SVD of the initial velocity $\dot{\Gamma}(0)$, which is simply given by \begin{equation}
	\dot{\Gamma}(0) = \sigma\cdot \hath^{\Theta,x}\cdot (\hatnab^{\Theta,x})^{\top},
\end{equation} where
\begin{equation}
	\hath^{\Theta,x}\triangleq \frac{\h{\Theta}{x}}{\norm{\h{\Theta}{x}}} \ \ \ \ \ \ \ \
	\hatnab^{\Theta,x} \triangleq \frac{\nabla^{\Theta,x}}{\norm{\nabla^{\Theta,x}}}
	 \ \ \ \ \ \ \ \ \sigma^{\Theta,x} \triangleq \norm{\h{\Theta}{x}}\cdot\norm{\nabla^{\Theta,x}}.
\end{equation}

% Henceforth, where the context is clear, we will suppress the superscripts $\Theta,x$.

Walking along the geodesic with initial velocity $\dot{\Gamma}(0)$ for time $\etav$ then yields the following update rule (for the details, see the derivation of equation (2.65) in \cite{edelman1998geometry}), \begin{equation}\label{eq:Vprime}
	V' \triangleq V - \left(\cos\left(\sigma^{\Theta,x}\etav\right) - 1\right)\cdot V\cdot \hatnab^{\Theta,x}  (\hatnab^{\Theta,x})^{\top} - \sin\left(\sigma\etav\right)\cdot \hath^{\Theta,x}\left(\hatnab^{\Theta,x}\right)^{\top}\triangleq V - \DeltaV^{\Theta,x}.
\end{equation} One readily checks that the columns of $V'$ are orthonormal.

We are now ready to state our boosting algorithm \textsc{GeoSGD}, which is composed of two alternating phases, \textsc{SubspaceDescent} and \textsc{RealignPolynomial} which execute the updates \eqref{eq:cprime} and \eqref{eq:Vprime} respectively. In the next two sections, we will analyze these two phases.

\begin{algorithm2e}[h]
\caption{\textsc{SubspaceDescent}($\calD,V^{(0)},\vec{c}\delta$)}
	\DontPrintSemicolon
	\KwIn{Sample access to $\calD$; frame $V^{(0)}\in\St^n_r$; coefficients $\vec{c}\in\R^M$, failure probability $\delta$}
	\KwOut{$V^{(T)}\in\St^n_r$ which is slightly closer to the true subspace than $V$, provided $(\vec{c},V^{(0)})$ satisfies certain conditions (see Theorem~\ref{thm:subspacedescent_guarantee} for formal guarantees)}
		Define iteration count $T$ according to \eqref{eq:main_T_bound}.\;
		Define learning rate $\etav$ according to \eqref{eq:main_etav_assumption}.\;
		$\Theta^{(0)}\gets (\vec{c},V^{(0)})$\;
		\For{$0\le t < T$}{
			Sample $(x^t,y^t)\sim\calD$
			$\hath \gets \frac{h^{\Theta^{(t)},x^t}}{\norm{h^{\Theta^{(t)},x^t}}}$ and $\hatnab \gets \frac{\nabla^{\Theta^{(t)},x^t}}{\norm{\nabla^{\Theta^{(t)},x^t}}}$\tcp*{equation~\eqref{eq:defhnabla}}
			$\sigma\gets\norm{\h{\Theta^{(t)}}{x^t}}\cdot\norm{\nabla^{\Theta^{(t)},x^t}}$.;
			$V^{(t+1)}\gets V^{(t)} - \DeltaV^{\Theta^{(t)},x^t}$\tcp*{equation~\eqref{eq:Vprime}}
			$\Theta^{(t+1)}\gets (\vec{c},V^{(t+1)})$\;
		}
		Output $V^{(T)}$.\;
\end{algorithm2e}

\begin{algorithm2e}[h]
\caption{\textsc{RealignPolynomial}($\calD,V,\underline{\epsilon},\delta$)}
	\DontPrintSemicolon
	\KwIn{Sample access to $\calD$; $V\in\St^n_r$; target error $\underline{\epsilon}$; failure probability $\delta$}
	\KwOut{$\vec{c}\in\R^M$  for which $(\vec{c}^{(T)},V)$ is close to a realization of $\calD$ (see Section~\ref{sec:realignpoly} for details)}
		Define batch size $B$ according to \eqref{eq:main_B_bound}.\;
		Define iteration count $T$ according to \eqref{eq:main_T_bound_realign}.\;
		Define learning rate $\etac$ according to \eqref{eq:main_etac_assumption}.\;
		$\vec{c}^{(0)}\gets \vec{0}$.\;
		$\Theta^{(0)}\gets (\vec{c}^{(0)},V)$.\;
		\For{$0\le t < T$}{
			Sample $(x^t_1,y^t_1),...,(x^t_B,y^t_B)\sim\calD$.\;
			For every $I$, $c^{(t+1)}_I\gets c^{(t)}_I - \frac{1}{B}\sum^{B-1}_{i=0}\left(\Deltac^{\Theta,x^t_i}\right)_I$ \tcp*{equation \eqref{eq:cprime}}
			$\vec{c}^{(t+1)}\gets \left\{c^{(t+1)}_I\right\}_I$ and $\Theta^{(t)}\gets (\vec{c}^{(t+1)},V)$\;
		}
		Output $\vec{c}^{(T)}$.\;
\end{algorithm2e}

\begin{algorithm2e}[h]
\caption{\textsc{GeoSGD}($\calD,V^{(0)},\epsilon,\delta$)}
\label{alg:geosgd}
	\DontPrintSemicolon
	\KwIn{Sample access to $\calD$, $V^{(0)}\in\St^n_r$, target error $\epsilon$, failure probability $\delta$}
	\KwOut{$\Theta = (\vec{c}^{(T)},V^{(T)})\in\calM$ for which $\procr{V^{(T)},V^*} \le \epsilon$ and $\norm{\vec{c} - \vec{c}^*}_2 \le \epsilon$ for some realization $(\vec{c}^*,V^*)$ of $\calD$}
			Define iteration count $T$ according to \eqref{eq:geosgd_T_bound}\;
			$\delta'\gets \delta/(2T+1)$\;
			\For{$0\le t < T$}{
				$\vec{c}^{(t)}\gets\text{\textsc{RealignPolynomial}}(\calD,V^{(t)},\epsilon/2,\delta')$\;
				$V^{(t+1)}\gets\text{\textsc{SubspaceDescent}}(\calD, V^{(t)},\vec{c}^{(t)},\delta')$\;
			}
			$\vec{c}^{(T)}\gets\text{\textsc{RealignPolynomial}}(\calD,V^{(T)},\epsilon/2,\delta')$\;
			Output $\Theta\triangleq (\vec{c}^{(T)},V^{(T)})$.\;
\end{algorithm2e}

%!TEX root = poly_main.tex

\section{Guarantees for {\mdseries\textsc{RealignPolynomial}}}
\label{sec:realignpoly}

% In this section, we show that given a frame $V\in\St^n_r$ for a subspace which is close to the true subspace, and given a set of coefficients $\vec{c}^{(0)}$ for which $(\vec{c}^{(0)},V)$ is a crude estimate to a realization $(\vec{c}^*,V^*)$ of $\calD$, \textsc{RealignPolynomial} will output a refined set of coefficients $\vec{c}^{(T)}$ which approaches the best possible approximation to $p_*$ one could hope for if one must use the misspecified frame $V$ in place of $V^*$. 

Before we can describe our main result of this section, we require some setup.

Henceforth, fix a frame $V\in \St^n_r$. The aim of \textsc{RealignPolynomial} is to approximately find the $r$-variate, degree-$d$ polynomial $p$ for which $p(V^{\top}x)$ is closest to the true low-rank polynomial. Suppose $V$ was $\beta$-far in subspace distance from the true subspace for some $\beta$, or equivalently, that there was some frame $V^*\in\St^n_r$ for the true subspace for which $\norm{V - V^*}_F = \beta$. By working with $V$ instead of $V^*$, we obviously cannot hope to produce $p$ for which $p(V^{\top}x)$ is exactly equal to the true low-rank polynomial $p_*({V^*}^{\top}x)$. But it is reasonable to hope for a $p$ for which the error incurred by $p$ is comparable to the inherent error $\beta$ contributed by the misspecified frame $V$. The main result of this section is to show that \textsc{RealignPolynomial} can find such a $p$ given $V$:

\begin{theorem}\label{thm:realign_guarantee}
	There are absolute constants $\Cl[c]{warmstart}, \Cl[c]{realignT},\Cl[c]{realigneta},\Cl[c]{floor},\Cl[c]{B}>0$ such that the following holds for any $\underline{\epsilon},\delta>0$. Let $V\in\St^n_r$, and let $(\vec{c}^*,V^*)$ be the realization of $\calD$ for which $\procr{V,V^*} = \norm{V-V^*}_F$. Suppose $\procr{V,V^*}\le (\Cl[c]{warmstart_coef}\cdot dr^3)^{-(d+1)/2}$.

	% Let $\vec{c}^{(0)}\in\R^M$ be a set of coefficents satisfying $\norm{\vec{c}^{(0)} - \vec{c}^*}_2 \le (\Cr{warmstart} dr^3)^{-(d+1)/2}$.

	Define $\vec{c}^{(T)}\, = \,$ \textsc{RealignPolynomial}($\calD,V,\underline{\epsilon},\delta$), where in the specification of \textsc{RealignPolynomial} we take
	\begin{equation}\label{eq:main_etac_assumption}
		\etac\triangleq \left(\Cr{realigneta}rd^4\right)^{d+1}
	\end{equation}
	 \begin{equation}\label{eq:main_T_bound_realign}
		T\triangleq \Cr{realignT}\cdot \left(\Cr{realigneta}rd^4\right)^{d+1} \cdot \log(1/\underline{\epsilon}).
	\end{equation}
	\begin{equation}\label{eq:main_B_bound}
		B\triangleq (\Cr{B}\cdot \log(T/\delta))^{2d}.
	\end{equation}
	Then with probability at least $1 - \delta$, we have that \begin{equation}\label{eq:realign_guarantee}
		\norm{\vec{c}^{(T)} - \vec{c}^*}_2 \le \left(1 + \Cr{floor}\cdot (\Cr{realigneta} dr^4)^{-(d+1)/2}\right)\cdot \{\Max{\underline{\epsilon}}{\procr{V,V^*}}\}.
	\end{equation} Furthermore, \textsc{RealignPolynomial} requires sample complexity \begin{equation}
		N \triangleq O(B\cdot T) = \poly\left(\log(1/\delta), r, d, \log\log(1/\underline{\epsilon})\right)^d \cdot \log(1/\underline{\epsilon})
	\end{equation} and runs in time $n\cdot r^{O(d)}\cdot N$.
\end{theorem}

Before turning to the proof, we set some conventions. Henceforth, fix any $V,V^*$ satisfying the hypotheses of Theorem~\ref{thm:realign_guarantee}. Given coefficients $\vec{c}$ corresponding to the $r$-variate polynomial $p$, define $\polydiff_{\vec{c}}\triangleq p_* - p$. In light of \eqref{eq:realign_guarantee}, it will be convenient in our analysis to quantify, for an iterate $\vec{c}^{(t)}$, the extent to which $\norm{\vec{c}^{(t)} - \vec{c}^*}_2$ differs from $\procr{V,V^*}$ via the (unknown) parameter \begin{equation}\label{eq:rhoc_def}
	\rho_{\vec{c}^{(t)}} \triangleq \frac{\procr{V,V^*}}{\norm{\vec{c}^{(t)}-\vec{c}^*}_2}.
\end{equation} For both $\polydiff_{\vec{c}}$ and $\rho_{\vec{c}}$, we will sometimes omit the subscript when the context is clear.

Note that we would like the eventual output $\vec{c}^{(T)}$ of \textsc{RealignPolynomial} to have large $\rho$. The proof of Theorem~\ref{thm:realign_guarantee} thus comes in two parts: 1) when $\rho_{\vec{c}^{(t)}}$ is small, the next $\rho_{\vec{c}^{(t+1)}}$ is larger by some margin, 2) when $\rho_{\vec{c}^{(t)}}$ is large, $\rho_{\vec{c}^{(t+1)}}$ may be smaller but will still be no smaller than the bound we are targeting in \eqref{eq:realign_guarantee}. Formally:

\begin{theorem}\label{thm:realign_contract_and_control}
	Suppose $\procr{V,V^*}\le O(dr^3)^{-(d+1)/2}$. For any $\delta>0$, let $\vec{c}$ be an iterate in the execution of \textsc{RealignPolynomial}, and let $\vec{c}'$ be the next iterate, given by \begin{equation}
		c' \triangleq c - \frac{1}{B}\sum^{B-1}_{i=0}\Deltac^{\Theta,x_i}
	\end{equation} as defined in \eqref{eq:cprime} for iid samples $(x^0,y^0),...,(x^{B-1},y^{B-1})\sim\calD$. If $\etac\triangleq \Theta(dr^4)^{-d-1}$, then with probability at least $1 - \delta$ over the samples $\{(x^i,y^i)\}_{i\in[B]}$,
	\begin{enumerate}
		\item If $\rho_{\vec{c}} \le 1$, then $\rho_{\vec{c}'} \ge (1 + \Omega(dr^4)^{-d-1})\cdot \rho_{\vec{c}}$.
		\item If $\rho_{\vec{c}} \ge 1$ then $\rho_{\vec{c}'} \ge 1 - O(dr^4)^{-(d+1)/2}$.
	\end{enumerate}
\end{theorem}

We quickly verify that Theorem~\ref{thm:realign_contract_and_control} implies Theorem~\ref{thm:realign_guarantee}.

\begin{proof}[Proof of Theorem~\ref{thm:realign_guarantee}]
	Take any iterate $\vec{c}^{(t)}$ in the execution of \textsc{RealignPolynomial}. Taking $\delta$ to be $1/T$ times the error probability in Theorem~\ref{thm:realign_contract_and_control}, we have by a union bound over all $T$ iterations of \textsc{RealignPolynomial} that with probability at least $1 - \delta$, \begin{equation}
		\rho_{\vec{c}^{(t+1)}} \ge \Min{\left\{1 - O(dr^4)^{-(d+1)/2}\right\}}{\left\{\rho_{\vec{c}^{(t)}}\cdot (1 + \Omega(dr^4)^{-d-1})\right\}},
	\end{equation} for every $0\le t<T$, which can be unrolled to give \begin{equation}
		\rho_{\vec{c}^{(T)}} \ge \Min{\left\{1 - O(dr^4)^{-(d+1)/2}\right\}}{\left\{\rho_{\vec{c}^{(0)}}\cdot (1 + \Omega(dr^4)^{-d-1})^T\right\}}.
	\end{equation}
	We can rewrite this inequality as \begin{equation}
		\norm{\vec{c}^{(t)} - \vec{c}^*}_2 \le \Max{\left\{\frac{\procr{V,V^*}}{1 - O(dr^4)^{-(d+1)/2}}\right\}}{\left\{\norm{\vec{c}^{(0)} - \vec{c}^*}_2\cdot (1 + \Omega(dr^4)^{-d-1})^{-T}\right\}}.
	\end{equation}

	As we are initializing $\vec{c}^{(0)} = \vec{0}$, we have that $\norm{\vec{c}^{(0)} - \vec{c}^*}_2 = \norm{\vec{c}^*}_2\le r$. The theorem follows from taking $T = \Theta(dr^4)^{d+1}\cdot \log(r/\underline{\epsilon}) = \Theta(dr^4)^{d+1}\cdot\log(1/\underline{\epsilon})$.
\end{proof}

As Theorem~\ref{thm:realign_contract_and_control} suggests, we just need to analyze \textsc{RealignPolynomial} on a per-iterate basis. Henceforth, fix an iterate $\vec{c}$; we will sometimes refer to the pair $(\vec{c},V)$ as $\Theta$. Let $(x^0,y^0),...,(x^{B-1},y^{B-1})\sim\calD$ be the batch of samples drawn for the next iteration of \textsc{RealignPolynomial}. 

We first show that it suffices to prove that with high probability, the step $-\frac{1}{B}\sum^{B-1}_{i=0}\Deltac^{x^i}$ is both 1) correlated with the direction $\vec{c} - \vec{c}^*$ in which we want to move, and 2) not too large. 1) and 2) can be interpreted respectively as curvature and smoothness of the gradient of the empirical risk in a neighborhood of our current iterate. Quantitatively, we claim that it suffices to show 

\begin{lemma}[Local Curvature with High Probability]\label{lem:local_curvature_coef}
	For any $\delta>0$ and $\gamma>0$, if $B = \Omega(\log(1/\delta))^{2d}\cdot \gamma^{-2}$, then we have that \begin{equation}\label{eq:local_curvature_coef}
		\frac{1}{B}\sum^{B-1}_{i=0}\left\langle \Deltac^{x^t}, \vec{c} - \vec{c}^*\right\rangle \ge \alphacurve_{\vec{c}}\cdot \etac\cdot \norm{\vec{c}-\vec{c}^*}^2_2
	\end{equation} for \begin{equation}
		\alphacurve_{\vec{c}}\triangleq 1 - \gamma\rho_{\vec{c}} - \norm{\vec{c} - \vec{c}^*}_2 \cdot \left(O(r^{3/2}d)\cdot \rho_{\vec{c}}^2 + O(dr^3)^{(d+1)/2} \cdot \rho_{\vec{c}}(1 + \rho_{\vec{c}})\right)
	\end{equation} with probability at least $1 - \delta$.
\end{lemma}

\begin{lemma}[Local Smoothness With High Probability]\label{lem:local_smoothness_coef}
	For any $\delta>0$, if $B = \Omega(\log(1/\delta))^{2d}$, then we have that \begin{equation}
		\Norm{\frac{1}{B}\sum^{B-1}_{i=0}\Deltac^{x_i}}^2_2 \le \alphasmooth_{\vec{c}}\cdot \etac^2 \norm{\vec{c} - \vec{c}^*}^2_2 \ \ \  \text{for} \ \ \ \alphasmooth_{\vec{c}}\triangleq O(dr^4)^{d+1}\cdot(\Max{1}{\rho_{\vec{c}}^2}).
	\end{equation} with probability at least $1 - \delta$.
\end{lemma}

We verify that Lemmas~\ref{lem:local_curvature_coef} and \ref{lem:local_smoothness_coef} are enough to prove Theorem~\ref{thm:realign_contract_and_control}.

\begin{proof}[Proof of Theorem~\ref{thm:realign_contract_and_control}]
	(Part 1) By \eqref{eq:cprime} we have \begin{equation}
		\norm{\vec{c}' - \vec{c}^*}^2_2 - \norm{\vec{c} - \vec{c}^*}^2_2 = \Norm{\frac{1}{B}\sum^{B-1}_{i=0}\Deltac^{x^i}}^2_2 - 2\Iprod{\frac{1}{B}\sum^{B-1}_{i=0}\Deltac^{x^i}}{\vec{c} - \vec{c}^*}.
	\end{equation} If the events of Lemmas~\ref{lem:local_curvature_coef} and \ref{lem:local_smoothness_coef} occur, then we get that \begin{equation}
		\norm{\vec{c}' - \vec{c}^*}^2_2 - \norm{\vec{c} - \vec{c}^*}^2_2 \le \norm{\vec{c} - \vec{c}^*}^2_2 \cdot \left(\etac\alphacurve_{\vec{c}} - \etac^2\alphasmooth_{\vec{c}}\right),
	\end{equation} If $\rho_{\vec{c}} \le 1$, then we have that \begin{equation}
		\alphacurve_{\vec{c}} \ge 1 - \gamma - \norm{\vec{c} - \vec{c}^*}_2 \cdot O(\rho_{\vec{c}})\cdot O(dr^3)^{(d+1)/2} = 1 - \gamma - O(\procr{V,V^*})\cdot O(dr^3)^{(d+1)/2},
	\end{equation} so if we take $\gamma = 1/4$ and $\procr{V,V^*}_2 \le O(dr^3)^{-(d+1)/2}$, then we ensure that $\alphacurve_{\vec{c}} \ge 1/2$. Additionally, $\rho_{\vec{c}}\le 1$ implies that $\alphasmooth_{\vec{c}} = O(dr^4)^{d+1}$. So if we take $\etac = \Theta(dr^4)^{-d-1}$, we conclude that \begin{equation}
		\norm{\vec{c}' - \vec{c}^*}^2_2 \le (1 - \etac/3)\cdot \norm{\vec{c} - \vec{c}^*}^2_2 \Longleftrightarrow \rho_{\vec{c}'} \ge \rho_{\vec{c}}\cdot (1 - \etac/3)^{-1/2}
	\end{equation} 

	(Part 2) By triangle inequality, \begin{equation}
		\norm{\vec{c}' - \vec{c}^*}_2 \le \norm{\vec{c} - \vec{c}^*}_2 + \Norm{\frac{1}{B}\sum^{B-1}_{i=0}\Deltac^{x^i}}_2.
	\end{equation} If Lemma~\ref{lem:local_smoothness_coef} occurs, then we get that \begin{equation}
		\norm{\vec{c}' - \vec{c}^*}_2 \le \norm{\vec{c} - \vec{c}^*}_2 \cdot (1 + \etac\cdot \sqrt{\alphasmooth_{\vec{c}}}) = \norm{\vec{c} - \vec{c}^*}_2 \cdot \left(1 + \etac\cdot O(dr^4)^{(d+1)/2}\cdot \rho_{\vec{c}}\right),
	\end{equation} or equivalently, \begin{equation}\label{eq:rhorecurse}
		\rho_{\vec{c}'} \ge \rho_{\vec{c}}\cdot \left(1 + \etac\cdot O(dr^4)^{(d+1)/2}\cdot \rho_{\vec{c}}\right)^{-1}.
	\end{equation} For our choice of $\etac = \Theta(dr^4)^{-d-1}$, note that the quantity on the right-hand side of \eqref{eq:rhorecurse}, as a function of $\rho_{\vec{c}}$, has minimum value $\left(1 + O(dr^4)^{-(d+1)/2}\right)^{-1}$ over $\rho_{\vec{c}}\in[1,\infty)$, attained by $\rho_{\vec{c}} = 1$, from which Part 2 of the theorem follows.
\end{proof}

We now proceed to show local curvature and smoothness.

\subsection{Local Smoothness}
\label{subsubsec:realign_local_smoothness}

In this section we show Lemma~\ref{lem:local_smoothness_coef}.

First, by Jensen's, \begin{equation}
	\Norm{\frac{1}{B}\sum^{B-1}_{i=0}\Deltac^{x_i}}^2_2 \le \frac{1}{B}\sum^{B-1}_{i=0}\norm{\Deltac^{x^i}}^2_2,
\end{equation} so to show Lemma~\ref{lem:local_smoothness_coef} it suffices to bound the expectation and variance of the random variable $\norm{\Deltac^x}^2_2$ with respect to $x\sim\N(0,\Id_n)$ and invoke Lemma~\ref{lem:polynomial_concentration}.

We will need the following helper lemma which is a straightforward consequence of Lemma~\ref{lem:main_taylorterms} and whose proof we defer to Appendix~\ref{app:main_taylorfourth}.

\begin{lemma}\label{lem:main_taylorfourth}
	For any $\Theta = (\vec{c},V)$ and $\Theta^* = (\vec{c}^*,V^*)$, $\E[(F_x(\Theta) - F_x(\Theta^*))^4]^{1/2} \le O(dr^3)^{d+1}\cdot\left(\norm{V - V^*}_F + \norm{\vec{c} - \vec{c}^*}_2\right)^2$.
\end{lemma}

We now use this to bound the expectation and variance of $\norm{\Deltac^x}^2_2$.

\begin{lemma}\label{lem:normDeltac_exp}
	$\E[\norm{\Deltac^x}^2_2] \le \etac^2 \cdot O(dr^4)^{d+1}\cdot \left(\norm{V - V^*}_F + \norm{\vec{c} - \vec{c}^*}_2\right)^2$.
\end{lemma}

\begin{proof}
	By Cauchy-Schwarz, \begin{align}
		\frac{1}{4\etac^2}\E\left[\norm{\Deltac^x}^2_2\right] &\le \E\left[\left(F_x(\Theta) - F_x(\Theta^*)\right)^4\right]^{1/2} \cdot \E\left[\left(\sum_I \phi_I(V^{\top}x)^2\right)^2\right]^{1/2} \\
		&\le O(dr^4)^{d+1}\cdot \left(\norm{V - V^*}_F + \norm{\vec{c} - \vec{c}^*}_2\right)^2,
	\end{align} where the second step follows by Lemma~\ref{lem:main_taylorfourth} and Lemma~\ref{lem:sumphisquared}.
\end{proof}

\begin{lemma}\label{lem:normDeltac_var}
	$\E[\norm{\Deltac^x}^4_2] \le \etac^4 \cdot O(dr^4)^{2d+2}\cdot \left(\norm{V - V^*}_F + \norm{\vec{c} - \vec{c}^*}_2\right)^4$.
\end{lemma}

\begin{proof}
	Note that $(F_x(\Theta) - F_x(\Theta^*))^2$ and $\sum_I\phi_I(V^{\top}x)^2$ are degree-$2d$ polynomials in $x$. So by Cauchy-Schwarz, \begin{align}
		\frac{1}{16\etac^4} \E\left[\norm{\Deltac^x}^4_2\right] &\le \E\left[\left(F_x(\Theta) - F_x(\Theta^*)\right)^8\right]^{1/2} \cdot \E\left[\left(\sum_I \phi_I(V^{\top}x)^2\right)^4\right]^{1/2} \\
		&\le 3^{4d}\cdot \E\left[\left(F_x(\Theta) - F_x(\Theta^*)\right)^4\right] \cdot \E\left[\left(\sum_I \phi_I(V^{\top}x)^2\right)^2\right] \\
		&\le O(dr^4)^{2d+2}\cdot \left(\norm{V - V^*}_F + \norm{\vec{c} - \vec{c}^*}_2\right)^4,
	\end{align} where the second step follows by Fact~\ref{fact:hypercontractivity}, and the third step follows by Lemmas~\ref{lem:main_taylorfourth} and \ref{lem:sumphisquared}.
\end{proof}

We are now ready to prove Lemma~\ref{lem:local_smoothness_coef}.

\begin{proof}[Proof of Lemma~\ref{lem:local_smoothness_coef}]
	Note that $\norm{\Deltac^x}^2_2$ is a polynomial of degree $2d$ in $x$. So by Lemma~\ref{lem:polynomial_concentration}, Lemma~\ref{lem:normDeltac_exp}, and Lemma~\ref{lem:normDeltac_var}, we see that \begin{equation}
		\frac{1}{B}\sum^{B-1}_{i=0}\norm{\Delta^{x^i}}^2_2 \le \etac^2 \cdot O(dr^4)^{d+1}\cdot \left(\norm{V - V^*}_F + \norm{\vec{c} - \vec{c}^*}_2\right)^2 \cdot \left(1 + \frac{1}{\sqrt{B}}\cdot O(\log(1/\delta))^d\right),
	\end{equation} so the lemma follows by recalling that $\norm{V - V^*}_F = \procr{V,V^*}$ so that \begin{equation}
		\left(\norm{V - V^*}_F + \norm{\vec{c} - \vec{c}^*}_2\right)^2 \le 4\norm{\vec{c} - \vec{c}^*}^2_2 \cdot (\Max{1}{\rho_{\vec{c}}^2})
	\end{equation} and taking $B = \Omega(\log(1/\delta))^{2d}$.
\end{proof}

We note that this is one of the first of many places where the fact that one cannot obtain a $\vec{c}$ whose error is much smaller than the ``misspecification error'' $\procr{V,V^*}$ incurred by the subspace $V$ manifests: here, our bounds on the magnitudes of the gradient steps $\norm{\Deltac^x}$ inherently depend on $\procr{V,V^*}$, yet we require that the gradient steps have norm bounded by $\norm{\vec{c} - \vec{c}^*}$. 

\subsection{Local Curvature}
\label{subsubsec:realign_strategy}

We begin by outlining our argument for proving Lemma~\ref{lem:local_curvature_coef}. It will be helpful to first decompose $\langle \Deltac, \vec{c} - \vec{c}^*\rangle$ into ``dominant'' and ``non-dominant'' terms.

\begin{proposition}\label{prop:basic_taylor_coef}
	For every monomial index $I$ and any $x\in\R^n$, let \begin{equation}
	({\Deltac'}^x)_I \triangleq -2\etac\cdot \iprod{\grad{F_x}{\Theta}}{\Theta^* - \Theta} \cdot \phi_I(V^{\top}x) \ \ \ \text{and} \ \ \ ({\Deltac''}^x)_I \triangleq -2\etac\cdot \residual^{x}\cdot \phi_I(V^{\top}x) \ \forall \ I.
\end{equation} Then $\Deltac^{x} = {\Deltac'}^x + {\Deltac''}^x$.
\end{proposition}

\begin{proof}
	${\Deltac'}^x$ and ${\Deltac''}^x$ correspond to the first-order and higher-order terms in the Taylor expansion of $\Deltac^x$. Concretely, recall that \begin{equation}
		(\Deltac^{x})_I = 2\etac\cdot (F_x(\Theta) - F_x(\Theta^*))\cdot \phi_I(V^{\top}x).
	\end{equation} 
	We can decompose $\Deltac^x$ by Taylor expanding the factor $F_x(\Theta) - F_x(\Theta^*)$ around $\Theta^* = \Theta$ to get \begin{equation}\label{eq:taylor_loss}
		F_x(\Theta^*) - F_x(\Theta) = \Iprod{\grad{F_x}{\Theta}}{\Theta^* - \Theta} + \residual^{\Theta,x} \ \ \text{for} \ \ \residual^{\Theta,x} \triangleq \sum^{d+1}_{\ell = 2}\frac{1}{\ell!}\Iprod{\nabla^{[\ell]}F_x(\Theta)}{(\Theta^* - \Theta)^{\otimes \ell}},
	\end{equation} from which the proposition follows.
\end{proof}

Motivated by Proposition~\ref{prop:basic_taylor_coef}, for any $x\in\R^n$ define \begin{equation}
	\domcoef^{x} \triangleq \iprod{{\Deltac'}^x}{\vec{c} - \vec{c}^*}, \ \ \ \text{and} \ \ \ \Ecoef^{x}\triangleq \iprod{{\Deltac''}^x}{\vec{c} - \vec{c}^*}.
\end{equation}
To show Lemma~\ref{lem:local_curvature_coef}, we will show that the random variables $\frac{1}{B}\sum^{B-1}_{i=0}\domcoef^{x^i}$ and $\frac{1}{B}\sum^{B-1}_{i=0}\Ecoef^{x^i}$ are respectively large and negligible with high probability. Eventually we will invoke the concentration inequalities of Lemmas~\ref{lem:polynomial_concentration} and \ref{lem:one_sided} to control them, so we will compute the expectations (Section~\ref{subsubsec:realign_expectation}) and variances (Section~\ref{subsec:domcoef_conc}) of their summands next.

\subsubsection{Local Curvature in Expectation}
\label{subsubsec:realign_expectation}

In this section we give bounds for $\mu_{\domcoef}\triangleq \E_x[\domcoef^{x}]$ and $\mu_{\Ecoef}\triangleq\E_x[\Ecoef^{x}]$ in the following two lemmas. Throughout this section, we will omit the superscript $x$ when the context is clear.

\begin{lemma}\label{lem:domcoef_term_curvature_expectation}
	$\mu_{\domcoef} \ge 2\etac\cdot \norm{\vec{c}-\vec{c}^*}_2\cdot \left(\norm{\vec{c} - \vec{c}^*}_2 - O(r^{3/2}d)\cdot \procr{V,V^*}^2\right)$.
\end{lemma}

\begin{lemma}\label{lem:Ecoef_term_curvature_expectation}
	$|\mu_{\Ecoef}| \le 2\etac\cdot  O(dr^3)^{(d+1)/2}\cdot \procr{V,V^*}\cdot \norm{\vec{c}-\vec{c}^*}_2\cdot (\procr{V,V^*} + \norm{\vec{c} - \vec{c}^*}_2)$.
\end{lemma}

In this section we will give the proof of Lemma~\ref{lem:domcoef_term_curvature_expectation}; we will defer the proof of Lemma~\ref{lem:Ecoef_term_curvature_expectation} to Appendix~\ref{app:Ecoef_term_curvature_expectation}.

\begin{proof}[Proof of Lemma~\ref{lem:domcoef_term_curvature_expectation}]
	We have that \begin{equation}\label{eq:deltacprimecor}
		\iprod{\Deltac'}{\vec{c} - \vec{c}^*} = -2\etac\iprod{\grad{F_x}{\Theta}}{\Theta^* - \Theta}\cdot \polydiff(V^{\top}x)
	\end{equation}
	Writing \begin{align}
		\iprod{\grad{F_x}{\Theta}}{\Theta^* - \Theta} &= \iprod{\gradvec{F_x}{\Theta}}{V^* - V} + \iprod{\gradcoef{F_x}{\Theta}}{\vec{c}^* - \vec{c}} \\
		&= x^{\top}(V^* - V)\nabla + \polydiff(V^{\top}x) \\
		&= x^{\top}\Pi^{\perp}_V(V^*-V)\nabla + x^{\top}\Pi_V(V^* - V)\nabla + \polydiff(V^{\top}x)  \\
		&= x^{\top}\Pi^{\perp}_V V^*\nabla + x^{\top}\Pi_V\cdot (V^* - V)\nabla + \polydiff(V^{\top}x),\label{eq:decompose_corr}
	\end{align}
	we see that \eqref{eq:deltacprimecor} is given by $2\etac$ times \begin{equation}\label{eq:cprimecor_lowest}
		\underbrace{\left(\polydiff(V^{\top}x)\right)^2}_{\circled{A}} + 
		\underbrace{\polydiff(V^{\top}x)\cdot \left(x^{\top}\Pi_V(V^*-V)\nabla\right)}_{\circled{B}} +
		\underbrace{\polydiff(V^{\top}x)\cdot \left(x^{\top}\Pi^{\perp}_V V^*\nabla\right)}_{\circled{C}} 
	\end{equation}

	Note that $x^{\top}\Pi_V$ and $x^{\top}\Pi^{\perp}_V$ are independent Gaussian vectors with mean zero and covariances $\Pi_V$ and $\Pi^{\perp}_V$ respectively. So we readily conclude that
	\begin{observation}\label{obs:Cprimevanish}
	 	For any $V$, the expectation of $\circled{C}$ with respect to $x$ vanishes.
	\end{observation}
	The following is also immediate:
	\begin{observation}\label{obs:Aprime_exp}
		$\E[\circled{A}] = \E_{g\sim\N(0,\Id_r)}[\polydiff(g)^2] = \norm{\vec{c} - \vec{c}^*}^2_2$.
	\end{observation}

	We now turn to bounding $\E[\circled{B}]$. We will make use of the following helper bound whose proof we defer to Appendix~\ref{app:helper_circledBprime}

	\begin{proposition}\label{prop:helper_circledBprime}
		If $\norm{V - V^*} = \procr{V,V^*}$, then \begin{equation}
			\E_g\left[\left(x^{\top}\Pi_V(V^* - V)\grad{p}{V^{\top}x}\right)^2\right]^{1/2} \le \procr{V,V^*}^2 \cdot O(r^{3/2}d).
		\end{equation}
	\end{proposition}

	\begin{lemma}\label{lem:Bprimeexp}
		$\E[\circled{B}] \le O(r^{3/2}d)\cdot \procr{V,V^*}^2\cdot \norm{\vec{c}-\vec{c}^*}_2$.
	\end{lemma}

	\begin{proof}
		Note that
		\begin{align}
			\left|\E[\circled{B}]\right| &= \left|\E\left[\polydiff(V^{\top}x)\cdot \left(x^{\top}\Pi_V(V^*-V)\grad{p}{V^{\top}x}\right)\right]\right| \\
			&\le \E_g\left[\polydiff(g)^2\right]^{1/2}\cdot \E_g\left[\left(g^{\top}V^{\top}(V^* - V)\grad{p}{g}\right)^2\right]^{1/2} \\
			&\le \norm{\vec{c} - \vec{c}^*}_2 \cdot \procr{V,V^*}^2\cdot O(r^{3/2}d),\label{eq:Bexp}
		\end{align} where the second step follows by Cauchy-Schwarz, and the third by Proposition~\ref{prop:helper_circledBprime}.
	\end{proof}

	Lemma~\ref{lem:domcoef_term_curvature_expectation} now follows from \eqref{eq:cprimecor_lowest}, Observations~\ref{obs:Cprimevanish} and \ref{obs:Aprime_exp}, and Lemma~\ref{lem:Bprimeexp}.
\end{proof}

% By \eqref{eq:cprime} we have that \begin{equation}
% 	\norm{\vec{c} - \vec{c}^*}^2_2 - \norm{\vec{c}' - \vec{c}^*}^2_2 = -\norm{\Deltac}^2_2 + 2\langle \Deltac, \vec{c} - \vec{c}^*\rangle,
% \end{equation} so \begin{align}
% 	\langle \Deltac, \vec{c} - \vec{c}^*\rangle &= -2\etac\cdot (F_x(\Theta) - F_x(\Theta^*))\cdot \polydiff(V^{\top}x) \\
% 	&= 2\etac\sum^{d + 1}_{\ell=1}\frac{1}{\ell!}\left\langle \nabla^{[\ell]}F_x(\Theta),(\Theta^* - \Theta)^{\otimes \ell}\right\rangle \cdot \polydiff(V^{\top} x)\label{eq:newtaylor2}.
% \end{align} 

% We would like to argue that \eqref{eq:Bexp} is upper bounded by something that scales with $\norm{\vec{c} - \vec{c}^*}_2\cdot \procr{V,V^*}^2$, in which case, as long as $\norm{\vec{c}- \vec{c}^*}_2$ is not too small, $\E[\circled{A}]$ dominates $\E[\circled{B}]$.

\subsubsection{Local Curvature with High Probability}
\label{subsec:domcoef_conc}

In this section, we complete the proof of Lemma~\ref{lem:local_curvature_coef} by establishing high-probability bounds for $\domcoef^x$ and $\Ecoef^x$. That is, we argue that with high probability, the dominant term given by $\domcoef$ is large and the error from Taylor approximation is small. Specifically, we will show:

\begin{lemma}\label{lem:domcoef_conc}
	For any $\delta> 0$ and $\gamma > 0$, if $B = \Omega(\log(1/\delta))^d \cdot O(\gamma^{-2})$, then \begin{equation}
		\frac{1}{B}\sum^{B-1}_{i=1}\domcoef^{x^i}\ge \etac\left(\norm{\vec{c}-\vec{c}^*}^2_2 - O(r^{3/2}d)\cdot\procr{V,V^*}^2\cdot \norm{\vec{c} - \vec{c}^*}_2 - \gamma\cdot\procr{V,V^*}\cdot\norm{\vec{c} - \vec{c}^*}_2\right)
	\end{equation}
\end{lemma}

\begin{lemma}\label{lem:Ecoef_conc}
	For any $\delta > 0$, if $B = \Omega(\log(1/\delta))^{2d}$, then \begin{equation}
		\left|\frac{1}{B}\sum^{B-1}_{i=0}\Ecoef^{x^i}\right| \le \etac\cdot O(dr^3)^{(d+1)/2} \cdot \procr{V,V^*} \cdot \norm{\vec{c}-\vec{c}^*}_2\cdot (\procr{V,V^*} + \norm{\vec{c} - \vec{c}^*}_2)
	\end{equation}
\end{lemma}

We defer their proofs to Appendices~\ref{app:domcoef_conc} and \ref{app:Ecoef_conc} respectively. We can finally deduce Lemma~\ref{lem:local_curvature}, completing the proof of Theorem~\ref{thm:realign_contract_and_control} and thus Theorem~\ref{thm:realign_guarantee}.

\begin{proof}[Proof of Lemma~\ref{lem:local_curvature}]
	By Lemmas~\ref{lem:domcoef_conc} and \ref{lem:Ecoef_conc}, and the earlier calculation showing that for any $x$, $\langle \Deltac^{x}, V-V^*\rangle = \domcoef^{x} + \Ecoef^{x}$, we see that under our choice of $B$, \eqref{eq:local_curvature_coef} holds with probability $1 - 3\delta$. By replacing $3\delta$ with $\delta$, and absorbing the constant factors, the lemma follows.
\end{proof}

%!TEX root = poly_main.tex

\section{Guarantees for {\mdseries\textsc{SubspaceDescent}}}
\label{sec:subspacedescent}

Henceforth, fix a set of coefficients $\vec{c}\in\R^M$. In contrast with \textsc{RealignPolynomial}, the aim of \textsc{SubspaceDescent} is to take a frame $V^{(0)}$ of a subspace which is somewhat close to the true subspace and refine it to some $V^{(T)}$ which is slightly closer, using only the misspecified coefficients $\vec{c}$. It turns out that if the misspecification error of $\vec{c}$ is comparable to the subspace distance from $V^{(0)}$ to the true subspace, \textsc{SubspaceDescent} can indeed accomplish this, and this is the main result of this section.

\begin{theorem}\label{thm:subspacedescent_guarantee}
	There are absolute constants $\Cl[c]{subspacedescentT}, \Cl[c]{subspacedescentetav} > 0$ and $\Cl[c]{tiny} < 1/10$ such that the following holds for any $\delta>0$. Let $V^{(0)} \in\St^n_r$, and let $(\vec{c}^*,V^*)$ be the realization of $\calD$ for which $\procr{V,V^*} = \norm{V - V^*}_F$. Suppose \begin{equation}\label{eq:V0warmstart}
		\procr{V^{(0)},V^*}\le \Cr{tiny}\cdot \condnumber \cdot O(dr^3)^{-d-2},
	\end{equation}
	Let $\vec{c}$ be a set of coefficients satisfying \begin{equation}\label{eq:ratio_bound}
		\procr{V^{(0)},V^*} \ge \frac{1}{2}\norm{\vec{c}-\vec{c}^*}_2 
	\end{equation}
	Define $V^{(T)} = $ \textsc{SubspaceDescent}($\calD, V^{(0)},\vec{c},\delta)$, where in the specification of \textsc{SubspaceDescent} we take
	\begin{equation}\label{eq:main_etav_assumption}
		\etav \triangleq \frac{\condnumber}{T\cdot n}\left(\Cr{subspacedescentetav}\cdot dr^3\ln(T/\delta)\right)^{-d-2}
	\end{equation}
	\begin{equation}\label{eq:main_T_bound}
		T \triangleq \left(\frac{r}{\condnumber}\right)^2\cdot \left(\Cr{subspacedescentT}\cdot d\cdot \log(1/\delta)\right)^{2\Cr{weibull}d}.
	\end{equation} Then with probability at least $1 - \delta$, we have that \begin{equation}
		1 - \frac{\procr{V^{(T)},V^*}^2}{\procr{V^{(0)},V^*}^2} \ge \frac{\condnumber}{n}\cdot \poly(\ln(1/\condnumber),r,d,\ln(1/\delta))^{-d}.
	\end{equation} Furthermore, \textsc{SubspaceDescent} draws $N \triangleq O(T)$ samples and runs in time $n\cdot r^{O(d)}\cdot N$.
\end{theorem}

Henceforth, let $\delta,V^{(0)},V^*,\vec{c},\vec{c}^*,T,\etav$ satisfy the hypotheses of Theorem~\ref{thm:subspacedescent_guarantee}.
% In light of \eqref{eq:ratio_bound} in Theorem~\ref{thm:subspacedescent_guarantee}, it will be convenient to quantify the extent to which the misspecification error $\norm{\vec{c} - \vec{c}^*}_2$ is comparable to the subspace distance $\procr{V^{(0)},V^*}$ via the (unknown) parameter \begin{equation}
% 	\rho \triangleq \Min{1}{\frac{\procr{V^{(0)},V^*}}{\norm{\vec{c}-\vec{c}^*}_2}}, \label{eq:rhodef}
% \end{equation}
%  for which we only assume $\rho\le \rhobar$, where the parameter $\rhobar$ is known.

As discussed in Section~\ref{sec:tracking}, a single execution of \textsc{SubspaceDescent} should be thought of as a single step of stochastic gradient descent over a batch of size $T$. The only difference lies in the fact that the empirical risk we work with in each iteration of \textsc{SubspaceDescent} is slightly different, as our subspace estimate $V^{(t)}$ continues to update by a small amount. So just as we analyzed the individual steps of \textsc{RealignPolynomial} in Lemma~\ref{thm:realign_contract_and_control} via local curvature and smoothness estimates, we would like to do the same for an entire execution of \textsc{SubspaceDescent}. That is, we want to show that with high probability, the steps $-\DeltaV^{\Theta^t,x^t}$ are 1) bounded, and 2) each correlated with the direction $V^* - V^{(t)}$ in which we want to move. Quantitatively, we claim that it suffices to show

\begin{lemma}[Local Smoothness With High Probability]\label{lem:local_smoothness}
	\begin{equation}
		\norm{V^{(0)} - V^{(T)}}^2_F \le \etav^2\cdot O(dr^3\ln(T/\delta))^{d+2}\cdot O(n)\cdot \procr{V^{(0)},V^*}^2.
	\end{equation} with probability at least $1 - \delta$. 
\end{lemma}

\begin{lemma}[Local Curvature with High Probability]\label{lem:local_curvature}
	\begin{equation}
		\sum^{T-1}_{t=0} \left\langle \DeltaV^{\Theta^{(t)},x^t}, V^{(t)} - V^*\right\rangle \ge T\cdot \etav\cdot (\condnumber/4)\cdot\procr{V^{(0)},V^*}^2
	\end{equation} with probability at least $1 - \delta$.
\end{lemma}

We verify that Lemmas~\ref{lem:local_curvature} and \ref{lem:local_smoothness} are enough to prove Theorem~\ref{thm:subspacedescent_guarantee}.

\begin{proof}[Proof of Theorem~\ref{thm:subspacedescent_guarantee}]
	For every $0\le t < T$, we have \begin{equation}\label{eq:gradient_expansion}
		\norm{V^{(t+1)}-V^*}^2_F - \norm{V^{(t)}-V^*}^2_F = \norm{\DeltaV^{\Theta^{(t)},x^t}}^2_F - 2\left\langle \DeltaV^{\Theta^{(t)},x^t}, V^{(t)}-V^*\right\rangle.
	\end{equation} 
	If the event of Lemma~\ref{lem:local_curvature} holds, then \begin{equation}
		\sum^{T-1}_{t=0}\left\langle \DeltaV^{\Theta^{(t)},x^t}, V^{(t)}-V^*\right\rangle \ge T\cdot(\condnumber/4)\cdot \etav\cdot\procr{V^{(0)},V^*}^2.
	\end{equation}
	If the event of Lemma~\ref{lem:local_smoothness} holds, then \begin{align}
		\sum^{T-1}_{t=0}\norm{\DeltaV^{\Theta^{(t)},x^t}}^2_F &\le T\cdot \etav^2\cdot O(dr^3\ln(T/\delta))^{d+2}\cdot O(n)\cdot \procr{V^{(0)},V^*}^2 \\
		&\le O\left(\condnumber\cdot\etav\cdot \procr{V^{(0)},V^*}^2\right).
	\end{align} where the last step follows by the choice of $\etav$ in \eqref{eq:main_etav_assumption}, and the constant factor in the last expression can be made arbitrarily small. By summing \eqref{eq:gradient_expansion} over $t$, telescoping, and recalling that $\norm{V^{(0)} - V^*}^2_F =\procr{V^{(0)},V^*}^2$,  we conclude that \begin{equation}
		\norm{V^{(T)}-V^*}^2_2 - \procr{V^{(0)},V^*}^2 \le -T\cdot (\condnumber/5)\cdot \etav\cdot \procr{V^{(0)},V^*}^2,
	\end{equation} from which we get, because $\procr{V^{(T)},V^*} \le \norm{V^{(T)}-V^*}_F$, that \begin{equation}
		1 - \frac{\procr{V^{(T)},V^*}^2}{\procr{V^{(0)},V^*}^2} \ge T\cdot (\condnumber/5)\cdot \etav.
	\end{equation} The claim follows by substituting the choice of $\etav$ and $T$ in \eqref{eq:main_etav_assumption} and \eqref{eq:main_T_bound}.
\end{proof}

We now proceed to show Lemma~\ref{lem:local_smoothness} and \ref{lem:local_curvature}.

\subsection{Local Smoothness}

In this section we establish Lemma~\ref{lem:local_smoothness}. We also show that $\procr{V^{(t)},V^*}$ does not change much, both in expectation (Lemma~\ref{lem:safe_exp}) and with high probability (Lemma~\ref{lem:safe}), as $t$ varies. While we have already seen that Lemma~\ref{lem:local_smoothness} is needed to prove Theorem~\ref{thm:subspacedescent_guarantee}, Lemmas~\ref{lem:safe} and \ref{lem:safe_exp} will be crucial to our arguments in later sections, where we argue that at each step $t$ we make progress scaling with the distance $\procr{V^{(t)},V^*}$ and thus need that this distance is comparable to the initial distance $\procr{V^{(0)},V^*}$.
% These bounds, while naive, will be essential to ensuring that with high probability, 1) the iterates never leave the local convergence ball that we assume $V^{(0)}$ starts out inside, and 2) the martingale difference sequences we consider in our arguments for concentration have bounded differences.

For a fixed $\Theta$, we will first show a high-probability bound on the norm of $\DeltaV^{\Theta,x}$, that is, we bound the size of the step made in a single iteration inside \textsc{SubspaceDescent}.

Where the context is clear, we will suppress superscript $\Theta,x$. Then very naively, using the inequalities $1 - \cos(x) \le x$ and $|\sin(x)| \le x$ for all $x \ge 0$, we have \begin{equation}\label{eq:DeltaVnorm}
	\norm{\DeltaV}_F \le (1 - \cos(\sigma\etav))(2\sqrt{r}) + |\sin(\sigma\etav)| \le 2\sqrt{r}\cdot \sigma\etav  + \sigma\etav \le 3\sqrt{r}\cdot \sigma\etav.
\end{equation}

We first bound the moments of $\sigma^2$.

\begin{lemma}\label{lem:sigsquared}
	For all integers $q \ge 1$, $\E[\sigma^{2q}]^{1/q} \le O(nrd)\cdot O(q^2dr^3)^{d+2}\cdot \left(\norm{V -V^*}_F + \norm{\vec{c} - \vec{c}^*}_2\right)^2$.
\end{lemma}

\begin{proof}
	Recall that $\sigma = 2(F_x(\Theta) - F_x(\Theta^*))\cdot \norm{\Pi^{\perp}_Vx}_2 \cdot \norm{\grad{p}{V^{\top}x}}_2$. So by Cauchy-Schwarz, \begin{equation}\label{eq:sigmasquared_main}
		\E[\sigma^{2q}]^{1/q} \le 4\E[(F_x(\Theta) - F_x(\Theta^*))^{4q}]^{1/2q}\cdot \E\left[\norm{\Pi^{\perp}_Vx}^{4q}_2 \cdot \norm{\grad{p}{V^{\top}x}}^{4q}_2\right]^{1/2q}.
	\end{equation} The second factor in \eqref{eq:sigmasquared_main} is simply \begin{align}
		\MoveEqLeft \E_{g'\sim\N(0,\Pi^{\perp}_V)}[\norm{g'}^{4q}_2]^{1/2q} \cdot \E_{g\sim\N(0,Id_r)}[\norm{\grad{p}{g}}^{4q}_2]^{1/2q} \\
		&\le \left((2q - 1)\cdot (n - r + 1)\right)\cdot \left(rd\cdot (4q-1)^{d} \cdot \Var[p]\right) \\
		&\le O(n)\cdot qrd\cdot (4q)^{d}\cdot\Var[p] \\
		&\le O(n)\cdot rd\cdot (4q)^{d+1},
	\end{align} where in the first step we used Corollary~\ref{cor:normgaussian} and Lemma~\ref{lem:normgrad}, and in the last step we used Fact~\ref{fact:cor_id} and triangle inequality to bound $\Var[p] = O(1)$.

	For the first factor in \eqref{eq:sigmasquared_main}, we have that \begin{align}
		\E\left[(F_x(\Theta) - F_x(\Theta^*))^{4q}\right]^{1/2q} &\le (2q - 1)^{d}\cdot \E\left[(F_x(\Theta) - F_x(\Theta^*))^4\right]^{1/2} \\
		&\le O(qdr^3)^{d+1}\cdot \left(\norm{V - V^*}_F + \norm{\vec{c} - \vec{c}^*}_2\right)^2
	\end{align} by Fact~\ref{fact:hypercontractivity} and Lemma~\ref{lem:main_taylorfourth} respectively, from which the claim follows.
\end{proof}

As a result, the random variable $\sigma^2$ enjoys sub-Weibull-type concentration.

\begin{corollary}\label{cor:sigmaconc}
	For any $0<\delta'<1$, let $\tau = \Omega(\ln(1/\delta'))^{d+2}$. Then
	\begin{equation}\label{eq:large_dev_sigma}
		\Pr\left[\sigma^2 \ge \tau\cdot \Omega(n)\cdot \Omega(dr^3)^{d+2}\cdot (\norm{V - V^*}_F + \norm{\vec{c} - \vec{c}^*}_2)^2\right] \le \delta'.
	\end{equation}
\end{corollary}

\begin{proof}
	Let $\gamma\triangleq n\cdot O(rd^3)^{d+2}\cdot (\norm{V - V^*}_F + \norm{\vec{c} - \vec{c}^*}_2)^2$. We wish to apply Lemma~\ref{lem:polynomial_concentration} to $\sigma^2$, which is a degree-$4d$ polynomial in $x$. By Lemma~\ref{lem:sigsquared} above, $\E[\sigma^2] \le O(\gamma)$ and $\Var[\sigma^2] \le \E[\sigma^4] \le O(\gamma^2)$. By Lemma~\ref{lem:polynomial_concentration} specialized to $T = 1$, \begin{equation}
		\Pr\left[\sigma^2 \ge O(\log(1/\delta'))^{2d}\cdot \gamma\right] \le \delta',
	\end{equation} from which the lemma follows.
\end{proof}

From \eqref{eq:DeltaVnorm} we conclude that for any $0<\delta<1$, \begin{equation}\label{eq:deltabound}
	\norm{\DeltaV}_F \le 3\sqrt{r}\cdot \etav\cdot O(\ln(1/\delta))^{(d+2)/2}\cdot O(\sqrt{n})\cdot O(dr^3)^{(d+2)/2}\cdot (\norm{V - V^*}_F + \norm{\vec{c} - \vec{c}^*}_2)
\end{equation} with probability at least $1 - \delta$.

Now consider the sequence of iterates $\{\Theta^{(t)}\}_{0\le t\le T}$ in \textsc{SubspaceDescent}. In this subsection alone, for convenience define \begin{equation}
	\alpha\triangleq 3\sqrt{r}\cdot \etav\cdot O(\ln(1/\delta))^{(d+2)/2}\cdot O(\sqrt{n})\cdot O(dr^3)^{(d+2)/2}
\end{equation} For every $0\le t< T$, let $\calE_t$ be the event that \eqref{eq:deltabound} holds for $\DeltaV^{\Theta^{(t)},x^t}$, that is, that $\norm{\DeltaV^{\Theta^{(t)},x^t}}_F \le \alpha (\norm{V^{(t)} - V^*}_F + \norm{\vec{c} - \vec{c}^*}_2)$. If $\calE_t$ held for every $t$, then by triangle inequality and induction, we would have that for every $0\le t< T$, \begin{align}
	\norm{\DeltaV^{\Theta^{(t)},x^t}}_F &\le \alpha\left(\norm{V^{(0)} - V^*}_F + \norm{\vec{c} - \vec{c}^*}_2 + \sum^{t-1}_{s=0}\norm{\DeltaV^{\Theta^{(s)},x^s}}_F\right) \\
	&\le \alpha(1 + \alpha)^{t}\left(\norm{V^{(0)} - V^*}_F + \norm{\vec{c} - \vec{c}^*}_2\right) \\
	&= \alpha(1 + \alpha)^{t}\left(\procr{V^{(0)},V^*} + \norm{\vec{c} - \vec{c}^*}_2\right) \\
	&\le 3\alpha(1 + \alpha)^{t}\cdot \procr{V^{(0)},V^*}, \label{eq:recursion}
\end{align} where the last step follows by \eqref{eq:ratio_bound}. So \begin{equation}\label{eq:unwrap_recursion}
	\sum^{T-1}_{t=0}\norm{\DeltaV^{\Theta^{(t)},x^t}}_F \le 3\left((1 + \alpha)^T - 1\right)\cdot \procr{V^{(0)},V^*}.
\end{equation}
% \begin{equation}
% 	\norm{V^{(t)} - V^*}_F \ge (1 - \alpha)^{t}\norm{V^{(0)} - V^*}_F - \left(1 - (1-\alpha)^{t}\right) \norm{\vec{c} - \vec{c}^*}_2.
% \end{equation}
Taking $\delta'$ in Corollary~\ref{cor:sigmaconc} to be $\delta/T$ and applying a union bound, we deduce by monotonicity of $L_p$ norms that Lemma~\ref{lem:local_smoothness} holds for our choice of $\etav,T$. We also deduce the following crude bound.

\begin{lemma}\label{lem:safe}
$\norm{V^{(t)} - V^*}_F \in [0.9,1.1]\cdot \procr{V^{(0)},V^*}$ for every $0\le t\le T$ with probability at least $1 - \delta$.
\end{lemma}

This modest level of control over how much the distance to the true subspace fluctuates over the course of \textsc{SubspaceDescent} will be sufficient for our subsequent analysis.

We pause to note that the assumption that the ``misspecification error'' $\norm{\vec{c} -\vec{c}^*}_2$ incurred by the coefficients $\vec{c}$ must, by \eqref{eq:ratio_bound}, be small relative to the subspace distance error incurred by the initial subspace $V^{(0)}$ is crucial here. Indeed, our bounds for the moments of $\sigma^2$, i.e. the moments of the size of the gradient steps, inherently scale with $\norm{\vec{c} - \vec{c}^*}$, yet we need local smoothness in the sense that the gradient steps have norm comparable to $\procr{V^{(0)},V^*}$.

Lastly, it will be useful to establish bounds on the moments of $\norm{V^{(t)} - V^*}_F$ for each $t$.

\begin{lemma}\label{lem:safe_exp}
	For any absolute, integer-valued constant $q\ge 1$, $\E\left[\norm{V^{(t)} - V^*}^q_F\right] \le 1.1\cdot \procr{V^{(0)},V^*}^q$ for every $0\le t<T$, where the expectation is in the randomness of the samples $x^0,...,x^{T-1}$ drawn in \textsc{SubspaceDescent}.
\end{lemma}

We defer the proof of this to Appendix~\ref{app:safe_exp}.

\subsection{Local Curvature}

We begin by outlining our argument for proving Lemma~\ref{lem:local_curvature}. As with the proof of Lemma~\ref{lem:local_curvature_coef} for \textsc{RealignPolynomial}, it will be helpful to first decompose $\langle \DeltaV,V-V^*\rangle$ into ``dominant'' and ``non-dominant'' terms. Here the ``non-dominant'' terms will be more complicated because of the trigonometric corrections associated with geodesic gradient descent.

\begin{proposition}\label{prop:trivial_decompose}
	For any $\Theta,x$, define \begin{equation}
		{\DeltaV'}^{\Theta,x} \triangleq -2\etav\cdot \langle \grad{F_x}{\Theta},\Theta^* - \Theta\rangle \cdot \Pi^{\perp}_V \cdot x\cdot (\nabla^{\Theta,x})^{\top} \ \ \ \text{and} \ \ \ {\DeltaV''}^{\Theta,x} \triangleq -2\etav\cdot \residual^{\Theta,x}\cdot \Pi^{\perp}_V\cdot x\cdot (\nabla^{\Theta,x})^{\top}
	\end{equation} and also \begin{align}
		\trigE^{\Theta,x} &\triangleq \DeltaV^{\Theta,x} - {\DeltaV'}^{\Theta,x} -{\DeltaV''}^{\Theta,x}  \\
		&= \left(\cos(\sigma^{\Theta,x}\etav) - 1\right)V\cdot \hatnab^{\Theta,x}(\hatnab^{\Theta,x})^{\top} + \left(\sin(\sigma^{\Theta,x}\etav) - \sigma^{\Theta,x}\etav\right)\hath^{\Theta,x}(\hatnab^{\Theta,x})^{\top}.
	\end{align} Then $\DeltaV^{\Theta,x} = {\DeltaV'}^{\Theta,x} + {\DeltaV''}^{\Theta,x} = \trigE^{\Theta,x}$.
\end{proposition}

\begin{proof}
	$\tildeDeltaV^{\Theta,x}\triangleq {\DeltaV'}^{\Theta,x} + {\DeltaV''}^{\Theta,x}$ is the lowest-order term in the Taylor expansion of $\DeltaV^{\Theta,x}$ around $\etav = 0$, given by \begin{equation}
		\tildeDeltaV^{\Theta,x} \triangleq \etav\cdot h^{\Theta,x}(\nabla^{\Theta,x})^{\top}.
	\end{equation} Recalling the factor $F_x(\Theta) - F_x(\Theta^*)$ in the definition of $h$ in \eqref{eq:defhnabla}, we Taylor expand around $\Theta^* = \Theta$ to get \eqref{eq:taylor_loss} from Section~\ref{sec:realignpoly} and therefore the decomposition of $\tildeDeltaV^{\Theta,x}$ into ${\DeltaV'}^{\Theta,x}$ and ${\DeltaV''}^{\Theta,x}$.
\end{proof}

$\widehat{\Delta}_v'$ Motivated by Proposition~\ref{prop:trivial_decompose}, for any $x\in\R^n$ and $\Theta = (\vec{c},V)$ define \begin{equation}\domvec^{\Theta,x} \triangleq \iprod{(\tildeDeltaV')^{\Theta,x}}{V - V^*}, \ \ \ \ \ \Evecone^{\Theta,x}\triangleq \iprod{(\tildeDeltaV'')^{\Theta,x}}{V - V^*}, \ \ \ \ \ \Evectwo^{\Theta,x}\triangleq \iprod{\trigE^{\Theta,x}}{V - V^*}.
\end{equation}

Consider a sequence of iid samples $(x^0,y^0),...,(x^{T-1},y^{T-1})\sim\calD$ and iterates $\Theta^{(0)},...,\Theta^{(T-1)}$ in the execution of \textsc{SubspaceDescent}, where each $\Theta^{(t)}$ is given by $\Theta^{(t)} = (\vec{c},V^{(t)})$. To show Lemma~\ref{lem:local_curvature}, we will show that the random variable $\sum^{T-1}_{t=0}\domvec^{\Theta^{(t)},x^{t}}$ is large with high probability, while the random variables $\sum^{T-1}_{t=0}\Evecone^{\Theta^{(t)},x^{t}}$, and $\sum^{T-1}_{t=0}\Evectwo^{\Theta^{(t)},x^{t}}$ are negligible with high probability. Eventually, we will invoke the martingale concentration inequalities of Lemmas~\ref{lem:martingale1_polynomial} and \ref{lem:martingale2} to control them. Before that, we first need to compute their expectations.

\subsubsection{Local Curvature in Expectation- Single Step}
\label{subsubsec:subspace_expectation}

In this section we give bounds on the \emph{expected} correlation between the direction in which we would like to move, and a step taken in a \emph{single} iteration in \textsc{SubspaceDescent}.

Given an iterate $\Theta = (\vec{c},V)$, let $\mu_{\domvec}(\Theta), \mu_{\Evecone}(\Theta), \mu_{\Evectwo}(\Theta)$ be the expectations $\E[\domvec^{\Theta,x}]$, $\E[\Evecone^{\Theta,x}]$, $\E[\Evectwo^{\Theta,x}]$ with respect to $x\sim\N(0,\Id_n)$. In this section we will bound these quantities in terms of the distance between $\Theta$ and $(\vec{c}^*,V^*)$. As usual, we will omit the superscript $\Theta,x$ when the context is clear.

\begin{lemma}\label{lem:domvec_term_curvature_expectation}
	$\mu_{\domvec}(\Theta)\ge 2\etav\cdot (\condnumber/4)\cdot \procr{V,V^*}^2$.
\end{lemma}

\begin{lemma}\label{lem:taylor_term_curvature_expectation}
	\begin{equation}\left|\mu_{\Evecone}(\Theta)\right|\le O(\etav)\cdot O(dr^3)^{(d+1)/2}\cdot \norm{V - V^*}_F\cdot \procr{V,V^*}\cdot \left(\norm{V - V^*}_F + \norm{\vec{c} - \vec{c}^*}_2\right).\end{equation}
\end{lemma}

\begin{lemma}\label{lem:trig_term_curvature_expectation}
	If $\etav \le O(1/n)$, then \begin{equation}\left|\mu_{\Evectwo}(\Theta)\right| \le O(\etav)\cdot O(dr^3)^{d+2}\cdot \norm{V-V^*}_F\cdot \left(\norm{V-V^*}_F + \norm{\vec{c} - \vec{c}^*}_2\right)^2.\end{equation}
\end{lemma}

At this point we pause to emphasize that Lemma~\ref{lem:domvec_term_curvature_expectation} is the key reason why we must work with $\G{n}{r}$ and not simply with the Euclidean space of $n\times r$ matrices, as Lemma~\ref{lem:domvec_term_curvature_expectation} says that the local curvature with respect to the empirical risk in a neighborhood of a subspace $V$ is dictated solely by its Procrustes distance to $V^*$ rather than by $\norm{V - V^*}_F$.

Additionally, note that once again, \eqref{eq:ratio_bound} is essential here, to ensure that the expectations from Lemmas~\ref{lem:taylor_term_curvature_expectation} and \ref{lem:trig_term_curvature_expectation} of the ``non-dominant'' terms do not overwhelm the expectation from Lemma~\ref{lem:domvec_conc} of the ``dominant'' term, which only depends on $\procr{V,V^*} \sim \procr{V^{(0)},V^*}$.

We now turn to proving Lemma~\ref{lem:domvec_term_curvature_expectation}.

\begin{proof}[Proof of Lemma~\ref{lem:domvec_term_curvature_expectation}]
	Fix a sample $(x,y)\sim\calD$. We have that \begin{align}
		\langle \tildeDeltaV', V - V^*\rangle &= -2\etav \langle \grad{F_x}{\Theta},\Theta^* - \Theta\rangle\cdot x^{\top}\Pi^{\perp}_V(V - V^*)\nabla \\
		&= 2\etav \langle \grad{F_x}{\Theta},\Theta^* - \Theta\rangle \cdot x^{\top}\cdot \Pi^{\perp}_VV^*\cdot \nabla \label{eq:newtaylor}
	\end{align} 

	By \eqref{eq:decompose_corr} we see that \eqref{eq:newtaylor} is given by $2\etav$ times \begin{equation}\label{eq:Vprimecor_lowest}
		\underbrace{\left(x^{\top}\Pi^{\perp}_V V^*\nabla\right)^2}_{\circled{A'}} + 
		\underbrace{\left(x^{\top}\Pi_V(V^* - V)\nabla\right)\cdot\left(x^{\top} \Pi^{\perp}_V V^*\nabla\right)}_{\circled{B'}} +
		\underbrace{\polydiff(V^{\top}x)\cdot\left(x^{\top}\Pi^{\perp}_V V^*\nabla\right)}_{\circled{C'}}.
	\end{equation}

	As in the proof of Lemma~\ref{lem:domcoef_term_curvature_expectation}, note that $x^{\top} \Pi_V$ and $x^{\top} \Pi^{\perp}_V$ are independent Gaussan random vectors with mean zero and covariances $\Pi_V$ and $\Pi^{\perp}_V$ respectively. So we immediately conclude that

	\begin{observation}\label{obs:BCvanish}
	 	For any $V$, the expectations of $\circled{B'}$ and $\circled{C'}$ with respect to $x$ vanish.
	\end{observation} 

	We next bound $\E[\circled{A'}]$.

	\begin{lemma}\label{lem:Aexp}
		$(\condnumber/4)\cdot \procr{V,V^*}^2 \le \E[\circled{A'}] \le 4\procr{V,V^*}^2$.
	\end{lemma}

	\begin{proof}
		Note that \begin{align}
			\E[\circled{A'}] &= \E\left[\left(x^{\top}\Pi^{\perp}_VV^*\nabla\right)^2\right] \\
			&= \E_{\substack{h\sim\N(0,\Pi_V) \\ h_{\perp}\sim\N(0,\Pi^{\perp}_V)}}\left[\grad{p}{V^{\top} h}^{\top}{V^*}^{\top}h_{\perp}h_{\perp}^{\top}V^*\grad{p}{V^{\top} h}\right] \\
			&= \E_{h\sim\N(0,\Pi_V)}\left[\grad{p}{V^{\top} h}^{\top}{V^*}^{\top}\Pi^{\perp}_V V^* \grad{p}{V^{\top} h} \right] \\
			&= \E_{g\sim\N(0,\Id_r)}\left[\grad{p}{g}^{\top}\cdot \left(\Id - {V^*}^{\top}VV^{\top}V^*\right)\cdot \grad{p}{g}\right]\\
			&= \left\langle \E_g\left[\grad{p}{g}\grad{p}{g}^{\top}\right], \Id - {V^*}^{\top}VV^{\top}V^*\right\rangle\label{eq:Aexp}
			% &\ge \condnumber\cdot \Tr(\Id - {V^*}^{\top}VV^{\top}V^*)\\
			% &= \condnumber\cdot \chord{V,V^*}^2
		\end{align} where we used independence of $h,h_{\perp}$ in the third step. We wil need the following bound.

		\begin{lemma}\label{lem:identifiability2}
			If $\norm{\vec{c} - \vec{c}^*}_2 \le O(r^{-3/2}d^{-1})$, then we have that
			\begin{equation}\label{eq:identifiability2}
				({\condnumber}/{2})\cdot \Id_r \preceq \E_{g\sim\N(0,\Id_r)}\left[\grad{p}{g}\grad{p}{g}^{\top}\right] \preceq 2\cdot \Id_r.
			\end{equation}
		\end{lemma}

		\begin{proof}
			For convenience, let $M$ and $M_*$ denote $\E\left[\grad{p}{g}\grad{p}{g}^{\top}\right]$ and $\E_{g\sim\N(0,\Id_r)}\left[\grad{p_*}{g}\grad{p_*}{g}^{\top}\right]$ respectively. For any $v\in\S^{r-1}$, we have that \begin{align}
				|v^{\top}M_*v - v^{\top}Mv| &= \left|\E\left[\langle v,\grad{p_*}{g}\rangle^2 - \langle v,\grad{p}{g}\rangle^2\right]\right| \\
				&= \left|\E\left[\langle v,\grad{\polydiff}{g}\rangle\cdot\langle v,\grad{(p+p_*)}{g}\rangle\right]\right| \\
				&\le \E\left[\norm{\grad{\polydiff}{g}}^2_2\right]^{1/2}\cdot \left(\E\left[\norm{\grad{p}{g}}^2_2\right]^{1/2} + \E\left[\norm{\grad{p_*}{g}}^2_2\right]^{1/2}\right) \\
				&\le rd\cdot \Var[\polydiff]^{1/2}\cdot(\Var[p]^{1/2} + \Var[p_*]^{1/2}) \\
				&< O(r^{3/2}d\cdot \norm{\vec{c} - \vec{c}^*}_2),
			\end{align} where in the third step we used Cauchy-Schwarz, in the fourth step we used Lemma~\ref{lem:normgrad}, and in the last step we upper bounded $\Var[p]$ and $\Var[p_*]$ by $O(r)$ using Corollary~\ref{fact:cor_id} and the fact that $\norm{\vec{c} - \vec{c}^*}_2 = O(1)$.
		\end{proof}

		To conclude the proof of Lemma~\ref{lem:Aexp}, we see that \begin{align}
			\E[\circled{A'}] &\in [\condnumber/2,2]\cdot \Tr(\Id - {V^*}^{\top}VV^{\top}V^*) \\
			&= [\condnumber/2,2]\cdot \chord{V,V^*}^2 \\
			&\in [\condnumber/4,4]\cdot \procr{V,V^*}^2,\label{eq:circledAbound}
		\end{align}
		where the first step follows by \eqref{eq:Aexp} and Lemma~\ref{lem:identifiability2}, the second step follows by the fact that $\Tr(\Id - {V^*}^{\top}VV^{\top}V^*) = d - \norm{{V^*}^{\top}V}^2_F$, and the last step follows by Lemma~\ref{lem:procrustes_chordal}.
	\end{proof}

	Lemma~\ref{lem:domvec_term_curvature_expectation} now follows from \eqref{eq:Vprimecor_lowest}, Observation~\ref{obs:BCvanish}, and Lemma~\ref{lem:Aexp}.
\end{proof}

We defer the proofs of Lemmas~\ref{lem:taylor_term_curvature_expectation} and \ref{lem:trig_term_curvature_expectation}, to Appendix~\ref{app:subspace}.

% So \eqref{eq:higher_order_bound}, together with the bounds in Section~\ref{subsec:firstorder}, implies that \begin{equation}
% 	\frac{1}{2\eta}\E\left[\langle \tildeDeltaV, V - V^*\rangle\right] \ge (\condnumber/4 - (dr)^{O(d)}\cdot \procr{V,V^*})\cdot d(\Theta,\Theta^*)^2.
% \end{equation} \assumption{By taking $\procr{V,V^*} \le (\condnumber/16)\cdot (dr)^{-\Omega(d)}$}, we get that \begin{equation}\E\left[\langle \tildeDeltaV, V - V^*\rangle\right] \ge (3\eta\cdot \condnumber/8)\cdot d(\Theta,\Theta^*)^2\label{eq:tildeVcor}.\end{equation}

% \subsubsection{Trigonometric Terms in Expectation}
% \label{subsubsec:trig_subspace}

% This, together with \eqref{eq:tildeVcor}, completes the proof that \begin{equation}
% 	\E\left[\langle \DeltaV, V-V^*\rangle\right] \ge \eta\cdot(\condnumber/4)\cdot d(\Theta,\Theta^*)^2, \label{eq:curvature_in_expectation}
% \end{equation} that is, local curvature holds in expectation.

\subsubsection{Local Curvature in Expectation- All Iterations}
\label{sec:local_curve_all_iters}

In this section we extend the results of the previous section to give bounds on the sum \emph{over all $t$} of the expected correlations between the direction in which we would like to move at time $t$, and the step we actually take at time $t$.

Specifically, for the sequence of iterates $\{\Theta^{(t)}\}_{0\le t \le T}$ in \textsc{SubspaceDescent}, we would like to bound $\E\left[\sum^{T-1}_{t=0}\mu_{\domvec}(\Theta^{(t)})\right]$, $\left|\E\left[\sum^{T-1}_{t=0}\mu_{\Evecone}(\Theta^{(t)})\right]\right|$, and $\left|\E\left[\sum^{T-1}_{t=0}\mu_{\Evectwo}(\Theta^{(t)})\right]\right|$. We emphasize that the expectation here is over the randomness of the samples $x^0,...,x^{T-1}$, so e.g. $\mu_{\domvec}(\Theta^{(t)})$ is a random variable depending on $x^0,...,x^{t-1}$ and is itself an expectation over the next sample $x^t$.

Intuitively, for our choice \eqref{eq:main_etav_assumption} of small step size $\etav$ which scales with $O(1/T)$, Lemma~\ref{lem:safe_exp} suggests that the expected behavior of the corresponding martingales should not be very different from that of a sum of iid random variables. That is, these expected sums should be not much different than $T$ times the expectation of their \emph{first} summand, corresponding to the first iteration which takes a step from $\Theta^{(0)}$. In Lemmas~\ref{lem:sumdomvecs}, \ref{lem:sumEvecones}, and \ref{lem:sumEvectwos}, we show that this is indeed the case:

\begin{lemma}\label{lem:sumdomvecs}
	$\E\left[\sum^{T-1}_{t=0}\mu_{\domvec}(\Theta^{(t)})\right] \le T\cdot \etav\cdot (\condnumber/2.2)\cdot \procr{V^{(0)},V^*}^2$.
\end{lemma}

\begin{lemma}\label{lem:sumEvecones}
$\E\left[\sum^{T-1}_{t=0}\mu_{\Evecone}(\Theta^{(t)})\right] \le T\cdot O(\etav)\cdot O(dr^3)^{(d+1)/2}\cdot \procr{V^{(0)},V^*}^3$.
\end{lemma}

\begin{lemma}\label{lem:sumEvectwos}
	$\E\left[\sum^{T-1}_{t=0}\mu_{\Evectwo}(\Theta^{(t)})\right] \le T\cdot O(\etav)\cdot O(dr^3)^{d+2}\cdot \procr{V^{(0)},V^*}^3$.
\end{lemma}

We defer their proofs to Appendices~\ref{app:sumdomvecs}, \ref{app:sumEvecones}, and \ref{app:sumEvectwos} respectively.

\subsubsection{Local Curvature with High Probability}

In this section, we complete the proof of Lemma~\ref{lem:local_curvature} by establishing high-probability bounds for the MDS's corresponding to $\domvec$, $\Evecone$, and $\Evectwo$. That is, we argue that with high probability, the dominant term given by $\domvec$ is large, while the error terms from Taylor approximation and from the trigonometric corrections are small. Specifically, we show:

\begin{lemma}\label{lem:domvec_conc}
	\begin{equation}
		\sum^{T-1}_{t=0}\domvec^{\Theta^{(t)},x^t} \ge T\cdot \etav \cdot (\condnumber/3)\cdot \procr{V^{(0)},V^*}^2
	\end{equation} with probabiliy at least $1 - \delta$.
\end{lemma}

\begin{lemma}\label{lem:Evecone_conc}
% For any $\delta>0$, if $T = \Omega\left((\log(1/\delta)\cdot d)^{\Cr{weibull}d}\right)$ and \begin{equation}\procr{V^{(0)},V^*}\le (\condnumber/1000)\cdot O(dr^3)^{-(d+1)/2}\cdot \rho,\end{equation} 
	\begin{equation}
		\left|\sum^{T-1}_{t=0}\Evecone^{\Theta^{(t)},x^t}\right| \le T\cdot \etav\cdot (\Cr{tiny}\cdot \condnumber)\cdot\procr{V^{(0)},V^*}^2
	\end{equation} with probability at least $1 - \delta$.
\end{lemma}

\begin{lemma}\label{lem:Evectwo_conc}
	% For any $\delta>0$, if $T = \Omega\left((\log(1/\delta)\cdot d)^{\Cr{weibull}d}\right)$ and \begin{equation}\procr{V^{(0)},V^*}\le (\condnumber/1000)\cdot O(dr^3)^{-d-2}\cdot \rho^2,\end{equation} then we have that
	\begin{equation}
		\left|\sum^{T-1}_{t=0}\Evectwo^{\Theta^{(t)},x^t}\right| \le T\cdot \etav\cdot (\Cr{tiny}\cdot \condnumber)\cdot\procr{V^{(0)},V^*}^2
	\end{equation} with probability at least $1 - \delta$.
\end{lemma}

We defer their proofs to Appendix~\ref{app:Evec_conc}. The key technical step in all three proofs is to upper bound the variance of the martingale differences, after which one can invoke the corresponding expectation bounds from Section~\ref{sec:local_curve_all_iters} together with the martingale concentration inequalities of Lemma~\ref{lem:martingale2} for Lemma~\ref{app:domvec_conc} and Lemma~\ref{lem:martingale1_polynomial} for Lemmas~\ref{lem:Evecone_conc} and \ref{lem:Evectwo_conc}. We emphasize that here we must again crucially use \eqref{eq:ratio_bound}, this time to ensure that the variances of the martingale differences, which depend in part on $\norm{\vec{c} -\vec{c}^*}_2$, do not swamp the expectation $\mu_{\domvec}(\Theta)$ of the dominant term.

Also, we remark that it is in the proof of Lemma~\ref{lem:Evecone_conc} and Lemma~\ref{lem:Evectwo_conc} that we finally use the assumption \eqref{eq:V0warmstart} that $\procr{V^{(0)},V^*}$ is somewhat small.

Finally, we can deduce Lemma~\ref{lem:local_curvature}, completing the proof of Theorem~\ref{thm:subspacedescent_guarantee}.

\begin{proof}[Proof of Lemma~\ref{lem:local_curvature}]
	By Lemmas~\ref{lem:domvec_conc}, \ref{lem:Evecone_conc}, and \ref{lem:Evectwo_conc},  and the earlier calculation showing that for any $\Theta= (\vec{c},V)$, $\langle \DeltaV^{\Theta,x}, V-V^*\rangle = \domvec^{\Theta,x} + \Evecone^{\Theta,x} + \Evectwo^{\Theta,x}$, we see that under our choice of $T,\etav$,
	\begin{equation}
		 \sum^{T-1}_{t=0} \left\langle\DeltaV^{\Theta^{(t)},x^t}, V^{(t)}-V^*\right\rangle \ge \condnumber\left(\frac{1}{3} - 2\Cr{tiny}\right)\cdot T\cdot\etav\cdot \procr{V^{(0)},V^*}^2
	\end{equation} with probability $1 - 3\delta$. By replacing $3\delta$ with $\delta$, and absorbing the constant factors, the lemma follows.
\end{proof}

%!TEX root = ./poly_main.tex

\section{Putting Everything Together for {\mdseries\textsc{GeoSGD}}}
\label{sec:finish}

In this section we conclude the proof of Theorem~\ref{thm:boost} using Theorems~\ref{thm:realign_guarantee} and \ref{thm:subspacedescent_guarantee}.

There is one last subtlety we must address.
In Theorem~\ref{thm:realign_guarantee} on the distance $\norm{\vec{c}-\vec{c}^*}$ between the coefficients $\vec{c}$ output by \textsc{RealignPolynomial} and the true coefficients $\vec{c}^*$, the upper bound is at best only in terms of the known parameter $\underline{\epsilon}$.
On the other hand, in Theorem~\ref{thm:subspacedescent_guarantee} on the error $\procr{V^{(T)},V^*}$ incurred by the subspace $V^{(T)}$ output by \textsc{SubspaceDescent} when initialized to $V^{(0)}$, the upper bound we can show only applies when \eqref{eq:ratio_bound} holds.

The scenario that these guarantees do not account for is when at some point in the middle of \textsc{GeoSGD}, we arrive upon a subspace $V^{(0)}$ for which $\procr{V^{(0)},V^*} \ll \underline{\epsilon}/2$, in which case running \textsc{RealignPolynomial} with $V^{(0)}$ gives coefficients $\vec{c}$ for which \eqref{eq:ratio_bound} fails to hold. Intuitively, this should be fine because $\procr{V^{(0)},V^*} < \underline{\epsilon}$, so \textsc{GeoSGD} has already produced a good enough estimate for the true subspace and we could just terminate.
Unfortunately, it is not immediately obvious how to tell when this has happened and terminate accordingly.

Instead, we argue that local smoothness for \textsc{SubspaceDescent} (Lemma~\ref{lem:local_smoothness}), implies that in this case, running \textsc{SubspaceDescent} initialized to $V^{(0)}$ will produce a subspace $V^{(T)}$ whose error is still good enough:

\begin{lemma}\label{lem:still_fine}
	Suppose all of the assumptions of Theorem~\ref{thm:subspacedescent_guarantee} hold except for \eqref{eq:ratio_bound}. Then we still have that $\procr{V^{(T)},V^*} \le \norm{\vec{c} - \vec{c}^*}_2$ with probability at least $1 - \delta$.
\end{lemma}

\begin{proof}
	Suppose the event of Lemma~\ref{lem:local_smoothness} occurs. We have that \begin{align}
		\procr{V^{(T)},V^*} &\le \procr{V^{(0)},V^*} + \procr{V^{(0)},V^{(T)}} \\
		&\le \procr{V^{(0)},V^*}\cdot \left(1 + \etav \cdot O(dr^3\ln(T/\delta))^{(d+2)/2}\cdot O(\sqrt{n})\right) \\
		&\le \frac{1}{2}\norm{\vec{c} - \vec{c}^*}_2\cdot\left(1 + \etav \cdot O(dr^3\ln(T/\delta))^{(d+2)/2}\cdot O(\sqrt{n})\right) \\
		&= \frac{1}{2}\norm{\vec{c} - \vec{c}^*}\cdot \left(1 + O\left(\frac{\condnumber}{T\sqrt{n}}\right)\right) \\
		&< \norm{\vec{c} - \vec{c}^*}_2,
	\end{align} 
	where the first step follows by triangle inequality for Procrustes distance (Fact~\ref{fact:triangle}), the second by the assumption that the event of Lemma~\ref{lem:local_smoothness} holds, the third by the assumption that \eqref{eq:ratio_bound} does not hold, and the fourth by the definition of $\etav$ in \eqref{eq:main_etav_assumption}.
\end{proof}

We can now complete the proof of Theorem~\ref{thm:boost}.

\begin{proof}[Proof of Theorem~\ref{thm:boost}]
	Let $\vec{c}^{(t)}$ and $V^{(t)}$ be the iterates of \textsc{GeoSGD}. Suppose for $0\le t < T$ we had $\procr{V^{(t)},V^*}\le \Cr{tiny}\cdot \condnumber \cdot O(dr^3)^{-d-2}$. By Theorem~\ref{thm:realign_guarantee}, we have that \begin{equation}
		\norm{\vec{c}^{(t+1)} - \vec{c}^*}_2 < 2\cdot \Max{\epsilon/2}{\procr{V^{(t)},V^*}}.
	\end{equation} If $\norm{\vec{c}^{(t+1)} - \vec{c}^*}_2 < \epsilon$, then by Lemma~\ref{lem:still_fine}, $\procr{V^{(t+1)},V^*}< \epsilon$. Otherwise, if $\norm{\vec{c}^{(t+1)} - \vec{c}^*}_2 \le 2\procr{V^{(t)},V^*}$, then \eqref{eq:ratio_bound} in Theorem~\ref{thm:subspacedescent_guarantee} holds and we get that \begin{equation}
		\procr{V^{(t+1)},V^*} \le (1 - \alpha)\cdot \procr{V^{(t)},V^*},
	\end{equation} where \begin{equation}
		\alpha\triangleq \frac{\condnumber}{n}\cdot \poly(\ln(1/\condnumber),r,d,\ln(1/\delta'))^{-d}
	\end{equation} for $\delta' = \delta/(2T+1)$ as defined in \textsc{GeoSGD}.

	In either case, $\procr{V^{(t+1)},V^*} \le \Cr{tiny}\cdot \condnumber \cdot O(dr^3)^{-d-2}$. And furthermore, if we unroll this recurrence, we conclude that \begin{equation}
		\procr{V^{(T)},V^*} \le \Max{\epsilon}{(1 - \alpha)^T\cdot \procr{V^{(0)},V^*}}.
	\end{equation} So by taking $T = \alpha^{-1}\cdot \log(1/\epsilon)$, we get that $\procr{V^{(T)},V^*} \le \epsilon$ as desired. This corresponds to the choice of $T$ in \eqref{eq:geosgd_T_bound}. Lastly, we get that $\norm{\vec{c}^{(T)} - \vec{c}}_2\le \epsilon$ by one last application of Theorem~\ref{thm:realign_guarantee}.
\end{proof}

\begin{proof}[Proof of Theorem~\ref{thm:boost_runtime}]
	This follows from the runtime and sample complexity guarantees of Theorems~\ref{thm:realign_guarantee} and \ref{thm:subspacedescent_guarantee}.
\end{proof}

\bibliographystyle{alpha}
\bibliography{biblio}

\appendix

%!TEX root = poly_main.tex

\section{Martingale Concentration Inequalities}
\label{app:martingale}

In this section we prove the two martingale concentration inequalities from Section~\ref{subsec:martingale_concentration} that are needed for the analysis of the boosting phase of our algorithm.

\subsection{Proof of Lemma~\ref{lem:martingale1_polynomial}}

% The first concentration inequality we show will be used to argue that the contributions of the terms which are non-dominant in expectation in the analysis of \textsc{GeoSGD} are non-dominant \emph{with high probability}.

We first prove the following more general statement.

\begin{lemma}\label{lem:martingale1}
	Let $\sigma> 0$ and $0< \alpha\le 2$ be constants, and let $\calE_i$ be the event that $\E[|Z_i|^q|\xi_1,...,\xi_{i-1}] \le \sigma^q\cdot q^{q/\alpha}$ for all $q\ge 1$.

	If $\Pr[\calE_i|\xi_1,...,\xi_{i-1}] \ge 1 - \beta$ for each $i\in[T]$, then for any $t > 0$, \begin{equation}\label{eq:apply_li_weibull}
		\Pr\left[\max_{\ell\in[T]}\left|\sum^{\ell}_{i=1}Z_i\right| \ge t\cdot \sqrt{T}\cdot\sigma\right] \le O\left(1 + t^2 (1/\alpha)^{O(1/\alpha)}\right)\cdot \exp\left(-\left(t^2/32\right)^{\frac{\alpha}{2+\alpha}}\right) + T\cdot\beta.
	\end{equation} In particular, there is an absolute constant $\Cr{weibull}>0$ such that for any $\delta > 0$,
	\begin{equation}
		\Pr\left[\max_{\ell\in[T]}\left|\sum^{\ell}_{i=1}Z_i\right| \ge (\log(1/\delta)/\alpha)^{2\Cr{weibull}/\alpha}\cdot \sqrt{T}\cdot\sigma\right] \le \delta + T\cdot\beta.
	\end{equation}
\end{lemma}

We first show that this implies Lemma~\ref{lem:martingale1_polynomial}.

\begin{proof}[Proof of Lemma~\ref{lem:martingale1_polynomial}]
	This is an immediate consequence of Lemma~\ref{lem:martingale1} together with Fact~\ref{fact:hypercontractivity}, which implies the requisite moment bounds for Lemma~\ref{lem:martingale1} for $\alpha = d/2$.
\end{proof}

To show Lemma~\ref{lem:martingale1}, we require the following theorem on the concentration of martingales with sub-Weibull differences, which is a consequence of the main result of \cite{li2018note}.

\begin{theorem}[\cite{li2018note}]\label{thm:weibull}
	Let $\sigma> 0$ and $0< \alpha\le 2$ be constants. Suppose that for every $i\in[T]$, we have that with probability one, $\E[|Z_i|^q|\xi_1,...,\xi_{i-1}] \le \sigma^q\cdot q^{q/\alpha}$ holds for all $q\ge 1$. Then for any $z > 0$, \begin{equation}\label{eq:li_weibull}
		\Pr\left[\max_{\ell\in[T]}\left|\sum^{\ell}_{i=1}Z_i\right| \ge t\cdot \sqrt{T}\cdot\sigma\right] \le O\left(1 + t^2 (1/\alpha)^{O(1/\alpha)}\right)\cdot \exp\left(-\left(t^2/32\right)^{\frac{\alpha}{2+\alpha}}\right)
	\end{equation}
\end{theorem}

We use a standard trick, see e.g. Lemma 3.1 of \cite{vu2002concentration}, to relax the assumption that the differences are sub-Weibull almost surely to the assumption that they are sub-Weibull with high probability. It will also be more convenient for us to state the inequality in terms of moment bounds rather than Orlicz norm bounds.

\begin{proof}[Proof of Lemma~\ref{lem:martingale1}]
	Given a realization $\xi$ of the random variables $(\xi_1,...,\xi_T)$, let $i_{\xi}$ be the first index $i$, if any, for which $\calE_i$ does not hold. Define $B_i\triangleq \{\xi: i_{\xi} = i\}$ and note that these sets are disjoint for different $i$. Let $Y'(\xi)$ be the function which agrees with $Y(\xi)$ for $\xi\in (\cup B_i)^c$ and which is equal to $\E_{B_i}[Y]$ for $\xi\in B_i$. $Y'$ and $Y$ have the same mean, so the lemma follows by union bounding over the events $\cup B_i$ together with the probability that the martingale $Y'$ fails to concentrate. For the former probabilities, by definition $\Pr[B_i]\le \beta$. And for the latter, because the martingale differences for $Y'$ satisfy the assumptions of Theorem~\ref{thm:weibull}, $Y'$ fails to concentrate with probability at most the right-hand side of \eqref{eq:li_weibull}. This yields \eqref{eq:apply_li_weibull}.
\end{proof}

\subsection{Proof of Lemma~\ref{lem:martingale2}}

% We also need a concentration result to argue that the contributions of the terms which are dominant in expectation (e.g. $\circled{A'}$) in the analysis of \textsc{GeoSGD} are dominant \emph{with high probability}.

To show Lemma~\ref{lem:martingale2}, we require the following theorem due to \cite{bentkus2003inequality}, which controls the tails of martingales whose differences are only bounded on one side.

\begin{theorem}[\cite{bentkus2003inequality}]\label{thm:bentkus}
	Let $\{c_i\}_{i\in[T]}$ and $\{s_i\}_{i\in[T]}$ be collections of positive constants for which $Z_i \le c_i$ and $\E[Z^2_i|\xi_1,...,\xi_{i-1}]\le s^2_i$ with probability one for every $i\in[T]$. Let $\sigma_i = \Max{c_i}{s_i}$, and define $\sigma^2 = \sum_i \sigma^2_i$. Then \begin{equation}
		\Pr\left[\sum^T_{i=1}Z_i \ge t\cdot \sigma\right] \le \exp(-t^2/2).
	\end{equation}
\end{theorem}

\begin{proof}[Proof of Lemma~\ref{lem:martingale2}]
	The proof is identical to that of Lemma~\ref{lem:martingale1}, except instead of applying Theorem~\ref{thm:weibull} to the auxiliary martingale, we apply Theorem~\ref{thm:bentkus} to get that for any $t > 0$, \begin{equation}
		\Pr\left[\sum^T_{i=1}Z_i \ge t\cdot \sigma\right] \le \exp(-t^2/2) + T\cdot \beta.
	\end{equation} The lemma follows by taking $t = \sqrt{2}\log(1/\delta)$.
\end{proof}

\section{Deferred Proofs from Section~\ref{sec:realignpoly}}
\label{app:realign}

\subsection{Proof of Lemma~\ref{lem:main_taylorfourth}}
\label{app:main_taylorfourth}

\begin{proof}
	\begin{align}
		\MoveEqLeft\E\left[(F_x(\Theta) - F_x(\Theta^*))^4\right] \\
		&\le \sum_{\ell_1,...,\ell_4\in[d+1]}\frac{1}{\prod^{4}_{\nu=1}\ell_{\nu}!}\E\left[\prod^{4}_{\nu = 1}\left\langle \nabla^{[\ell_{\nu}]}F_x(\Theta),(\Theta^* - \Theta)^{\otimes\ell_{\nu}}\right\rangle\right] \\
		&\le \sum_{\ell_1,...,\ell_{4}\in[d+1]}\frac{1}{\prod^{4}_{\nu=1}\ell_{\nu}!}\cdot 16\cdot(8dr^2)^{2(d+1)}\cdot \norm{V - V^*}^{\sum_{\nu}\ell_{\nu}}_F \cdot \left(1 + \frac{\norm{\vec{c} - \vec{c^*}}_2}{\norm{V^* - V}_F}\right)^{4} \\
		&\le 16\cdot(8dr^2)^{2(d+1)}\left(\sum^{d+1}_{\ell=1}\frac{1}{\ell!} \cdot \norm{V - V^*}^{\ell}_F\right)^{4} \cdot \left(1 + \frac{\norm{\vec{c} - \vec{c^*}}_2}{\norm{V^* - V}_F}\right)^{4} \\
		&\le 16\cdot(8dr^2)^{2(d+1)}\cdot (e\cdot (4r)^{d/2}\norm{V - V^*}_F)^{4} \cdot \left(1 + \frac{\norm{\vec{c} - \vec{c^*}}_2}{\norm{V^* - V}_F}\right)^{4} \\
		&\le (2e)^4\cdot(32dr^3)^{2(d+1)}\cdot\left(\norm{V - V^*}_F + \norm{\vec{c} - \vec{c}^*}_2\right)^{4},
	\end{align} where the second step follows by Lemma~\ref{lem:main_taylorterms}, the fourth by the fact that $\norm{V - V^*}_F\le 2\sqrt{r}$ and the fact that $\sum^{d+1}_{\ell=1}\frac{1}{\ell!}\cdot x^{\ell}\le e\cdot (4r)^{d/2}\cdot x$ for $x\in[0,2\sqrt{r}]$.
\end{proof}

\subsection{Proof of Lemma~\ref{lem:Ecoef_term_curvature_expectation}}
\label{app:Ecoef_term_curvature_expectation}

\begin{proof}
	We have that \begin{align}
		\frac{1}{2\etac}\left|\iprod{\Deltac''}{\vec{c}-\vec{c}^*}\right| &= \left|\E\left[\sum^{d+1}_{\ell=2}\frac{1}{\ell!}\Iprod{\nabla^{[\ell]}F_x(\Theta)}{(\Theta^*-\Theta)^{\otimes\ell}}\cdot \polydiff(V^{\top}x)\right]\right| \\
		&\le \E\left[\left(\sum^{d+1}_{\ell=2}\frac{1}{\ell!}\Iprod{\nabla^{[\ell]}F_x(\Theta)}{(\Theta^*-\Theta)^{\otimes\ell}}\right)^2\right]^{1/2}\cdot \E\left[\polydiff(V^{\top}x)^2\right]^{1/2} \\
		&\le O(dr^3)^{(d+1)/2}\cdot \norm{V - V^*}_F\cdot (\norm{V - V^*}_F + \norm{\vec{c} - \vec{c}^*}_2)\cdot \norm{\vec{c}-\vec{c}^*}_2 \\
		&= O(dr^3)^{(d+1)/2}\cdot \procr{V,V^*}\cdot \norm{\vec{c}-\vec{c}^*}_2\cdot (\procr{V,V^*} + \norm{\vec{c} - \vec{c}^*}_2),
	\end{align} where the second step follows by Cauchy-Schwarz, the third step follows by Lemma~\ref{lem:taylorsquared}, and the last step follows by the assumption that $\norm{V - V^*}_F = \procr{V,V^*}$.
\end{proof}

\subsection{Proof of Proposition~\ref{prop:helper_circledBprime}}
\label{app:helper_circledBprime}

\begin{proof}
	Note that \begin{align}
		\E_g\left[\left(x^{\top}\Pi_V(V^* - V)\grad{p}{V^{\top}x}\right)^2\right]^{1/2} &\le \norm{\Id - V^{\top}V^*}_2\cdot \E_g\left[\norm{g}^2_2 \cdot \norm{\grad{p}{g}}^2_2\right]^{1/2} \\
		&\le \procr{V,V^*}^2\cdot O(r^{3/2}d),
	\end{align} where the second step follows by the second part of Lemma~\ref{lem:sigmamax_bound}, Lemma~\ref{lem:x2grad2}, and the fact that \begin{equation}\Var[p]^{1/2} \le \norm{\vec{c} - \vec{c}^*}_2 + \Var[p^*]^{1/2} \le O(r)\end{equation} because $\norm{\vec{c} - \vec{c}^*}_2 \le 1$ by assumption and because of Corollary~\ref{fact:cor_id}.
\end{proof}

\subsection{Proof of Lemma~\ref{lem:domcoef_conc}}
\label{app:domcoef_conc}

We will split up $\frac{1}{B}\sum^{B-1}_{i=0}\domcoef^{x^i}$ according to the decomposition \eqref{eq:cprimecor_lowest}. That is, define \begin{align}
	\circled{A}^x &\triangleq \left(\polydiff(V^{\top}x)\right)^2 \\
	\circled{B}^x &\triangleq \polydiff(V^{\top}x)\cdot \left(x^{\top}\Pi_V(V^* - V)\nabla\right) \\
	\circled{C}^x &\triangleq \polydiff(V^{\top}x)\cdot \left(x^{\top}\Pi^{\perp}_V V^*\nabla\right)
\end{align}
so that for any $x$, \begin{equation}\label{eq:domcef_decomp}\frac{1}{2\etac}\domcoef^x = \circled{A}^x + \circled{B}^x + \circled{C}^x.\end{equation}
We will show concentration for these three random variables separately.

\begin{lemma}\label{lem:circledAprimeconc}
	For any $\delta>0$, if $B = \Omega(\log(1/\delta)^2\cdot 9^d)$, then \begin{equation}
		\frac{1}{B}\sum^{B-1}_{i=0}\circled{A}^{x^i}\ge \frac{1}{2}\norm{\vec{c} - \vec{c}^*}^2_2
	\end{equation} with probability at least $1 - \delta$.
\end{lemma}

\begin{lemma}\label{lem:circledBprimeconc}
	For any $\delta>0$, if $B = \Omega(\log(1/\delta))^{2d}$, then \begin{equation}
		\left|\frac{1}{B}\sum^{B-1}_{i=0}\circled{B}^{x^i}\right|\le O(r^{3/2}d)\cdot \norm{\vec{c} - \vec{c}^*}_2 \cdot \procr{V,V^*}^2
	\end{equation} with probability at least $1 - \delta$.
\end{lemma}

\begin{lemma}\label{lem:circledCprimeconc}
	For any $\delta>0$ and $\gamma>0$, if $B = \Omega(\log(1/\delta))^{2d} \cdot \gamma^{-2}$, then \begin{equation}
		\left|\frac{1}{B}\sum^{B-1}_{i=0}\circled{C}^{x^i}\right|\le \gamma\cdot \procr{V,V^*}\cdot \norm{\vec{c} - \vec{c}^*}_2
	\end{equation} with probability at least $1 - \delta$.
\end{lemma}

We prove these in the subsequent Appendices~\ref{app:circledAprimeconc}, \ref{app:circledBprimeconc}, and \ref{app:circledCprimeconc}. Note that Lemma~\ref{lem:domcoef_conc} immediately follows from these lemmas.

\begin{proof}[Proof of Lemma~\ref{lem:domcoef_conc}]
	By a union bound over the failure probabilities of Lemmas~\ref{lem:circledAprimeconc}, \ref{lem:circledBprimeconc}, and \ref{lem:circledCprimeconc}, we see by triangle inequality and \eqref{eq:domcef_decomp} that \begin{equation}
		\frac{1}{B}\sum^{B-1}_{i=0}\domcoef^{x^i} \ge 2\etac \cdot \left(\frac{1}{2}\norm{\vec{c}-\vec{c}^*}^2_2 - O(r^{3/2}d)\cdot \norm{\vec{c} - \vec{c}^*}_2 \cdot \procr{V,V^*}^2 - \gamma\cdot\procr{V,V^*}\cdot\norm{\vec{c} - \vec{c}^*}_2\right)
	\end{equation} with probability at least $1 - 3\delta$, provided $B = \Omega(\log(1/\delta))^d\cdot \gamma^{-2}$. The result follows by replacing $3\delta$ with $\delta$ and absorbing constants.
\end{proof}

\subsubsection{Proof of Lemma~\ref{lem:circledAprimeconc}}
\label{app:circledAprimeconc}

\begin{proof}
	Observe that $\frac{1}{B}\sum^{B-1}_{i=0}\left(\E_{x}[\circled{A}^x] - \circled{A}^{x^i}\right)$ is an average of $B$ iid copies of a mean-zero random variable satisfying one-sided bounds, so we wish to apply Lemma~\ref{lem:one_sided}.

	To do so, we just need to bound the variances of the summands.

	\begin{lemma}\label{lem:Aprimevar}
		$\Var_x[\circled{A}^x] \le 9^d\cdot \norm{\vec{c} - \vec{c}^*}^4_2$.
	\end{lemma}

	\begin{proof}
		Clearly $\Var[\circled{A}^x] \le \E[(\circled{A}^x)^2]$, so it suffices to bound the latter. By Fact~\ref{fact:hypercontractivity} applied to the degree-$d$ polynomial $\polydiff$, \begin{equation}\label{eq:circAprime_squared}
			\E[(\circled{A}^x)^2] = \E_{g\sim\N(0,\Id_r)}[\polydiff(g)^4] \le 9^d \cdot \E[\polydiff(g)^2]^2 = 9^d\cdot\norm{\vec{c} - \vec{c}^*}^4_2
		\end{equation} as claimed.
	\end{proof}
	
	We can now complete the proof of Lemma~\ref{lem:circledAprimeconc}. 

	By Lemma~\ref{lem:one_sided}, Observation~\ref{obs:Aprime_exp}, and Lemma~\ref{lem:Aprimevar}, \begin{equation}
		\frac{1}{B}\sum^{B-1}_{i=0}\circled{A}^{x^i} \ge \norm{\vec{c} - \vec{c}^*}^2_2 - \frac{1}{\sqrt{B}}\cdot \sqrt{2}\log(1/\delta)\cdot 3^d \cdot \norm{\vec{c} - \vec{c}^*}^2
	\end{equation} with probability at least $1 - \delta$. The lemma follows by taking $B = \Omega(\log(1/\delta)^2$.
\end{proof}

\subsubsection{Proof of Lemma~\ref{lem:circledBprimeconc}}
\label{app:circledBprimeconc}

\begin{proof}
	Note that $\circled{B}^x$ is a polynomial of degree $2d$ in $x$, so by Lemma~\ref{lem:polynomial_concentration}, we just need to upper bound its variance.

	\begin{lemma}\label{lem:Bprimevar}
		$\Var_x[\circled{B}^x]\le 9^d \cdot O(r^{3/2}d)\cdot \norm{\vec{c} - \vec{c}^*}^2_2 \cdot \procr{V,V^*}^4$.
	\end{lemma}

	\begin{proof}
		We will upper bound $\E_x[(\circled{B}^x)^2]$ via \begin{align}
			\E\left[\circled{B}^2\right] &\le \E\left[\polydiff(V^{\top}x)^4\right]^{1/2} \cdot \E\left[\left(x^{\top}\Pi_V(V^* - V)\nabla\right)^4\right]^{1/2} \\
			&\le \E[\circled{A}^2] \cdot 3^d\cdot \E\left[\left(x^{\top}\Pi_V(V^* - V)\nabla\right)^2\right] \\
			&\le 9^d \cdot O(r^3d^2)\cdot \norm{\vec{c} - \vec{c}^*}^2_2 \cdot \procr{V,V^*}^4,
		\end{align} where in the first step we used Cauchy-Schwarz, in the second we used Proposition~\ref{prop:helper_circledBprime}, and in the third we used \eqref{eq:circAprime_squared}.
	\end{proof}

	We can now complete the proof of Lemma~\ref{lem:circledAprimeconc}. 

	By Lemma~\ref{lem:polynomial_concentration}, Lemma~\ref{lem:Bprimeexp}, and Lemma~\ref{lem:Bprimevar}, \begin{equation}
		\left|\frac{1}{B}\sum^{B-1}_{i=0}\circled{B}^{x^i}\right| \le O(r^{3/2}d)\cdot\norm{\vec{c} - \vec{c}^*}_2\cdot\procr{V,V^*}^2\cdot\left(1 + \frac{1}{\sqrt{B}}\cdot O(\log(1/\delta))^{d}\cdot 3^d\right),
	\end{equation} with probability at least $1 - \delta$. The lemma follows by taking $B = \Omega(\log(1/\delta))^{2d}\cdot\Omega(9^d)$.
\end{proof}

\subsubsection{Proof of Lemma~\ref{lem:circledCprimeconc}}
\label{app:circledCprimeconc}

\begin{proof}
	Note that $\circled{C}^x$ is a polynomial of degree $2d$ in $x$, so by Lemma~\ref{lem:polynomial_concentration}, we just need to upper bound its variance.

	\begin{lemma}\label{lem:Cprimevar}
		For any $\Theta$, $\E_x[(\circled{C}^{\Theta,x})^2] \le \procr{V,V^*}^2\cdot\norm{\vec{c} - \vec{c}^*}^2_2\cdot \exp(O(d))$.
	\end{lemma}

	\begin{proof}
		This is shown in Lemma~\ref{lem:Cvar} below. The proof involves calculations which are more pertinent to the behavior of \textsc{SubspaceDescent}, so we defer the details to there.
	\end{proof}

	We can now complete the proof of Lemma~\ref{lem:circledCprimeconc}. By Lemma~\ref{lem:polynomial_concentration}, Observation~\ref{obs:Cprimevanish}, and Lemma~\ref{lem:Cprimevar}, \begin{equation}
		\left|\frac{1}{B}\sum^{B-1}_{i=1}\circled{C}^{x^i}\right| \le \frac{1}{\sqrt{B}}\cdot O(\log(1/\delta))^{d}\cdot \procr{V,V^*}\cdot \norm{\vec{c} - \vec{c}^*}_2 \cdot \exp(O(d))
	\end{equation} with probability at least $1 - \delta$. The lemma follows by taking $B = \Omega(\log(1/\delta))^{2d}\cdot \gamma^{-2}$.
\end{proof}

\subsection{Proof of Lemma~\ref{lem:Ecoef_conc}}
\label{app:Ecoef_conc}

\begin{proof}
	Note that $\Ecoef^x$ is a polynomial of degree $2d$ in $x$, so by Lemma~\ref{lem:polynomial_concentration}, we just need to upper bound its variance.

	To do so, we will need the following helper lemma, which like Lemma~\ref{lem:main_taylorfourth} is a straightforward consequence of Lemma~\ref{lem:main_taylorterms}.

	\begin{lemma}\label{lem:taylorfourth}
		$\E[(\residual^{\Theta,x})^4]^{1/2} \le O(dr^3)^{d+1}\cdot \norm{V - V^*}^2_F \cdot \left(\norm{\vec{c} - \vec{c}^*}_2 + \norm{V^* - V}_F\right)^2$
	\end{lemma}

	\begin{proof}
		We have that
		\begin{align}
			\E\left[(\residual^{\Theta,x})^4\right]^{1/2} &= \left(\sum_{\ell_1,...,\ell_4>1}\frac{1}{\prod^4_{\nu=1}\ell_{\nu}!}\E\left[\prod^4_{\nu=1}\left\langle \nabla^{[\ell_{\nu}]}F_x(\Theta),(\Theta^* - \Theta)^{\otimes \ell_{\nu}}\right\rangle\right]\right)^{1/2} \\
			&\le \left(\sum_{\ell_1,...,\ell_4>1}\frac{1}{\prod^4_{\nu=1}\ell_{\nu}!} 16\cdot(8dr^2)^{2(d+1)}\cdot\norm{V^* - V}^{\sum_{\nu}\ell_{\nu}}_F \cdot \left(1 + \frac{\norm{\vec{c} - \vec{c}^*}_2}{\norm{V^* - V}_F}\right)^4\right)^{1/2} \\
			&= 4(8dr^2)^{d+1}\left(\sum^{d+1}_{\ell=2}\frac{1}{\ell!}\norm{V^* - V}^{\ell}_F\right)^2 \cdot \left(1 + \frac{\norm{\vec{c} - \vec{c}^*}_2}{\norm{V^* - V}_F}\right)^2 \\
			&\le 4(8dr^2)^{d+1}\cdot\left(e^2\cdot (4r)^{d-1}\norm{V^* - V}^4_F\right) \cdot \left(1 + \frac{\norm{\vec{c} - \vec{c}^*}_2}{\norm{V^* - V}_F}\right)^2 \\
			&= 4e^2\cdot (32dr^3)^{d+1}\cdot \norm{V - V^*}^2_F \cdot \left(\norm{\vec{c} - \vec{c}^*}_2 + \norm{V^* - V}_F\right)^2,
		\end{align} where the second step follows by Lemma~\ref{lem:main_taylorterms}, and the fourth step follows by the fact that we always have $\norm{V - V^*}_F\le 2\sqrt{r}$, and $\sum^{d+1}_{\ell=2}\frac{1}{\ell!}x^{\ell} < e\cdot (4r)^{(d-1)/2}\cdot x^2$ for $x\in[0,2\sqrt{r}]$.
	\end{proof}

	We can now show the variance bound.

	% Variance bound
	\begin{lemma}\label{lem:Ecoef_var}
		$\E_x[(\Ecoef^x)^2] \le \etac^2 \cdot O(dr^3)^{d+1}\cdot \procr{V,V^*}^2\cdot\norm{\vec{c}-\vec{c}^*}^2_2\cdot(\procr{V,V^*} + \norm{\vec{c} - \vec{c}^*})^2$.
	\end{lemma}

	\begin{proof}
		By Cauchy-Schwarz, \begin{align}
			\frac{1}{4\etac^2}\E\left[\left(\Ecoef^x\right)^2\right] &\le \E[(\residual^{\Theta,x})^4]^{1/2} \cdot \E[\polydiff(g)^4]^{1/2} \\
			&\le 4e^2\cdot (32dr^3)^{d+1}\cdot \norm{V - V^*}^2_F \cdot \left(\norm{\vec{c} - \vec{c}^*}_2 + \norm{V^* - V}_F\right)^2\cdot 3^d\cdot\norm{\vec{c} - \vec{c}^*}^2_2 \\
			&= O(dr^3)^{d+1}\cdot \procr{V,V^*}^2\cdot\norm{\vec{c}-\vec{c}^*}^2_2\cdot(\procr{V,V^*} + \norm{\vec{c} - \vec{c}^*})^2
		\end{align} where the second step follows by Lemma~\ref{lem:taylorfourth} and the third step follows by the assumption that $\norm{V - V^*}_F = \procr{V,V^*}$.
	\end{proof}
	
	Finally, by Lemma~\ref{lem:polynomial_concentration}, Lemma~\ref{lem:Ecoef_term_curvature_expectation}, and Lemma~\ref{lem:Ecoef_var}, \begin{equation}
		\left|\frac{1}{B}\sum^{B-1}_{i=0}\Ecoef^{x^i}\right| \le O(dr^3)^{(d+1)/2} \, \procr{V,V^*} \, \norm{\vec{c}-\vec{c}^*}_2\, (\procr{V,V^*} + \norm{\vec{c} - \vec{c}^*}_2)\cdot \left(1 + \frac{1}{\sqrt{B}}\cdot O(\log(1/\delta))^{d}\right).
	\end{equation}

	The lemma follows by taking $B = O(\log(1/\delta))^{2d}$.
\end{proof}

\section{Deferred Proofs from Section~\ref{sec:subspacedescent}}
\label{app:subspace}

\subsection{Proof of Lemma~\ref{lem:main_taylorterms}}
\label{app:main_taylorterms}

\begin{proof}
	We begin by explicitly computing the higher-order terms in the Taylor-expansion of $F_x(\Theta) - F_x(\Theta^*)$. For any $\ell\in[d+1]$, recalling the notation of \eqref{eq:deriv_notation} and \eqref{eq:matrix_deriv_notation},
	\begin{align}
		\MoveEqLeft \left\langle \nabla^{[\ell]}F_x(\Theta),(\Theta^* - \Theta)^{\otimes \ell}\right\rangle \\
		&= \sum_{\vec{i}\in[n]^{\ell}, \vec{j}\in[r]^{\ell}}\prod^{\ell}_{a=1}(V^*_{i_a,j_a} - V_{i_a,j_a})\cdot \D{\vec{i},\vec{j}}{F_x(\Theta)} + \sum_{I, \vec{i}\in[n]^{\ell}, \vec{j}\in[r]^{\ell-1}}\prod^{\ell-1}_{a=1}(V^*_{i_a,j_a} - V_{i_a,j_a})\cdot (c^*_I - c_I)\cdot \D{\vec{i},\vec{j}}{F_x(\Theta)} \\
		&= \sum_{\vec{i}\in[n]^{\ell}, \vec{j}\in[r]^{\ell}}\prod^{\ell}_{a=1}(V^*_{i_a,j_a} - V_{i_a,j_a})\cdot x_{i_a}\cdot \D{\vec{j}}{p(V^{\top}x)} + \sum_{\vec{i}\in[n]^{\ell}, \vec{j}\in[r]^{\ell-1}}\prod^{\ell-1}_{a=1}(V^*_{i_a,j_a} - V_{i_a,j_a})\cdot x_{i_a}\cdot \D{\vec{j}}{\polydiff(V^{\top}x)} \\
		&= \sum_{\vec{j}\in[r]^{\ell}}\prod^{\ell}_{a=1}\langle (V^*-V)_{j_a}, x\rangle \cdot \D{\vec{j}}{p(V^{\top}x)} + \sum_{\vec{j}\in[r]^{\ell-1}}\prod^{\ell-1}_{a=1}\langle (V^*-V)_{j_a}, x\rangle\cdot \D{\vec{j}}{\polydiff(V^{\top}x)}\label{eq:grad_expanded}
	\end{align}

	From \eqref{eq:grad_expanded}, we can rewrite the quantity in the expectation as \begin{equation}
		\sum_{\substack{\vec{b}\in\{0,1\}^m \\ \{\vec{j}^{(\nu)}\}_{\nu\in[m]}}}\prod^m_{\nu=1}\left(\prod^{\ell_{\nu} - b_{\nu}}_{a=1}\left\langle (V^* - V)_{j^{(\nu)}_a},x\right\rangle\right) \left(\bone{b_{\nu} = 0}\cdot \D{\vec{j^{(\nu)}}}{p(V^{\top}x)} + \bone{b_{\nu} = 1}\cdot \D{\vec{j^{(\nu)}}}{\polydiff(V^{\top}x)}\right).\label{eq:main_ugly}
	\end{equation} We will bound the expected absolute values of each of these summands individually, so henceforth fix an arbitrary $\vec{b},\{\vec{j}^{(\nu)}\}$. For convenience, define $C_{\nu}\triangleq \left(\bone{b_{\nu} = 0}\cdot \D{\vec{j^{(\nu)}}}{p(V^{\top}x)} + \bone{b_{\nu} = 1}\cdot \D{\vec{j^{(\nu)}}}{\polydiff(V^{\top}x)}\right)$.

	% For any $\vec{b}\in\{0,1\}^m$, define $w_{\vec{b}}\triangleq \sum^m_{\nu=1}\ell_{\nu} - b_{\nu}$.
	By AM-GM, we have that \begin{align}
		\MoveEqLeft \E\left[\left(\prod^m_{\nu=1}|C_{\nu}|\right)\cdot \left(\prod^m_{\nu=1}\prod^{\ell_{\nu} - b_{\nu}}_{a=1}\left|\left\langle (V^* - V)_{j^{(\nu)}_a},x\right\rangle\right|\right)\right] \\
		&\le \E\left[\left(\prod^m_{\nu=1}C_{\nu}\right)\cdot \left(\prod^m_{\nu=1}\frac{1}{\ell_{\nu} - b_{\nu}}\sum^{\ell_{\nu} - b_{\nu}}_{a=1}\left|\left\langle (V^* - V)_{j^{(\nu)}_a},x\right\rangle\right|^{\ell_{\nu} - b_{\nu}}\right)\right] \\
		&\le \E\left[\prod^m_{\nu = 1}\frac{C^2_{\nu}}{(\ell_{\nu} - b_{\nu})^2}\right]^{1/2}\cdot \E\left[\left(\sum_{\vec{a}\in\prod_{\nu}[\ell_{\nu} - b_{\nu}]}\prod^m_{\nu = 1}\left|\left\langle (V^*-V)_{j^{(\nu)}_{a_{\nu}}},x\right\rangle\right|^{\ell_{\nu} - b_{\nu}}\right)^2\right]^{1/2}\label{eq:split_cs}
		% \frac{1}{w_{\vec{b}}}\sum^m_{\nu=1}\sum^{\ell_{\nu} - b_{\nu}}_{a=1}\left|\left\langle (V^* - V)_{j^{(\nu)}_a},x\right\rangle\right|^{w_{\vec{b}}} \cdot \left(\bone{b_{\nu} = 0}\cdot \left|\D{\vec{j^{(\nu)}}}{p(V^{\top}x)}\right|^{w_{\vec{b}}} + \bone{b_{\nu} = 1} \cdot \left|\D{\vec{j^{(\nu)}}}{\polydiff(V^{\top}x)}\right|^{w_{\vec{b}}}\right)
	\end{align} where the last inequality follows by Cauchy-Schwarz.

	Defining $w_{\vec{b}}  = \sum_{\nu} \ell_{\nu} - b_{\nu}$, we may write the second factor in \eqref{eq:split_cs} as \begin{multline}
		\E\left[\sum_{\vec{a}^1,\vec{a}^2}\prod^m_{\nu=1}\left|\left\langle (V^*-V)_{j^{(\nu)}_{a^1_{\nu}}},x\right\rangle\right|^{\ell_{\nu} - b_{\nu}}\cdot \prod^m_{\nu=1}\left|\left\langle (V^*-V)_{j^{(\nu)}_{a^2_{\nu}}},x\right\rangle\right|^{\ell_{\nu} - b_{\nu}}\right]^{1/2} \\
		\le (2w_{\vec{b}})^{w_{\vec{b}}/2}\norm{V^* - V}_F^{w_{\vec{b}}}\cdot \prod_{\nu}(\ell_{\nu} - b_{\nu}) \le (2m)^{m/2}\norm{V^* - V}_F^{w_{\vec{b}}}\cdot \prod_{\nu}(\ell_{\nu} - b_{\nu}),
	\end{multline} where we used the standard bound for moments of a univariate Gaussian, the fact that there are $\prod_{\nu}(\ell_{\nu} - b_{\nu})^2$ pairs of summands $\vec{a}^1,\vec{a}^2$, and the fact that any column of $V^* - V$ has $L_2$ norm at most $\norm{V^* - V}_F$.

	By Holder's, we may upper bound the first factor in \eqref{eq:split_cs} by $\prod^m_{\nu = 1}\frac{1}{\ell_{\nu} - b_{\nu}}\E\left[C^{2m}_{\nu}\right]^{1/2m}$.

	By Corollary~\ref{cor:derivpower_bound}, 
	\begin{equation}
		\E\left[\left(\D{\vec{j}^{(\nu)}}{\polydiff(V^{\top}x)}\right)^{2m}\right]^{1/2m} \le (2m)^{d/2}d^{(\ell_{\nu}-1)/2}\cdot\Var[\polydiff]^{1/2} \le (2m)^{d/2}d^{\ell_{\nu}/2}\cdot\norm{\vec{c} - \vec{c}_*}_2.
	\end{equation}
	\begin{equation}
		\E\left[\left(\D{\vec{j}^{(\nu)}}{p(V^{\top}x)}\right)^{2m}\right]^{1/2m} \le (2m)^{d/2}d^{\ell_{\nu}/2}\cdot\Var[p]^{1/2}\le 2\cdot (2m)^{d/2}d^{\ell_{\nu}/2},
	\end{equation} where in the last step we used that $\Var[p]^{1/2} \le \Var[p*]^{1/2} + \Var[\polydiff]^{1/2} \le 2$. So the first factor in \eqref{eq:split_cs} is at most \begin{equation}
		\left(\prod^m_{\nu = 1}\frac{1}{\ell_{\nu} - b_{\nu}}\right)\cdot 2^m\cdot (2m)^{md/2}d^{\sum_{\nu}\ell_{\nu}/2}\norm{\vec{c} - \vec{c}_*}^{\sum_{\nu}b_{\nu}}_2,
	\end{equation} so \eqref{eq:split_cs} is at most $2^{m} \cdot (2m)^{m(d+1)/2}d^{m(d+1)/2}\cdot \norm{V^* - V}^{w_{\vec{b}}}_F\cdot \norm{\vec{c} - \vec{c}_*}^{\sum_{\nu}b_{\nu}}_2$. The proof follows by noting that \begin{equation}
		\sum_{\vec{b}}\norm{V^* - V}^{w_{\vec{b}}}_F\cdot \norm{\vec{c} - \vec{c}_*}^{\sum_{\nu}b_{\nu}}_2 = \norm{V^* - V}^{\sum_{\nu}\ell_{\nu}}_F \cdot \sum_{\vec{b}}\left(\frac{\norm{\vec{c}-\vec{c}_*}_2}{\norm{V^* - V}_F}\right)^{\sum_{\nu}b_{\nu}}
	\end{equation} and summing \eqref{eq:split_cs} over all choices of $\vec{b}$ and all $\prod_{\nu} r^{\ell_{\nu}} \le r^{m(d+1)}$ choices of $\{\vec{j}^{(\nu)}\}$.
\end{proof}

\subsection{Proof of Lemma~\ref{lem:safe_exp}}
\label{app:safe_exp}

\begin{proof}
	Let \begin{equation}\label{eq:alphaprime}
		\alpha_q \triangleq 3\sqrt{r}\cdot \etav\cdot O(\sqrt{n})\cdot O(dr^3)^{(d+2)/2}.
	\end{equation}
	($q = 1$). Analogous to the derivation of \eqref{eq:unwrap_recursion}, we have that \begin{align}
		\E\left[\norm{V^{(t)} - V^*}_F\right] &\le \E\left[\norm{V^{(t-1)} - V^*}_F\right] + \E\left[\norm{\DeltaV^{\Theta^{(t-1)},x^{t-1}}}_F\right] \\
		&\le \E\left[\norm{V^{(t-1)} - V^*}_F\right] + 3\sqrt{r}\cdot \etav\E\left[(\sigma^{\Theta^{(t-1)},x^{t-1}})^2\right]^{1/2} \\
		&\le (1 + \alpha_1)\E\left[\norm{V^{(t-1)} - V^*}_F\right] + \alpha_1\cdot\norm{\vec{c} - \vec{c}^*}_2 \\
		&\le (1 + \alpha_1)^t\cdot \norm{V^{(0)} - V^*}_F + \left((1 + \alpha_1)^t - 1\right)\cdot \norm{\vec{c} - \vec{c}^*}_2 \\
		&= (1 + \alpha_1)^t\cdot \procr{V^{(0)},V^*} + \left((1 + \alpha_1)^t - 1\right)\cdot \norm{\vec{c} - \vec{c}^*}_2
	\end{align} where in the second step we used Cauchy-Schwarz and \eqref{eq:DeltaVnorm}, in the third step we used Lemma~\ref{lem:sigsquared}, in the fourth step we unrolled the recurrence, and in the last step we used the assumption that $\norm{V^{(0)} - V^*}_F = \procr{V^{(0)},V^*}$. The proof follows by taking $\etav$ small enough that \begin{equation}
		(1 + \alpha_1)^t + \left((1 + \alpha_1)^t - 1\right)\cdot\frac{\norm{\vec{c} - \vec{c}^*}_2}{\procr{V^{(0)},V^*}} \le 1.1.
	\end{equation} $\etav$ given by \eqref{eq:main_etav_assumption} will easily satisfy this.

	(Larger $q$) We have that \begin{align}
		\E\left[\norm{V^{(t)} - V^*}^q_F\right]^{1/q} &\le \E\left[\norm{V^{(t-1)} - V^*}^q_F\right]^{1/q} + \E\left[\norm{\DeltaV^{\Theta^{(t-1)},x^{t-1}}}^q_F\right]^{1/q} \\
		&\le \E\left[\norm{V^{(t-1)} - V^*}^q_F\right]^{1/q} + \E\left[\norm{\DeltaV^{\Theta^{(t-1)},x^{t-1}}}^{2q}_F\right]^{1/2q} \\
		&\le \E\left[\norm{V^{(t-1)} - V^*}^q_F\right]^{1/q} + \alpha_q\cdot \left(\E\left[\norm{V^{(t-1)} - V^*}_F\right] + \norm{\vec{c} - \vec{c}^*}_2\right) \\
		&\le \E\left[\norm{V^{(t-1)} - V^*}^2_F\right]^{1/q} + 1.1\alpha_q\cdot\left(\procr{V^{(0)},V^*} + \cdot\norm{\vec{c} - \vec{c}^*}_2\right) \\
		&\le \procr{V^{(0)},V^*} + 1.1t\cdot \alpha_q\cdot\left(\procr{V^{(0)},V^*} + \cdot\norm{\vec{c} - \vec{c}^*}_2\right)
	\end{align} where the first step follows by triangle inequality, the second by monotonicity of $L_p$ norms, the third by Lemma~\ref{lem:sigsquared}, the fourth by Lemma~\ref{lem:safe_exp}, and the fifth by unrolling the recurrence and using the assumption that assumption that $\norm{V^{(0)} - V^*}_F = \procr{V^{(0)},V^*}$.

	The proof follows by taking $\etav$ small enough that $1.1T\cdot \alpha_q\cdot\norm{\vec{c} - \vec{c}^*}_2 \le O(\alpha_q\cdot T)$ is a negligible constant, which is certainly the case if $\etav$ satisfies \eqref{eq:main_etav_assumption} (with hidden constant factors there depending on $q$).
\end{proof}

\subsection{Proof of Lemma~\ref{lem:taylor_term_curvature_expectation}}

We first prove the following basic consequence of Lemma~\ref{lem:main_taylorterms}:

\begin{lemma}\label{lem:taylorsquared}
	\begin{multline}
		\E\left[\left(\sum^{d+1}_{\ell=2}\frac{1}{\ell!}\Iprod{\nabla^{[\ell]}F_x(\Theta)}{(\Theta^*-\Theta)^{\otimes\ell}}\right)^2\right]^{1/2} \\ \le O(dr^3)^{(d+1)/2}\cdot\norm{V-V^*}_F \cdot \left(\norm{V-V^*}_F + \norm{\vec{c} - \vec{c^*}}_2\right)
		\label{eq:taylorsquared}
	\end{multline}
\end{lemma}

\begin{proof}
	The left-hand side of \eqref{eq:taylorsquared} can be rewritten as \begin{align}
		\MoveEqLeft \left(\sum_{\ell_1,\ell_2 > 1}\frac{1}{\ell_1!\ell_2!}\E\left[\prod^2_{\nu = 1}\left\langle \nabla^{[\ell_{\nu}]}F_x(\Theta),(\Theta^* - \Theta)^{\otimes \ell_{\nu}}\right\rangle\right]\right)^{1/2} \cdot \E\left[(x^{\top}\cdot \Pi^{\perp}_V V^*\cdot \Delta)^2\right]^{1/2} \\
		&\le \left(\sum_{\ell_1,\ell_2>1}\frac{1}{\ell_1!\ell_2!}4\cdot (4dr^2)^{d+1}\cdot \norm{V-V^*}^{\ell_1+\ell_2}_F \cdot \left(1 + \frac{\norm{\vec{c} - \vec{c^*}}_2}{\norm{V^* - V}_F}\right)^2\right)^{1/2} \\
		&= 2(4dr^2)^{(d+1)/2}\sum_{\ell > 1}\frac{1}{\ell!}\norm{V - V^*}^{\ell}_F\cdot \left(1 + \frac{\norm{\vec{c} - \vec{c^*}}_2}{\norm{V^* - V}_F}\right) \\
		&\le 2e(16dr^3)^{(d+1)/2}\cdot\norm{V-V^*}_F \cdot \left(\norm{V-V^*}_F + \norm{\vec{c} - \vec{c^*}}_2\right),
	\end{align}
	where the first step follows by Lemma~\ref{lem:main_taylorterms}, and the last step follows by the fact that $\norm{V - V^*}_F\le 2\sqrt{r}$ and the fact that $\sum^{d+1}_{\ell=2}\frac{1}{\ell!}x^{\ell} < e\cdot (4r)^{(d-1)/2}\cdot x^2$ for $x\in[0,2\sqrt{r}]$.
\end{proof}

\begin{proof}[Proof of Lemma~\ref{lem:taylor_term_curvature_expectation}]
 	We have that \begin{align}
		\MoveEqLeft \frac{1}{2\etav}\left|\langle \tilde{\Delta}''_V, V - V^*\rangle\right| \\
		&= \left|\E\left[\sum^{d+1}_{\ell=2}\frac{1}{\ell!}\left\langle \nabla^{[\ell]}F_x(\Theta),(\Theta^* - \Theta)^{\otimes \ell}\right\rangle \cdot x^{\top}\cdot \Pi^{\perp}_V V^*\cdot \Delta\right]\right| \\
		&\le \E\left[\left(\sum^{d+1}_{\ell=2}\frac{1}{\ell!}\Iprod{\nabla^{[\ell]}F_x(\Theta)}{(\Theta^*-\Theta)^{\otimes\ell}}\right)^2\right]^{1/2} \cdot \E\left[(x^{\top}\cdot \Pi^{\perp}_V V^*\cdot \Delta)^2\right]^{1/2} \\
		&\le O(dr^3)^{(d+1)/2}\cdot\norm{V-V^*}_F \cdot \left(\norm{V-V^*}_F + \norm{\vec{c} - \vec{c^*}}_2\right)\cdot \E[\circled{A'}]^{1/2} \\
		&\le O(dr^3)^{(d+1)/2}\cdot\norm{V-V^*}_F \cdot \left(\norm{V-V^*}_F + \norm{\vec{c} - \vec{c^*}}_2\right)\cdot (2\procr{V,V^*}),\label{eq:higher_order_bound}
	\end{align} where the second step follows by Cauchy-Schwarz, the third by Lemma~\ref{eq:taylorsquared} and the definition of $\circled{A'}$, the fourth by the upper bound in \eqref{eq:circledAbound}.
\end{proof}

\subsection{Proof of Lemma~\ref{lem:trig_term_curvature_expectation}}

By Holder's, \begin{align}
	\MoveEqLeft \left|\E\left[\langle \trigE, V - V^*\rangle\right]\right| \\
	&\le \E\left[|\cos(\sigma\etav)-1|\right]\cdot \sup_{\hatnab}\left|\langle V\cdot\hatnab\hatnab^{\top},V - V^*\rangle\right| + \E\left[|\sin(\sigma\etav) - \sigma\etav|\right]\cdot \sup_{\hath,\hatnab}\left|\langle \hath\hatnab^{\top},V - V^*\rangle\right| \\
	&\le O(\etav^2)\cdot \E[\sigma^2]\cdot\left(\sup_{\hatnab}\left|\langle V\cdot\hatnab\hatnab^{\top},V - V^*\rangle\right| + \sup_{\hath,\hatnab}\left|\langle \hath\hatnab^{\top},V - V^*\rangle\right|\right),\label{eq:trig_error_main}
\end{align} where in the second step we used that $|\cos(x) - 1|\le x^2/2$ and $|\sin(x) - x|\le x^2/\pi$ for all $x\ge 0$, and in the third step we invoked Lemmas~\ref{lem:maxterm1} and \ref{lem:maxterm2} below.

% In the calculations that follow, we emphasize that to get linear convergence, it is crucial that our eventual bound scales with the \emph{square} of $\procr{V,V^*}$.

\begin{lemma}\label{lem:maxterm1}
	For any $\hatnab\in\S^{r-1}$, $\left|\langle V\cdot\hatnab\hatnab^{\top},V - V^*\rangle\right|\le \norm{V-V^*}_F$.
\end{lemma}

\begin{proof}
	We may write the quantity on the left-hand side as \begin{equation}
		\hatnab^{\top}\cdot\left((V - V^*)^{\top}V\right)\cdot \hatnab = \hatnab^{\top}\left(\Id - {V^*}^{\top}V\right)\hatnab \le \norm{\Id - {V^*}^{\top}V}_2 \le \norm{V - V^*}_F,
	\end{equation} where the last step follows by the first part of Lemma~\ref{lem:sigmamax_bound}.
\end{proof}

\begin{lemma}\label{lem:maxterm2}
	For any $\hatnab\in\S^{r-1}$ and $\hath\in\S^{n-1}$ for which $\hath$ lies in the orthogonal complement of the column span of $V$, $\left|\left\langle \hath\hatnab^{\top},V - V^*\right\rangle\right| \le \procr{V,V^*}$.
\end{lemma}

\begin{proof}
	Because $\Pi^{\perp}_V\hath = \hath$, The left-hand side can be rewritten as \begin{equation}
		\hath^{\top}(V - V^*)\hatnab = \hath^{\top}\Pi^{\perp}_V(V - V^*)\hatnab,
	\end{equation} it is upper-bounded by \begin{align}
		\sigma_{\max}(\Pi^{\perp}(V - V^*)) &\le \Tr((V - V^*)^{\top}(\Id - VV^{\top})(V - V^*)^{1/2} \\
		&= \Tr(\Id - V^*V^{\top}V{V^*}^{\top})^{1/2} \\
		&= \chord{V,V^*} \le \procr{V,V^*},
	\end{align} where the last step follows by Lemma~\ref{lem:procrustes_chordal}.
\end{proof}

\begin{proof}[Proof of Lemma~\ref{lem:trig_term_curvature_expectation}]
	We have \begin{equation}
		\left|\E\left[\langle \trigE, V - V^*\rangle\right]\right| \le O(\etav^2)\cdot O(n)\cdot O(dr^3)^{d+2}\cdot \norm{V-V^*}_F\cdot \left(\norm{V-V^*}_F + \norm{\vec{c} - \vec{c}_*}_2\right)^2,
	\end{equation} by \eqref{eq:trig_error_main}, Lemmas~\ref{lem:maxterm1}, \ref{lem:maxterm2}, and \ref{lem:sigsquared}. The lemma follows by taking $\etav\le O(1/n)$.
\end{proof}
	% so \assumption{by taking $\eta \le \sqrt{2/n}$ and $\procr{V,V^*} \le (\condnumber/16)\cdot (rd)^{-\Omega(d)}$}, we get that this quantity is upper bounded by $2\eta\cdot (\condnumber/16)\cdot d(\Theta,\Theta^*)^2$.

\subsection{Proof of Lemma~\ref{lem:sumdomvecs}}
\label{app:sumdomvecs}

\begin{proof}
	We will bound each $\E_{x^0,...,x^{t-1}}[\mu_{\domvec}(\Theta^{(t)}]$ individually. By Lemma~\ref{lem:domvec_term_curvature_expectation}, for any realization of $x^0,...,x^{t-1}$ giving rise to iterate $\Theta^{(t)} = (\vec{c},V^{(t)})$, $\mu_{\domvec}(\Theta^{(t)}) \ge (\condnumber/4)\cdot\procr{V^{(t)},V^*}^2$. We have that \begin{align}
		\MoveEqLeft \E\left[\procr{V^{(t)},V^*}^2\right] \\
		&\ge \E\left[\left(\procr{V^{(t-1)},V^*} - \procr{V^{(t)}, V^{(t-1)}}\right)^2\right] \\ 
		&\ge \E\left[\procr{V^{(t-1},V^*}^2\right] -2 \E\left[\procr{V^{(t-1)},V^*}^2\right]^{1/2}\cdot\E\left[\procr{V^{(t)},V^{(t-1)}}^2\right]^{1/2} \\
		&\ge \E\left[\procr{V^{(t-1},V^*}^2\right] -2 \E\left[\procr{V^{(t-1)},V^*}^2\right]^{1/2}\cdot\E\left[\norm{\DeltaV^{\Theta^{(t-1)},x^{t-1}}}^2_F\right]^{1/2} \\
		&\ge \E\left[\procr{V^{(t-1},V^*}^2\right] - 6\sqrt{r}\cdot \etav\E\left[\procr{V^{(t-1)},V^*}^2\right]^{1/2}\cdot\E\left[(\sigma^{\Theta^{(t-1)},x^{t-1}})^2\right]^{1/2} \label{eq:dpsquared}
	\end{align} where the first step follows by triangle inequality (Fact~\ref{fact:triangle}), the second by Cauchy-Schwarz, the third by the definition of Procrustes distance, and the fourth by \eqref{eq:DeltaVnorm}. By Lemma~\ref{lem:sigsquared} and Lemma~\ref{lem:safe_exp}, \begin{align}
		\MoveEqLeft 6\sqrt{r}\cdot \etav \E\left[(\sigma^{\Theta^{(t-1)},x^{t-1}})^2\right]^{1/2} \\
		&\le 6\sqrt{r}\cdot \etav\cdot O(\sqrt{n})\cdot (dr^3)^{(d+2)/2}\cdot\left(\E\left[\norm{V^{(t-1)} - V^*}_F\right] + \norm{\vec{c} - \vec{c}^*}_2\right) \\
		&\le 6\sqrt{r}\cdot O(\sqrt{n})\cdot (dr^3)^{(d+2)/2}\cdot\left(1.1\procr{V^{(0)},V^*} + \norm{\vec{c} - \vec{c}^*}_2\right) \\
		&\le \frac{1}{100T}\procr{V^{(0)},V^*},
	\end{align} where the last step follows by our choice of $\etav$ in \eqref{eq:main_etav_assumption}.
	So by \eqref{eq:dpsquared} we conclude that as long as $\E[\procr{V^{(s)},V^*}^2] > \procr{V^{(0)},V^*}^2/1.1$ for all $s < t$, \begin{align}
		\E\left[\procr{V^{(t)},V^*}^2\right] &\ge \left(1 - \frac{\sqrt{1.1}}{100T}\right)\E\left[\procr{V^{(t-1)},V^*}^2\right] \\
		&\ge \left(1 - \frac{\sqrt{1.1}}{100T}\right)^t \procr{V^{(0)},V^*}^2 \\
		&\ge \procr{V^{(0)},V^*}^2/1.1.
	\end{align} 
	By induction, $\procr{V^{(t)},V^*}^2\ge \procr{V^{(0)},V^*}^2/1.1$ for all $0\le t<T$. Recalling that $\mu_{\domvec}(\Theta^{(t)})\ge(\condnumber/4)\cdot \procr{V^{(t)},V^*}^2$, we conclude that
	 \begin{equation}
		\E\left[\sum^{T-1}_{t=0}\mu_{\domvec}(\Theta^{(t)})\right] \ge T\cdot(\condnumber/4)\cdot\left(\procr{V^{(0)},V^*}^2/1.1\right)
	\end{equation} as desired.
\end{proof}

\subsection{Proof of Lemma~\ref{lem:sumEvecones}}
\label{app:sumEvecones}

\begin{proof}
	We will bound each $\E_{x^0,...,x^{t-1}}[\left|\mu_{\Evecone}(\Theta^{(t)}\right|]$ individually and apply triangle inequality.

	By Lemma~\ref{lem:taylor_term_curvature_expectation}, for any realization of $x^0,...,x^{t-1}$ giving rise to iterate $\Theta^{(t)} = (\vec{c},V^{(t)})$, \begin{align}
		\left|\mu_{\Evecone}(\Theta^{(t)})\right| &\le O(\etav)\cdot O(dr^3)^{(d+1)/2}\cdot \norm{V^{(t)} - V^*}_F\cdot \procr{V^{(t)},V^*}\cdot \left(\norm{V^{(t)} - V^*}_F + \norm{\vec{c} - \vec{c}^*}_2\right). \\
		&\le O(\etav)\cdot O(dr^3)^{(d+1)/2}\cdot \left(\norm{V^{(t)} - V^*}^3_F + \norm{V^{(t)} - V^*}^2_F \cdot \norm{\vec{c} - \vec{c}^*}_2\right).
	\end{align} By Lemma~\ref{lem:safe_exp} and \eqref{eq:ratio_bound}, we conclude that \begin{align}
		\E\left[\left|\mu_{\Evecone}(\Theta^{(t)})\right|\right] &\le O(\etav)\cdot O(dr^3)^{(d+1)/2}\cdot\left(1.1\procr{V^{(0)},V^*}^3 + 1.1\procr{V^{(0)},V^*}^2\cdot\norm{\vec{c} - \vec{c}^*}_2\right) \\
		&\le O(\etav)\cdot O(dr^3)^{(d+1)/2}\cdot \procr{V^{(0)},V^*}^3.
		% &\le O\left(\frac{1}{T\sqrt{n}}\right)\cdot \procr{V^{(0)},V^*}^2,
	\end{align} The claim follows by summing over $t$.
\end{proof}

\subsection{Proof of Lemma~\ref{lem:sumEvectwos}}
\label{app:sumEvectwos}

\begin{proof}
	We will bound each $\E_{x^0,...,x^{t-1}}[\left|\mu_{\Evectwo}(\Theta^{(t)}\right|]$ individually and apply triangle inequality.

	By Lemma~\ref{lem:trig_term_curvature_expectation}, for any realization of $x^0,...,x^{t-1}$ giving rise to iterate $\Theta^{(t)} = (\vec{c},V^{(t)})$, \begin{align}
		\MoveEqLeft\left|\mu_{\Evectwo}(\Theta^{(t)})\right| \\
		&\le \etav\cdot O(dr^3)^{d+2}\cdot \norm{V-V^*}_F\cdot \left(\norm{V-V^*}_F + \norm{\vec{c} - \vec{c}^*}_2\right)^2. \\
		&\le \etav\cdot O(dr^3)^{d+2}\cdot \left(\norm{V^{(t)} - V^*}^3_F + 2\norm{V^{(t)} - V^*}^2_F \cdot \norm{\vec{c} - \vec{c}^*}_2 + \norm{V^{(t)} - V^*}_F \cdot\norm{\vec{c} - \vec{c}^*}_2\right).
	\end{align} By Lemma~\ref{lem:safe_exp} and \eqref{eq:ratio_bound}, we conclude that \begin{equation}
		\E\left[\left|\mu_{\Evectwo}(\Theta^{(t)})\right|\right] \le O(\etav)\cdot O(dr^3)^{d+2}\cdot \procr{V^{(0)},V^*}^3
	\end{equation}
	The claim follows by summing over $t$.
\end{proof}

\subsection{Proof of Lemma~\ref{lem:domvec_conc}}
\label{app:domvec_conc}

Analogous to the proof of Lemma~\ref{lem:domcoef_conc} in Appendix~\ref{app:domcoef_conc}, we will prove concentration by decomposing the MDS $\{\mu_{\domvec}(\Theta^{(t)}) - \domvec^{\Theta^{(t)},x^t}\}_{0\le t<T}$ into components corresponding to the decomposition \eqref{eq:Vprimecor_lowest}. That is, define $\circled{A'}^{\Theta,x}$, $\circled{B'}^{\Theta,x}$, $\circled{C'}^{\Theta,x}$ to be the quantities in \eqref{eq:Vprimecor_lowest} for an iterate $\Theta$ and sample $x$. So by Observation~\ref{obs:BCvanish}, $\left\{\frac{1}{2\etav}\mu_{\domvec}(\Theta^{(t)}) - \circled{A'}^{\Theta^{(t)},x^t}\right\}$, $\{\circled{B'}^{\Theta^{(t)},x^t}\}$, and $\{\circled{C'}^{\Theta^{(t)},x^t}\}$ are MDS's, and for any $\Theta,x$, \begin{equation}
	\frac{1}{2\etav} \domvec^{\Theta,x} = \circled{A'}^{\Theta,x} + \circled{B'}^{\Theta,x} + \circled{C'}^{\Theta,x}
\end{equation} by \eqref{eq:Vprimecor_lowest}. We will show concentration for these MDS's separately.

\begin{lemma}\label{lem:circledA_conc}
	\begin{equation}
		\sum^{T-1}_{t=0}\circled{A'}^{\Theta^{(t)},x^t} \ge T\cdot (\condnumber/5)\cdot \procr{V^{(0)},V^*}^2
	\end{equation} with probability at least $1 - \delta$.
\end{lemma}

\begin{lemma}\label{lem:circledB_conc}
	\begin{equation}
		\left|\sum^{T-1}_{t=0}\circled{B'}^{\Theta^{(t)},x^t}\right| \le T\cdot (\condnumber/60)\cdot\procr{V^{(0)},V^*}^2
	\end{equation} with probability at least $1 - \delta$.
\end{lemma}

\begin{lemma}\label{lem:circledC_conc}
	\begin{equation}
		\left|\sum^{T-1}_{t=0}\circled{C'}^{\Theta^{(t)},x^t}\right| \le T\cdot (\condnumber/60)\cdot\procr{V^{(0)},V^*}^2
	\end{equation} with probability at least $1 - \delta$.
\end{lemma}

We prove these in the subsequent Appendices~\ref{app:circledA_conc}, \ref{app:circledB_conc}, and \ref{app:circledC_conc}. Note that Lemma~\ref{lem:domvec_conc} follows easily from these three lemmas:

\begin{proof}[Proof of Lemma~\ref{lem:domvec_conc}]
	The claim follows immediately from Lemmas~\ref{lem:circledA_conc}, \ref{lem:circledB_conc}, and \ref{lem:circledC_conc}; triangle inequality; replacing $3\delta$ in the resulting union bound with $\delta$; and absorbing constant factors.
\end{proof}

\subsubsection{Proof of Lemma~\ref{lem:circledA_conc}}
\label{app:circledA_conc}

\begin{proof}
	Observe that $\left\{\frac{1}{2\etav}\mu_{\domvec}(\Theta^{(t)}) - \circled{A'}^{\Theta^{(t)},x^t}\right\}$ is an MDS which satisfies one-sided bounds, as $\circled{A'}^{\Theta,x}\ge 0$ with probability one for any $\Theta,x$, so we wish to apply Lemma~\ref{lem:martingale2}. To do so, we just need to bound the variances of the differences.

	\begin{lemma}\label{lem:varbound_A}
		For any $\Theta$, $\Var_x[\circled{A'}^{\Theta,x}] \le 2^{4d+4}\cdot \procr{V,V^*}^4$.
	\end{lemma}

	\begin{proof}
		We will suppress superscripts $\Theta,x$ in this proof. $\Var[\circled{A'}] \le \E[\circled{A'}^2]$, so it suffices to bound the latter. But note that $x^{\top}\Pi^{\perp}_VV^*\grad{p}{V^{\top} x}$ is a polynomial, call it $f(x)$, of degree $d$ in the Gaussians $x_1,...,x_n$. By Fact~\ref{fact:hypercontractivity}, \begin{equation}\label{eq:Asquaredbound}
			\E[\circled{A'}^2] = \E[f(x)^4] \le \left(4^{d/2}\cdot \E[f(x)^2]^{1/2}\right)^4 \le 2^{4d}\cdot \E[f(x)^2]^2 = 2^{4d}\cdot \E[\circled{A'}]^2\le 2^{4d+4}\cdot\procr{V,V^*}^4,
		\end{equation} where the last step is by Lemma~\ref{lem:Aexp}.
		% Let $V'\in\R^{n\times(n-r)}$ be any matrix of orthonormal columns such that $(V|V')\in O(n)$. Note that for $x\sim\N(0,\Id_n)$, $\Pi^{\perp}_V x$ and $V^{\top} x$ are independent random vectors which are identical in distribution to $V'g_1$ and $g_2$ for $g_1\sim\N(0,\Id_{n-r})$ and $g_2\sim\N(0,\Id_r)$. So we have \begin{align}
		% 	\E[\circled{A'}^2] &= \E\left[\left(x^{\top}\Pi^{\perp}_VV^*\nabla\right)^4\right] \\
		% 	&= \E\left[\left(g_1^{\top} V'^{\top} V^*\grad{p}{g_2}\right)^4\right] \\
		% 	&\le \left(4^{d/2}\cdot \E\left[\left(g_1^{\top} V'^{\top} V^*\grad{p}{g_2}\right)^2\right]^{1/2}\right)^4 \\
		% 	&= 2^{4d}\cdot \E[\circled{A'}]^2
		% \end{align}
	\end{proof}

	We can now complete the proof of Lemma~\ref{lem:circledA_conc}. By Lemma~\ref{lem:safe} and Lemma~\ref{lem:varbound_A}, if $\etav$ satisfies \eqref{eq:main_etav_assumption}, then with probability $1 - \delta$ we have that for all $0\le t < T$, 
	\begin{equation}\frac{1}{2\etav}\mu_{\domvec}(\Theta^{(t)}) - \circled{A'}^{\Theta^{(t)},x^t} \le \frac{1}{2\etav}\mu_{\domvec}(\Theta^{(t)}) \le 4\procr{V^{(t)},V^*}^2 \le 4.84\procr{V^{(0)},V^*}^2.\end{equation} 
	\begin{equation}\Var_{x^t}[\circled{A'}^{\Theta^{(t)},x^t}] \le 1.1^4\cdot 2^{4d+4}\cdot \procr{V^{(0)},V^*}^4\end{equation} 
	Applying Lemma~\ref{lem:martingale2} with the parameter $\sigma^2$ taken to be $T\cdot 1.1^4\cdot 2^{4d+4}\cdot \procr{V^{(0)},V^*}^4$, we get \begin{equation}
		\Pr\left[\sum^{T-1}_{t=0}\circled{A'}^{\Theta^{(t)},x^t}\ge \frac{1}{2\etav}\sum^{T-1}_{t=0}\E\left[\mu_{\domvec}\left(\Theta^{(t)}\right)\right] - O\left(4^d\log(1/\delta)\sqrt{T}\cdot \procr{V^{(0)},V^*}^2\right)\right] \ge 1 - 2\delta,
	\end{equation} where the expectation in $\E\left[\domvec\left(\Theta^{(t)}\right)\right]$ is over the randomness of the samples $x^0,...,x^{t-1}$.

	By Lemma~\ref{lem:sumdomvecs}, we conclude that \begin{equation}\label{eq:final_prob_bound}
		\sum^{T-1}_{t=0}\circled{A'}^{\Theta^{(t)},x^t} \ge T\cdot (\condnumber/4.4)\cdot \procr{V^{(0)}, V^*}^2 - O\left(4^d\log(1/\delta)\sqrt{T}\cdot \procr{V^{(0)},V^*}^2\right)
	\end{equation} with probability at least $1 - 2\delta$. Taking $T$ according to \eqref{eq:main_T_bound} will certainly ensure the right-hand side of \eqref{eq:final_prob_bound} is at least $T\cdot (\condnumber/5)\cdot \procr{V^{(0)}, V^*}^2$. The proof is completed by replacing $2\delta$ in the above with $\delta$ and absorbing the resulting constant factors.
\end{proof}

\subsubsection{Proof of Lemma~\ref{lem:circledB_conc}}
\label{app:circledB_conc}

\begin{proof}
	For fixed $x^1,...,x^{t-1}$, the martingale difference $\circled{B'}^{\Theta^{(t)},x^t}$ is a polynomial of degree $2d$ in $x^t$, so by Lemma~\ref{lem:martingale1_polynomial} we just need to upper bound the second moments of the differences, which we do in the following lemma.

	\begin{lemma}\label{lem:Bvar}
		For any $\Theta$, $\E_x[(\circled{B'}^{\Theta,x})^2] \le \procr{V,V^*}^2\cdot \norm{V - V^*}_F^2 \cdot O(r^2)\cdot \exp(O(d))$.
	\end{lemma}

	\begin{proof}
		By Cauchy-Schwarz, \begin{align}
			\E\left[\circled{B'}^2\right] &\le \E\left[\left(x^{\top}\Pi_V(V^* - V)\nabla\right)^4\right]^{1/2}\cdot \E\left[\left(x^{\top}\Pi^{\perp}_VV^*\nabla\right)^4\right]^{1/2} \\
			&= \E_{g\sim\N(0,\Id_r)}\left[\left(g^{\top}V^{\top}(V^* - V)\grad{p}{g}\right)^4\right]^{1/2}\cdot \E\left[\circled{A'}^2\right]^{1/2} \\
			&\le \E_{g\sim\N(0,\Id_r)}\left[\left(g^{\top}(\Id - V^{\top}V^*)\grad{p}{g}\right)^4\right]^{1/2}\cdot 2^{2d+2}\cdot \procr{V,V^*}^2 \label{eq:leftover_poly},
		\end{align} where the third step follows by \eqref{eq:Asquaredbound}. It remains to bound the first factor in \eqref{eq:leftover_poly}. As this factor is independent of $n$, we do not need a particularly sharp bound. We have \begin{align}
			\E_g\left[\left(g^{\top}(\Id - V^{\top}V^*)\grad{p}{g}\right)^4\right]^{1/2} &\le \norm{\Id - V^{\top}V^*}_2^2\cdot \E_g[\norm{g}_2^4 \cdot \norm{\grad{p}{g}}_2^4]^{1/2} \\
			&\le \norm{V - V^*}^2_F \E_g[\norm{g}^8_2]^{1/4}\cdot \E_g[\norm{\grad{p}{g}}^8_2]^{1/4} \\
			&\le \norm{V - V^*}^2_F \cdot 3(r+1) \cdot (rd\cdot 7^d\cdot \Var[p]) \\
			&\le \norm{V - V^*}^2_F \cdot O(r^2 d\cdot 7^d),
		\end{align} where the second step follows by Lemma~\ref{lem:sigmamax_bound}, the third step follows by Corollary~\ref{cor:normgaussian} and Lemma~\ref{lem:normgrad} applied to $q = 4$, and the last step follows by noting that $\Var[p] = O(1)$ by triangle inequality and absorbing constant factors. The claimed bound follows.
	\end{proof}

	We now complete the proof of Lemma~\ref{lem:circledB_conc}. By Lemma~\ref{lem:safe}, $\procr{V^{(t)},V^*} \le \norm{V^{(t)} - V^*}_F \le 1.1\cdot\procr{V^{(0)},V^*}$ for every $0\le t\le T$ with probability at least $1 - \delta$, in which case Lemma~\ref{lem:Bvar} implies that for every $0\le t<T$, \begin{equation}\E\left[\left(\circled{B'}^{\Theta^{(t)},x^t}\right)^2 \;\middle|\; x^1,...,x^{t-1}\right] \le \procr{V^{(0)},V^*}^4 \cdot O(r^2)\cdot \exp(O(d))\end{equation} with probability at least $1 - \delta$. So by Lemma~\ref{lem:martingale1_polynomial}, \begin{equation}
		\left|\sum^{T-1}_{t=0}\circled{B'}^{\Theta^{(t)},x^t}\right| \le (\log(1/\delta)\cdot d)^{\Cr{weibull} d}\cdot \sqrt{T}\cdot \procr{V^{(0)},V^*}^2 \cdot O(r)\cdot \exp(O(d))
	\end{equation} with probability at least $1 - 2\delta$. By taking $T$ according to \eqref{eq:main_T_bound}, we ensure that this quantity is upper bounded by a negligible multiple of $T\cdot (\condnumber/5)\cdot\procr{V^{(0)},V^*}^2$ as desired. The proof is completed by replacing $2\delta$ in the above with $\delta$ and absorbing the resulting constant factors.
\end{proof}

\subsubsection{Proof of Lemma~\ref{lem:circledC_conc}}
\label{app:circledC_conc}

\begin{proof}
	As in the proof of Lemma~\ref{lem:circledB_conc}, for fixed $x^1,...,x^{t-1}$, the martingale difference $\circled{C'}^{\Theta^{(t)},x^t}$ is a polynomial of degree $2d$ in $x^t$, so by Lemma~\ref{lem:martingale1_polynomial} we just need to upper bound the second moments of the differences, which we do in the following lemma.

	\begin{lemma}\label{lem:Cvar}
		For any $\Theta$, $\E_x[(\circled{C'}^{\Theta,x})^2] \le \procr{V,V^*}^2\cdot\norm{\vec{c} - \vec{c}^*}^2_2\cdot \exp(O(d))$.
	\end{lemma}

	\begin{proof}
		By Cauchy-Schwarz,
		\begin{align}
			\E\left[\circled{C'}^2\right] &\le \E\left[(\polydiff(V^{\top}x)^4\right]^{1/2}\cdot \E\left[\left(x^{\top}\Pi^{\perp}_VV^*\nabla\right)^4\right]^{1/2} \\
			&= \E_{g\sim\N(0,\Id_r)}\left[\polydiff(g)^4\right]^{1/2}\cdot \E\left[\circled{A'}^2\right]^{1/2} \\
			&\le \left(3^d\cdot\norm{\vec{c} - \vec{c}^*}^2_2\right) \cdot \left(2^{2d+2}\cdot \procr{V,V^*}^2\right),
		\end{align} where the third step follows by Fact~\ref{fact:hypercontractivity} and \eqref{eq:Asquaredbound}.
	\end{proof}

	We now complete the proof of Lemma~\ref{lem:circledC_conc}. By Lemma~\ref{lem:safe}, $\procr{V^{(t)},V^*}\le \norm{V^{(t)} - V^*}_F \le 1.1\cdot\procr{V^{(0)},V^*}$ for every $0\le t\le T$ with probability at least $1 - \delta$, in which case Lemma~\ref{lem:Cvar} implies that for every $0\le t<T$, \begin{equation}\E[(\circled{C'}^{\Theta^{(t)},x^t})^2|x^1,...,x^{t-1}] \le \procr{V^{(0)},V^*}\cdot\norm{\vec{c} - \vec{c}^*}_2\cdot \exp(O(d))\end{equation} with probability at least $1 - \delta$. So by Lemma~\ref{lem:martingale1_polynomial}, \begin{equation}
		\left|\sum^{T-1}_{t=0}\circled{C'}^{\Theta^{(t)},x^t}\right| \le (\log(1/\delta)\cdot d)^{\Cr{weibull} d}\cdot \sqrt{T}\cdot \procr{V^{(0)},V^*}\cdot\norm{\vec{c} - \vec{c}^*}_2\cdot \exp(O(d))
	\end{equation} with probability at least $1 - 2\delta$. By taking $T$ satisfying the bound in the lemma statement and invoking \eqref{eq:ratio_bound}, we ensure that this quantity is upper bounded by a negligible multiple of $T\cdot (\condnumber/5)\cdot\procr{V^{(0)},V^*}^2$ as desired. As in the proof of Lemma~\ref{lem:circledA_conc}, the proof is completed by replacing $2\delta$ in the above with $\delta$ and absorbing the resulting constant factors.
\end{proof}

\subsubsection{Proof of Lemmas~\ref{lem:Evecone_conc} and \ref{lem:Evectwo_conc}}
\label{app:Evec_conc}

We will apply Lemma~\ref{lem:martingale1_polynomial} to the MDS's $\left\{\Evecone^{\Theta^{(t)},x^t} - \mu_{\Evecone}(\Theta^{(t)})\right\}$ and $\left\{\Evectwo^{\Theta^{(t)},x^t} - \mu_{\Evectwo}(\Theta^{(t)})\right\}$. As in the analysis of the MDS's for Lemmas~\ref{lem:circledB_conc} and \ref{lem:circledC_conc}, the differences in these MDS's are polynomials of degree at most $2d$, so we just need to bound the second moments of their differences. We do so in the following two lemmas.

\begin{lemma}\label{lem:varbound_Evecone}
	For any $\Theta$,
	\begin{equation}
		\E_x[(\Evecone^{\Theta,x})^2] \le O(\etav^2)\cdot O(dr^3)^{d+1}\cdot \norm{V - V^*}^2_F \cdot \procr{V,V^*}^2 \cdot \left(\norm{\vec{c} - \vec{c}^*}_2 + \norm{V^* - V}_F\right)^2
	\end{equation}
\end{lemma}

\begin{proof}
	We have that \begin{align}
		\frac{1}{4\etav^2}\E\left[\left(\Evecone^{\Theta,x}\right)^2\right] &= \E\left[(\residual^{\Theta,x})^2 \cdot \left(x^{\top}\cdot \Pi^{\perp}_VV^*\cdot\Delta\right)^2\right] \\
		&\le \E\left[(\residual^{\Theta,x})^4\right]^{1/2}\cdot \E\left[\left(x^{\top}\cdot \Pi^{\perp}_VV^*\cdot\Delta\right)^4\right]^{1/2} \\
		&= \E\left[(\residual^{\Theta,x})^4\right]^{1/2}\cdot \E[\circled{A'}^2]^{1/2} \\
		&\le O(dr^3)^{d+1}\cdot \norm{V - V^*}^2_F \cdot \procr{V,V^*}^2 \cdot \left(\norm{\vec{c} - \vec{c}^*}_2 + \norm{V^* - V}_F\right)^2,
	\end{align} where the second step follows by Cauchy-Schwarz, the third step follows by definition of $\circled{A'}$, and the fourth step follows by Lemma~\ref{lem:main_taylorfourth}.
\end{proof}

\begin{lemma}\label{lem:varbound_Evectwo}
	For any $\Theta$, if $\etav\le O(1/n)$, then
	\begin{equation}
		\E_x[(\Evectwo^{\Theta,x})^2] \le O(\etav^2)\cdot (64dr^3)^{2d+4}\cdot \norm{V - V^*}^2_F \cdot (\norm{V-V^*}_F + \norm{\vec{c} - \vec{c}^*}_2)^4
	\end{equation}
\end{lemma}

\begin{proof}
	By triangle inequality and Jensen's, $\E[(\Evectwo^{\Theta,x})^2]^{1/2} = \E[\langle \trigE,V - V^*\rangle^2]^{1/2}$ is at most \begin{equation}	
		\E\left[(\cos(\sigma\etav)-1)^2\cdot \langle V\cdot \hatnab\hatnab^{\top}, V - V^*\rangle^2\right]^{1/2} + \E\left[(\sin(\sigma\etav)-\sigma\etav)^2\cdot \langle \hath\hatnab^{\top},V - V^*\rangle^2\right]^{1/2}
	\end{equation} By Holder's and the fact that $|\cos(x) - 1|\le x^2/2$ and $|\sin(x) - x|\le x^2/\pi$ for all $x\ge 0$, we may upper bound the first term by \begin{equation}
		\E\left[(\cos(\sigma\etav)-1)^2\right]^{1/2}\cdot \max_{\hatnab}\left|\langle V\cdot \hatnab\hatnab^{\top}, V - V^*\rangle\right| \le O(\etav^2)\cdot \E[\sigma^4]^{1/2}\cdot \max_{\hatnab}\left|\langle V\cdot \hatnab\hatnab^{\top}, V - V^*\rangle\right|
	\end{equation} and the second term by \begin{equation}
		\E\left[(\sin(\sigma\etav)-\sigma\etav)^2\right]^{1/2}\cdot \max_{\hath,\hatnab}\left|\langle \hath\hatnab^{\top},V - V^*\rangle\right| \le O(\etav^2)\cdot \max_{\hath,\hatnab}\left|\langle \hath\hatnab^{\top},V - V^*\rangle\right|.
	\end{equation}
	So $\E[(\Evectwo^{\Theta,x})^2]^{1/2}$ is at most \begin{align}
		\MoveEqLeft O(\etav^2)\cdot \E[\sigma^4]^{1/2}\cdot \left(\max_{\hatnab}\left|\langle V\cdot \hatnab\hatnab^{\top}, V - V^*\rangle\right| + \max_{\hath,\hatnab}\left|\langle \hath\hatnab^{\top},V - V^*\rangle\right|\right) \\
		&\le  O(\etav^2)\cdot O(n)\cdot (64dr^3)^{d+2}\cdot (\norm{V-V^*}_F + \norm{\vec{c} - \vec{c}^*}_2)^2 \cdot \norm{V - V^*}_F \\ 
		&\le O(\etav)\cdot (64dr^3)^{d+2}\cdot (\norm{V-V^*}_F + \norm{\vec{c} - \vec{c}^*}_2)^2 \cdot \norm{V - V^*}_F,
	\end{align} where the first step follows by Lemma~\ref{lem:sigsquared}, Lemma~\ref{lem:maxterm1}, and Lemma~\ref{lem:maxterm2}, and the sixth follows by the assumption that $\etav\le O(1/n)$.
\end{proof}

We are now ready to complete the proofs of Lemma~\ref{lem:Evecone_conc} and \ref{lem:Evectwo_conc}.

\begin{proof}[Proof of Lemma~\ref{lem:Evecone_conc}]
	By Lemma~\ref{lem:safe}, $\procr{V^{(t)},V^*}\le \norm{V^{(t)} - V^*}_F \le 1.1\cdot\procr{V^{(0)},V^*}$ for every $0\le t\le T$ with probability at least $1 - \delta$, in which case Lemma~\ref{lem:varbound_Evecone} implies that for every $0\le t<T$, \begin{align}\E[(\Evecone^{\Theta^{(t)},x^t})^2|x^1,...,x^{t-1}] &\le O(\etav^2)\cdot O(dr^3)^{d+1}\cdot \procr{V^{(0)},V^*}^4 \cdot \left(\norm{\vec{c} - \vec{c}^*}_2 + \procr{V^{(0)},V^*}\right)^2 \\
	&\le O(\etav^2)\cdot O(dr^3)^{d+1}\cdot \procr{V^{(0)},V^*}^6 \end{align} with probability at least $1 - \delta$. So by Lemma~\ref{lem:martingale1_polynomial}, \begin{align}
		\MoveEqLeft \left|\sum^{T-1}_{t=0}\left(\Evecone^{\Theta^{(t)},x^t} - \E\left[\mu_{\Evecone}(\Theta^{(t)})\right]\right)\right| \\
		&\le (\log(1/\delta)\cdot d)^{\Cr{weibull} d}\cdot \sqrt{T}\cdot O(\etav)\cdot O(dr^3)^{(d+1)/2}\cdot \procr{V^{(0)},V^*}^3
	\end{align} with probability at least $1 - 2\delta$. By Lemma~\ref{lem:sumEvecones}, we conclude that \begin{equation}
		\left|\sum^{T-1}_{t=0}\Evecone^{\Theta^{(t)},x^t}\right| \le O(\sqrt{T}\cdot \etav)\cdot O(dr^3)^{(d+1)/2}\cdot \procr{V^{(0)},V^*}^3\cdot \left((\log(1/\delta)\cdot d)^{\Cr{weibull} d} + \sqrt{T}\right)
	\end{equation}

	By taking $T$ according to \eqref{eq:main_T_bound} and using the bound \eqref{eq:V0warmstart}, we ensure that this quantity is upper bounded by a negligible multiple of $T\cdot\etav\cdot (\condnumber/3)\cdot\procr{V^{(0)},V^*}^2$ as desired. As usual, the proof is completed by replacing $2\delta$ in the above with $\delta$ and absorbing the resulting constant factors.
\end{proof}

\begin{proof}[Proof of Lemma~\ref{lem:Evectwo_conc}]
	By Lemma~\ref{lem:safe}, $\norm{V^{(t)} - V^*}_F \le 1.1\cdot\procr{V^{(0)},V^*}$ for every $0\le t\le T$ with probability at least $1 - \delta$, in which case Lemma~\ref{lem:varbound_Evectwo} implies that for every $0\le t<T$, \begin{align}\E[(\Evectwo^{\Theta^{(t)},x^t})^2|x^1,...,x^{t-1}] &\le O(\etav^2)\cdot O(dr^3)^{2d+4}\cdot \procr{V^{(0)},V^*}^2 \cdot \left(\norm{\vec{c} - \vec{c}^*}_2 + \procr{V^{(0)},V^*}\right)^4 \\
	&\le O(\etav^2)\cdot O(dr^3)^{2d+4}\cdot \procr{V^{(0)},V^*}^6 \end{align} with probability at least $1 - \delta$. So by Lemma~\ref{lem:martingale1_polynomial}, \begin{align}
		\MoveEqLeft \left|\sum^{T-1}_{t=0}\left(\Evectwo^{\Theta^{(t)},x^t}\sum^{T-1}_{t=0}\E\left[\mu_{\Evectwo}(\Theta^{(t)})\right]\right)\right| \\
		&\le (\log(1/\delta)\cdot d)^{\Cr{weibull} d}\cdot \sqrt{T}\cdot O(\etav)\cdot O(dr^3)^{d+2}\cdot \procr{V^{(0)},V^*}^3
	\end{align} with probability at least $1 - 2\delta$. By \eqref{lem:sumEvectwos}, we conclude that \begin{equation}
		\left|\sum^{T-1}_{t=0}\Evectwo^{\Theta^{(t)},x^t}\right| \le O(\sqrt{T}\cdot \etav)\cdot O(dr^3)^{d+2}\cdot \procr{V^{(0)},V^*}^3 \cdot \left((\log(1/\delta)\cdot d)^{\Cr{weibull} d} + \sqrt{T}\right)
	\end{equation}

	By taking $T$ according to \eqref{eq:main_T_bound} and using the bound \eqref{eq:V0warmstart}, we ensure that this quantity is upper bounded by a negligible multiple of $T\cdot\etav\cdot (\condnumber/3)\cdot\procr{V^{(0)},V^*}^2$ as desired. As usual, the proof is completed by replacing $2\delta$ in the above with $\delta$ and absorbing the resulting constant factors.
\end{proof}

\end{document}